\documentclass[a4paper,11pt]{article}
\pdfoutput=1 % if your are submitting a pdflatex (i.e. if you have images in pdf, png or jpg format)
\usepackage{jheppub} 
\usepackage[utf8]{inputenc}
\usepackage{amsthm}
\usepackage{mathrsfs} 
\usepackage[all]{xy}
\graphicspath{{./images/}}
\theoremstyle{plain}
\newtheorem{thm}{Theorem}
\newtheorem{lem}[thm]{Lemma}

\newtheorem{cor}[thm]{Corollary}

\theoremstyle{definition}
\newtheorem{rem}[thm]{Remark}
\newtheorem{defi}[thm]{Definition}
\newtheorem{ex}[thm]{Example}

\numberwithin{thm}{section}

\newcommand{\ket}[1]{ | {#1} \rangle }

\DeclareMathOperator{\tr}{tr}

\usepackage{slashed}

\def\CC{{\cal C}}

\def\CF{{\cal F}}
\def\CG{{\cal G}}
\def\CH{{\cal H}}
\def\CI{{\cal I}}

\def\CL{{\cal L}}

\def\CT{{\cal T}}

\def\BC{{\mathbb C}}

\def\BR{{\mathbb R}}

\def\BZ{{\mathbb Z}}

\def\sC{{\mathscr C}}
\def\sD{{\mathscr D}}
\def\U{\mathrm{U}}

\def\O{\mathrm{O}}
\def\SO{\mathrm{SO}}

\def\Spin{\mathrm{Spin}}
\def\Pin{\mathrm{Pin}}

\def\beq#1\eeq{\begin{align}#1\end{align}}

\title{On the cobordism classification of symmetry protected topological phases }

\preprint{IPMU-18-0040}

\author{Kazuya Yonekura}
\affiliation{Kavli IPMU (WPI), UTIAS, 
The University of Tokyo, 
Kashiwa, Chiba 277-8583, Japan
}

\abstract{In the framework of Atiyah's axioms of topological quantum field theory with unitarity, 
we give a direct proof of the fact that symmetry protected topological (SPT) phases without Hall effects are classified by cobordism invariants.
We first show that the partition functions of those theories are cobordism invariants after a tuning of the Euler term.
Conversely, for a given cobordism invariant, we construct a unitary topological field theory
whose partition function is given by the cobordism invariant, assuming that a certain bordism group is finitely generated. 
Two theories having the same cobordism invariant partition functions are isomorphic.
} 

\begin{document}

\maketitle
%%%%%%%%%%%%%%%%%%%%%%%%%%%%%%%%%%%%%%%%%%%%%%%%%%%%%%%%%%%

%%%%%%%%%%%%%%%%%%%%%%%%%%%%%%%%%%%%%%%%%%%%%%%%%%%%%%%%%%%%%%%%%%%%%%%%%%%%%%%%%%%%%%%%%%%%%%%%%%%%%%%%%%%%%%%%%%%%%%%%%%%%%%%%%%%%%%%%%%%%%%%%%%%%%%%%%%%%%%%%%%%%%%%%%%%%%%%%%%%%%%
\section{Introduction}
%%%%%%%%%%%%%%%%%%%%%%%%%%%%%%%%%%%%%%%%%%%%%%%%%%%%%%%%%%%%%%%%%%%%%%%%%%%%%%%%%%%%%%%%%%%%%%%%%%%%%%%%%%%%%%%%%%%%%%%%%%%%%%%%%%%%%%%%%%%%%%%%%%%%%%%%%%%%%%%%%%%%%%%%%%%%%%%%%%%%%%

\subsection{The main theorem}
In this paper, we prove the following theorem whose physics motivations, the precise mathematical meaning and the conditions will be explained later in this paper:
\begin{thm}\label{thm:main}
There is a 1:1 correspondence between the following two sets:
\begin{enumerate}
\item  the set of isomorphism classes of $d$-dimensional unitary invertible topological field theories with the symmetry group $H_d$ satisfying Atiyah's axioms,
with the particular choice of the Euler term such that the sphere partition function is unity,
\item the cobordism group ${\rm Hom}(\Omega_d^{H}, \U(1))$, where $\Omega_d^H$ is the bordism group of $d$-dimensional manifolds with $H_d$-structure,
\end{enumerate}
The precise statements are given in Theorems~\ref{thm:inv}, \ref{thm:unit}, \ref{thm:construction}, \ref{thm:identification} and Remark~\ref{rem:euler}.
\end{thm}
This theorem essentially proves the conjecture in \cite{Kapustin:2014tfa,Kapustin:2014dxa}, at least in the framework of relativistic field theory.
See also \cite{Freed:2016rqq} for a very closely related theorem by Freed and Hopkins for fully extended invertible field theory.
Our theorem is about the case of non-extended topological quantum field theory (TQFT) which may be more familiar to physicists. 
Despite the slight difference in the axioms, the final classification results in \cite{Freed:2016rqq} and in this paper are essentially the same.
However, we remark that in this paper we always discuss about isomorphism classes of theories rather than deformation classes as in \cite{Freed:2016rqq}.
This makes our discussion more or less elementary and explicit.

\subsection{Symmetry protected topological phases, anomalies, and theta angles}
The physics motivations of the main theorem above are the classifications of (i) symmetry protected topological (SPT) phases in $d$-spacetime dimensions,
(ii) anomalies in $(d-1)$-spacetime dimensions, and
(iii) generalized theta angles in $d$-spacetime dimensional gauge theories and quantum gravity.
We will always assume relativistic symmetry in this paper. 
Empirically, the classification assuming relativity also 
coincides with other classifications of topological phases in condensed matter physics. 

SPT phases are gapped phases of quantum systems in $d$-spacetime dimensions which have a certain global symmetry group $H_d$
which include internal as well as spacetime symmetries.
In the language of quantum field theory (QFT), we assume that the system has a mass gap
and the Hilbert space of the ground states on any closed spatial manifold is one-dimensional. Then we would like to
classify the low energy (or long distance) limit of such systems up to some equivalence relations. 
This class includes important condensed matter systems such as topological insulators and superconductors~\cite{Hasan:2010xy, Qi:2011zya}.
By definition, the theory in the low energy limit has only a single state in the Hilbert space on any closed manifold.
However, they can have nontrivial partition functions.
If we are given a closed spacetime manifold $X$ with background field of the symmetry $H_d$ (e.g. electric-magnetic field for $\U(1)$ symmetry), 
we have the partition function $Z(X)$ on $X$ with the given background.
The values of the partition function turn out to classify these SPT phases.
The partition function must satisfy various requirements, or axioms, to be a QFT as 
we review later. 

The class of QFTs with one-dimensional Hilbert space on any closed spatial manifold is called invertible field theory~\cite{Freed:2004yc}. 
This name comes from the following fact.
If we have an invertible field theory $\CI$, there exists a theory $\CI^{-1}$
such that their product $\CI \times \CI^{-1}$ is a completely trivial theory. 
The classification of SPT phases amounts to a classification of invertible field theories up to the equivalence relation under continuous deformation.
The invertible field theories themselves are classified by isomorphism classes of theories, while the SPT phases are classified by deformation classes of theories.
Here, deformation classes mean that we identify two theories if they are related to each other by continuous change of parameters 
(e.g. continuous change of the composition of material by doping).

One of the most important characteristic properties of SPT phases is
as follows. If they are put on a spatial manifold with boundary, then there must appear a nontrivial boundary theory. Namely,
the boundary cannot be in a trivial gapped phase with a single ground state. 
Given an SPT phase, the possible boundary theories are not unique. However, they must have the same
't~Hooft anomaly of the global symmetry group $H_d$. The anomaly of $H_d$ of the boundary theories is 
completely determined by the bulk SPT phase. Conversely, it is believed that all anomalies are realized by SPT phases in the way described above. 
Therefore, SPT phases are relevant to the classification of 't~Hooft anomalies of the symmetry group $H_d$
(see e.g. \cite{Ryu:2010ah,Wen:2013oza,Kapustin:2014zva,Freed:2014iua,Wang:2014pma,Hsieh:2015xaa,Witten:2015aba,Witten:2016cio,Guo:2017xex} for a partial list of references).

In the case that the invertible field theories are realized as the low energy limit of massive fermions,
the situation is well-understood~\cite{Witten:2015aba,Witten:2016cio}. 
The partition functions are described by the $\eta$ invariant~\cite{Atiyah:1975jf} of Dirac operators coupled to the background field.
The relevant mathematics governing the $\eta$ invariant, such as the Atiyah-Patodi-Singer index theorem~\cite{Atiyah:1975jf} and the Dai-Freed theorem~\cite{Dai:1994kq} 
also have physical understanding~\cite{Yonekura:2016wuc, Fukaya:2017tsq},
analogous to the case that the Atiyah-Singer index theorem is understood by Fujikawa's method of path integral measure.
In more general cases, there are several proposals for the classification of SPT phases such as generalized group cohomology~\cite{Chen:2011pg,Gu:2012ib,Wang:2017moj}, 
cobordism group~\cite{Kapustin:2014tfa,Kapustin:2014dxa},
and more general approach based on generalized cohomology~\cite{Kitaev:2013a,Gaiotto:2017zba,Xiong:2016deb,Freed:2014eja}. They are not independent but are related to each other.
Our main theorem above, as well as the theorem in \cite{Freed:2016rqq}, prove
the conjectured classification by cobordism group \cite{Kapustin:2014tfa,Kapustin:2014dxa} in the context of relativistic field theory.

Invertible field theories are also relevant for the classification of theta angles in gauge theories and quantum gravity theories (e.g. the worldsheet theories of superstring theories)
in which a subgroup $G \subset H_d$ is ``gauged", meaning that the field associated to it 
is dynamical instead of background~\cite{Freed:2004yc}. In general, if we have a theory $\CT$ with a symmetry $G \subset H_d$,
then we may gauge the symmetry $G$ by coupling it to dynamical gauge field and/or gravity. Then we get a gauge theory with the gauge group $G$, which we may denote as $\CT/G$.
Suppose that we have a generic theory $\CT$ and an invertible field theory $\CI$ both of which have the symmetry $G$.
Then, we can consider another theory $\CT \times \CI$, and gauge
the group $G$ to get another gauge theory $(\CT \times \CI)/G$. In this gauge theory,
the invertible theory $\CI$ plays the role of a theta angle for the gauge group $G$.
For example, in $d=2$ spacetime dimensions with a symmetry $\U(1)$, we can consider a theory $\CI_\theta$ whose
partition function is given as $Z(X)=\exp( i \theta \int_X c_1)$, where $c_1= \frac{i F}{2\pi}$ is the first Chern class of the background field
strength $F$ for $\U(1)$,
and $\theta \in \BR/ 2\pi \BZ$ is a parameter. 
After making the $\U(1)$ gauge field dynamical, it is obvious that $\CI_\theta$ gives the theta term
for the $\U(1)$ gauge theory. See also \cite{Freed:2017rlk} for a recent discussion of more sophisticated examples.
For the classification of the theta terms, we need isomorphism classes of invertible field theories rather than deformation classes of theories.

\subsection{The principles of locality and unitarity in invertible field theory}
Because of the above motivations, it is important to classify possible invertible field theories with global symmetry group $H_d$.
How can we classify them? For this purpose, we use the most fundamental principles of QFT: locality and unitarity. 
Let us sketch them in the particular case of invertible field theories. The details will be reviewed in Sec.~\ref{sec:review}.
If the following discussions look abstract, we refer the reader to e.g. \cite{Witten:2015aba,Witten:2016cio,Yonekura:2016wuc} for the concrete case of free massive fermions.

\paragraph{Locality.}
The partition function of a theory on a closed manifold $X$ coupled to background field can be regarded as the effective action of 
the background field. We abbreviate the whole information of the manifold and the background field by just $X$.
Then the effective action $S_{\rm eff}(X)$ is given as $Z(X)=e^{-S_{\rm eff}(X)}$. In a system with a mass gap $\Delta$, the correlation length of the system is of the order of $\Delta^{-1}$.
If we only consider the system in length scales which are much larger than $\Delta^{-1}$, the correlation length is negligibly small and
the effective action is a ``local functional" of the background field. Roughly speaking, this means that 
the effective action is given by the integral of a local effective Lagrangian $\CL_{\rm eff}$ as 
\beq
``~S_{\rm eff}(X) = \int_X \CL_{\rm eff}~". \label{eq:integralL}
\eeq
For example, in the case of the theta angle in two-dimensional $\U(1)$ field discussed above, we have $\CL_{\rm eff} =  \frac{\theta}{2\pi} F$.

Suppose that the manifold $X$ is decomposed as $X = X_1 \cup X_2$
with $\partial X_1 = \overline{\partial X_2} = Y$ as in Figure~\ref{fig:div}.
The bar on $\overline{\partial X_2} $ is a generalization of orientation flip which will be explained in Sec.~\ref{sec:review}.
Then roughly speaking, we have 
\beq
``~S_{\rm eff}(X) = S_{\rm eff}(X_1) + S_{\rm eff}(X_2)~" \label{eq:sumS}
\eeq
because $S_{\rm eff}$ is the integral of $\CL_{\rm eff}$.
Hence we get
\beq
Z(X)=Z(X_1) Z(X_2). \label{eq:productZ}
\eeq
This is the rough statement of locality in the context of theories with a large mass gap. 
\begin{figure}
\centering
\includegraphics[width=.5\textwidth]{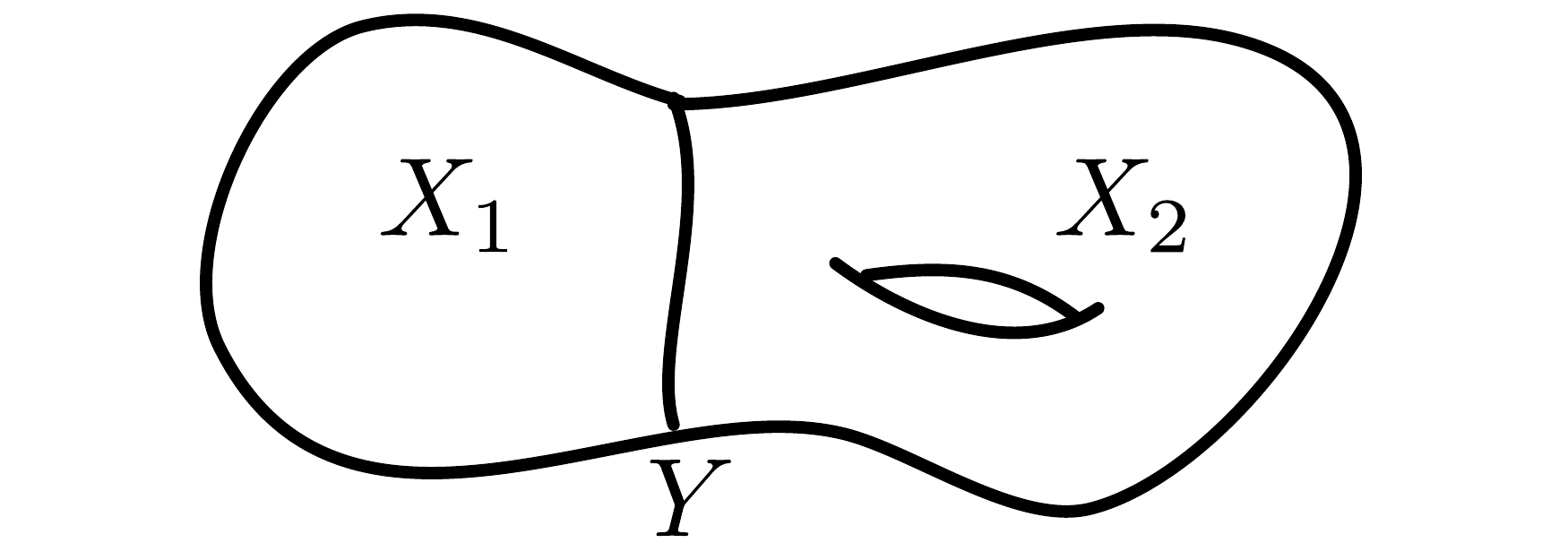}
\caption{ A manifold $X$ which consists of $X_1$ and $X_2$ glued at their common boundary $Y$.  \label{fig:div}}
\end{figure}

However, the above equations \eqref{eq:integralL} and \eqref{eq:sumS} are not precise, while \eqref{eq:productZ} needs a careful reinterpretation. 
The reason is as follows. Even if the correlation length is infinitesimally small, the two manifolds $X_1$ and $X_2$ are touching each other at their common boundary $\partial X_1 = \overline{\partial X_2} = Y$.
Therefore, there is a room for some nontrivial effect from the boundary. 

One way to think about the effect of the boundary is as follows. In Euclidean path integral, we are free to say which direction is the ``time" direction and which are the ``space" directions.
Let us regard $Y$ as ``space", and the direction orthogonal to it as ``time". Then we have the physical Hilbert space $\CH(Y)$ associated to $Y$.
The path integral over $X_1$ naturally gives a state vector on $Y$ (which is the defining property of the path integral), 
and hence we can think of $Z(X_1)$ as taking values in $\CH(Y)$,
\beq
Z(X_1) \in \CH(Y).
\eeq
By the assumption of invertible field theory, the Hilbert space $\CH(Y)$ is one-dimensional.
The $Z(X_2)$ takes values in the dual vector space $\CH(\overline{Y}) \simeq \overline{\CH(Y)}$ and their product $Z(X)=Z(X_1) Z(X_2)$ takes values in $\BC$.
Thus, the value of the partition function on a closed manifold gives a number in $\BC$, but the partition function on a manifold with boundary 
gives an element of the Hilbert space $\CH(Y)$ associated to the boundary $Y$.
Under this interpretation, the equation \eqref{eq:productZ} makes sense.
This is the condition of locality that the partition function must satisfy. 

\paragraph{Unitarity.}
Another fundamental principle is unitarity. 
Let us look at the structure of the effective action $S_{\rm eff}$ in a unitary theory. The unitarity is most straightforwardly
understood when the metric of the manifold $X_{\rm Lorentz}$ has Lorentzian signature rather than Euclidean signature.
Then, the unitarity means that $Z(X_{\rm Lorentz})=e^{i S_{\rm eff}(X_{\rm Lorentz}) }$ with real $S_{\rm eff}(X_{\rm Lorentz}) \in \BR$. This is due to the fact that
the sum of probability in quantum mechanics must be unity, so any time evolution must be given by a unitary matrix,
and invertible field theory has only one state in the Hilbert space, so the $Z(X_{\rm Lorentz})$ must be a pure phase with the unit absolute value $|Z(X_{\rm Lorentz})|=1$.

Now we want to Wick-rotate the above statement to a manifold $X$ with Euclidean signature. 
For simplicity, here we assume that the theory depends just on the orientation of the manifold, whose generalization will be discussed later in Sec.~\ref{sec:review}.
We may divide the action into two parts $S_{\rm eff}(X_{\rm Lorentz}) = S^{\rm even}_{\rm eff}(X_{\rm Lorentz}) + S^{\rm odd}_{\rm eff}(X_{\rm Lorentz}) $,
where $S^{\rm even}_{\rm eff}(X_{\rm Lorentz})$ is the part which is even under the orientation flip, and $S^{\rm odd}_{\rm eff}(X_{\rm Lorentz})$
is the part which is odd under orientation flip. Roughly speaking, the odd part contains the totally anti-symmetric tensor $\epsilon^{\mu_1 \cdots \mu_d}$.
For example, the theta term for the two-dimensional $\U(1)$ field is $\frac{\theta}{4\pi} \int_X d^2x (\epsilon^{\mu\nu}F_{\mu\nu})$ if we use the explicit component notation.
This totally anti-symmetric tensor $\epsilon^{\mu_1 \cdots \mu_d}$ produces an imaginary factor $i =\sqrt{-1}$ when we Wick-rotate the metric
from the Lorentzian signature to Euclidean signature. Therefore, after the Wick rotation $X_{\rm Loretnz} \to X$, we get
\beq
Z(X) = \exp \left( - S^{\rm even}_{\rm eff}(X) + i S^{\rm odd}_{\rm eff}(X) \right),
\eeq
where both $S^{\rm even}_{\rm eff}(X) $ and $S^{\rm odd}_{\rm eff}(X)$ take values in the real numbers $\BR$.

Let $\overline{X}$ be the orientation flip of the manifold $X$. Then, from the above structure, we clearly get
\beq
Z(\overline{X}) = \overline{Z(X)}
\eeq
where $ \overline{Z(X)}$ is the complex conjugate of $Z(X)$. This is the requirement of unitarity.
We require it also when the manifold has a boundary $\partial X = Y$. In that case, we first require that
the Hilbert space $\CH(\overline{Y})$ on the orientation-flipped manifold $\overline{Y}$ is given by the complex conjugate Hilbert space $\overline{\CH(Y)}$.
Then, $Z(X)$ takes values in $\CH(Y)$, and $Z(\overline{X})$ takes values in $\CH(\overline{Y}) \cong \overline{\CH(Y)}$ with the value $Z(\overline{X})=\overline{Z( X )}$,
which is the complex conjugate vector of $Z(X)$. Furthermore, the unitarity requires that every Hilbert space $\CH(Y)$ has a positive definite hermitian inner product, and hence
we can identify the complex conjugate vector space $\overline{\CH(Y)}$ and the dual vector space of $\CH(Y)$. 
Then $Z(X) \overline{Z(X)} \geq 0$ can be naturally regarded as a non-negative number in the unitary theory.

Locality and unitarity are believed to be necessary conditions for relativistic quantum systems.
It is not clear whether they are also sufficient or not, especially because the locality in the above sense might be weaker than physically expected, 
and it may be weaker than the fully extended version of locality~\cite{Baez:1995xq,Lurie:2009keu} (see also \cite{Schommer-Pries:2017sdd} for invertible field theory) 
used in \cite{Freed:2016rqq}.
However, locality and unitarity in the above sense may be sufficient for invertible field theories. 
They are already very powerful constraints on physical systems,
and the classifications by imposing just locality and unitarity agree with classifications obtained by totally different methods.
One of the nontrivial examples is the classification of time-reversal invariant topological superconductors in $3+1$-dimensions.
There the classification by $\BZ_{16}$ was found by studies of solvable models 
\cite{Fidkowski:2013jua,Wang:2014lca,Metlitski:2014xqa,Morimoto:2015lua,Tachikawa:2016xvs, Witten:2016cio},
and the same classification was obtained in \cite{Kapustin:2014dxa,Witten:2015aba} which is essentially by the requirement of locality and unitarity.
Therefore, in this paper we classify invertible field theories under these principles.
More precise axioms are explained in Sec.~\ref{sec:review}.

The above discussions are general. However,
in the rest of the paper, we work entirely in the framework of topological QFT.
We need to make a remark about what we mean by ``topological". 
By ``topological", we mean that partition functions on closed manifolds are invariant under continuous deformation of background field.
Therefore, it may be more properly rephrased as ``homotopy invariance under the change of background field".
This excludes some important invertible field theories. For example, for $\U(1) $ symmetry in $d=2+1$ spacetime dimensions,
we have the invertible field theory whose effective action is given by the Chern-Simons invariant of the background field $A$ which is roughly given as $S_{\rm eff}(X)=\frac{i k}{4\pi} \int_X AdA$ for $k \in \BZ$.
This system describes integer quantum Hall effects (see \cite{Witten:2015aoa} for a review). The effective action
depends on the precise field configurations of $A$ and is not invariant under the continuous change of $A$.
For more discussions and a conjecture about the classification of theories including those ``non-topological" cases, see \cite{Freed:2016rqq}.
Throughout the paper, we assume the absence of such Hall conductivity and generalizations of it,
as well as other less interesting non-topological terms in $S_{\rm eff}$.

\subsection{Organization of the paper}
In the rest of the paper, we will prove Theorem~\ref{thm:main} whose precise statements are given in
Theorems~\ref{thm:inv}, \ref{thm:unit}, \ref{thm:construction}, \ref{thm:identification} and Remark~\ref{rem:euler}.

In section~\ref{sec:review}, we review the Atiyah's axioms of TQFT, following \cite{Freed:2016rqq}.
We describe the precise version of locality and unitarity which are sketched above.
In section~\ref{sec:cobinv}, we show that partition functions of invertible field theories are invariant under
bordisms of manifolds. More explicitly, partition functions are given by cobordism invariants ${\rm Hom}(\Omega_d^{H}, \U(1))$,
up to the caveat that we need to tune the term in the effective action which is proportional to the Euler density.
The Euler term is rather trivial and can be easily factored out.
In section~\ref{sec:explicit}, for each element of ${\rm Hom}(\Omega_d^{H}, \U(1))$,
we give an explicit construction of a TQFT whose partition function is given by that element.
Moreover, we show that invertible TQFTs are completely characterized by their partition functions.
Namely, two theories with the same partition function are isomorphic in the sense which is made precise there.
Therefore, we conclude that unitary invertible TQFTs are classified by ${\rm Hom}(\Omega_d^{H}, \U(1))$.
In appendix~\ref{sec:app} we summarize some technical definitions in category theory.

%%%%%%%%%%%%%%%%%%%%%%%%%%%%%%%%%%%%%%%%%%%%%%%%%%%%%%%%%%%%%%%%%%%%%%%%%%%%%%%%%%%%%%%%%%%%%%%%%%%%%%%%%%%%%%%%%%%%%%%%%%%%%%%%%%%%%%%%%%%%%%%%%%%%%%%%%%%%%%%%%%%%%%%%%%%%%%%%%%%%%%
%%%%%%%%%%%%%%%%%%%%%%%%%%%%%%%%%%%%%%%%%%%%%%%%%%%%%%%%%%%
\section{Manifold with structure and axioms of topological field theory}\label{sec:review}
%%%%%%%%%%%%%%%%%%%%%%%%%%%%%%%%%%%%%%%%%%%%%%%%%%%%%%%%%%%
%%%%%%%%%%%%%%%%%%%%%%%%%%%%%%%%%%%%%%%%%%%%%%%%%%%%%%%%%%%%%%%%%%%%%%%%%%%%%%%%%%%%%%%%%%%%%%%%%%%%%%%%%%%%%%%%%%%%%%%%%%%%%%%%%%%%%%%%%%%%%%%%%%%%%%%%%%%%%%%%%%%%%%%%%%%%%%%%%%%%%%
QFT requires spacetime dimension $d$ and a symmetry group $H_d$ as the data.
The background field for the symmetry group $H_d$ is incorporated as principal $H_d$-bundles.
We review the axioms of TQFT with the symmetry group $H_d$. 
See Sec.~2,3,4 of \cite{Freed:2016rqq} for more details. 
However, we give more explicit discussions on spin structure, so the discussions in Sec.~\ref{sec:subtle} and related discussions in later subsections may be new.

It may be possible to extend our discussions to more general cases such as higher form symmetries~\cite{Gaiotto:2014kfa}, sigma models~\cite{Freed:2017rlk,Thorngren:2017vzn},
2-groups~\cite{Kapustin:2013uxa,Tachikawa:2017gyf,Cordova:2018cvg,Benini:2018reh}, duality groups~\cite{Seiberg:2018ntt}, and so on. But we do not discuss them in this paper.

\subsection{$H_d$-manifold} \label{sec:Hd}
The group $H_d$ contains both internal as well as the Euclidean signature version of Lorentz symmetries, and it must satisfy some properties
which are described in \cite{Freed:2016rqq}. We review them to the extent needed in this paper.
We work entirely in the Euclidean signature of the Lorentz group.

$H_d$ is a compact Lie group with
a homomorphism 
\beq
\rho_d: H_d \to \O(d). \label{eq:HtoO}
\eeq
This is a map which ``forgets the internal symmetry group", and $\O(d)$ is regarded as the Euclidean signature version of Lorentz symmetry group of spacetime. 
The image $\rho_d(H_d)$ is either $\SO(d)$, $\O(d)$ or the trivial group, but we do not consider the case of the trivial group following \cite{Freed:2016rqq}.
(The case of the trivial $\rho_d(H_d)$ would give framings of the tangent bundle of manifolds.)
If $\rho_d(H_d) = \O(d)$, the theory has a time reversal symmetry.
The inverse image of $\SO(d)$ is of the form
\beq
\rho^{-1}_d(\SO(d)) \cong (\Spin(d) \times K)/\langle (-1, k_0 ) \rangle \label{eq:orientableH}
\eeq
where $K$ is a compact group which is the kernel of $\rho_d$ (i.e., the internal symmetry group),
$k_0$ is a central element of $K$ of order 2 (i.e. $(k_0)^2=1$), and $\langle (-1, k_0) \rangle $ is the $\BZ_2$ subgroup of $\Spin(d) \times K$ 
generated by $(-1, k_0)$ where $-1 \in \Spin(d)$ is the center of the spin group which maps to the identity in $\SO(d)$.
For example, if $H_d=\Spin(d)$, then $K=\BZ_2$ because $(\Spin(d) \times \BZ_2)/ \BZ_2 \cong \Spin(d)$. 
We call the central element  $(-1,1) \sim (1, k_0)$ as the fermion parity 
and denote it as $(-1)^F$.
The equation \eqref{eq:orientableH} is derived under the condition $d \geq 3$ in \cite{Freed:2016rqq}, but for simplicity 
we only consider the cases that $H_d$ has this property for any $d \geq 1$.

Typical (though not general) examples are of the form
\beq
H_d = (L_d \ltimes K)/\langle (-1, k_0 ) \rangle , \label{eq:tytpicalH}
\eeq
where $L_d$ is a Euclidean signature of the Lorentz group such as $\Spin(d) $ and $ \Pin^\pm(d)$.
The semidirect product is given as follows. Take an automorphism $\alpha: K \to K$ such that $\alpha^2=1$, which may be trivial ($\alpha=1$).
Then we define the product of $(s_1, k_1) \in L_d \ltimes K$ and $(s_2, k_2) \in L_d \ltimes K$
as $(s_1,k_1)(s_2,k_2)=(s_1 s_2, k_1 \alpha^{n_1}(k_2)) $ where $n_1 = 0$ or $1$ mod 2 depending
on whether $\rho_d(s_1,k_1)$ is in the connected component of the identity in $\O(d)$ or in the other component. 
For example, if there is no symmetry other than the Lorentz group, we have $H_d=\SO(d) =\Spin(d)/\BZ_2$ for the bosonic case and $H_d=\Spin(d)$ for the fermionic case.
With time-reversal symmetry, we have $H_d=\O(d) = \Pin^{\pm}(d)/\BZ_2$ for the bosonic case and $H_d=\Pin^{\pm}(d)$ for the fermionic case.
A bosonic system protected by $\U(1)$ symmetry is given by $H_d = \SO(d) \times \U(1)$. A $d=3+1$ dimensional fermionic topological insulator 
is protected by $\Pin^{+}(4) \ltimes \U(1)$ with the nontrivial automorphism $\alpha$ by complex conjugation.

Now we define $H_d$-manifold. 
Recall that a manifold $X$ has the frame bundle $FX$ associated to the tangent bundle $TX$, whose fiber $F_x X$ at $x \in X$ consists of
ordered bases $(e_1, \cdots, e_d)$ of the tangent space $T_x X$.
In the following, we abuse the notation $\pi$ to represent the projection from any fiber bundle to the base manifold $X$.
\begin{defi}
An $H_d$-structure on a $d$-dimensional manifold $X$ is a pair $(P, \varphi)$ where
$P$ is a principal $H_d$-bundle over $X$, and $\varphi$ is a bundle map $\varphi: P \to FX$ such that $\pi \circ \varphi = \pi$
and $\varphi(p \cdot h) = \varphi(p) \cdot \rho_d(h)$ for $p \in P$ and $h \in H_d$. An $H_d$-manifold is a manifold with an $H_d$-structure equipped.
\end{defi}
The image $\varphi(P) \subset FX$ has the structure of a $\rho_h(H_d)$ principle bundle which is either $\SO(d)$ or $\O(d)$ (or trivial if $\rho_d$ is trivial), 
and hence it automatically gives a Riemann metric to the tangent bundle $TX$. (If $\rho_d$ is trivial, then this would give a framing of $TX$.)
We denote this image as
\beq
F_OX =\varphi(P)
\eeq
which has the $\SO(d)$ or $\O(d)$ principal bundle structure.
If $\rho_d(H_d) = \SO(d)$, it also gives an orientation to $X$.
If $H_d =\Spin(d)$, then the $H_d$-structure is a spin structure of the manifold.
In the following, an $H_d$-manifold is often denoted by just $X$ by omitting to write $(P, \varphi)$ explicitly unless necessary.

Given an $H_d$-structure on $X$, we can define an $H_{d+1}$ structure on the bundle $\underline{\BR} \oplus TX$ in a canonical way,
where $\underline{\BR} $ is the trivial line bundle on $X$. 
First we need some preparation. 

We embed $\O(d)$ into $\O(d+1)$ as
\beq
j_d: \O(d) \ni A \mapsto \left( \begin{array}{cc}
1 & 0 \\
0 & A 
\end{array}
\right) \in \O(d+1).
\eeq 
From $H_d$, there is a way to construct $H_{n}$ for other $n \neq d$, % which is unique up to isomorphism, 
such that there exists an embedding $i_n : H_n \to H_{n+1}$
with the commutative diagram
\beq
\xymatrix{
\cdots \ar[r] &H_{d-1} \ar[d]^{\rho_{d-1} }  \ar[r]^{i_{d-1} } & H_d \ar[d]^{\rho_d}  \ar[r]^{i_d} & H_{d+1} \ar[d]^{\rho_{d+1}} \ar[r] & \cdots \\
\cdots \ar[r] &\O(d-1) \ar[r]_{j_{d-1}} & \O(d) \ar[r]_{j_d} & \O(d+1) \ar[r] & \cdots
} \label{eq:Hcomm}
\eeq
where each square is a pullback diagram.
For the examples of the form \eqref{eq:tytpicalH}, we may explicitly take $H_{n} =( L_{n} \ltimes K)/\langle (-1, k_0 ) \rangle$ for any $n$.
See Sec.~2 of \cite{Freed:2016rqq}, for more general construction. 
\begin{rem}
In the present paper, we regard \eqref{eq:HtoO}, \eqref{eq:orientableH} and \eqref{eq:Hcomm} as the axioms characterizing
$H_d$. 
\end{rem}

Now, let $(P_d, \varphi_d)$ be an $H_d$-structure on $X$. 
Define a principal $H_{d+1}$-bundle by 
$P_{d+1}=P_d \times_{i_d} H_{d+1}$. Here $\times_{i_d}$ means that for pairs $(p_d, h_{d+1}) \in P_d \times H_{d+1}$ we impose the equivalence relation 
$(p_d h_d, h_{d+1})  \sim (p_d, i_{d}(h_d)h_{d+1}) $ for $h_d \in H_d$.
Also, define a map 
\beq
\varphi_{d+1} : P_{d+1}= P_d \times_{i_d} H_{d+1} \to F(\underline{\BR} \oplus TX)
\eeq
as follows, where $F (\underline{\BR} \oplus TX)$ is the frame bundle associated to $\underline{\BR} \oplus TX$. 
(More generally, in the following, we denote the frame bundle associated to 
a vector bundle $V$ as $F V$. The orthonormal frame bundle is denoted as $F_O V$.)
For elements of the form $(p_d, 1) \in P_d \times_{i_d} H_{d+1}$ we set 
$\varphi_{d+1}(p_d, 1)=(e_0, \varphi_d(p_d))$ where $e_0 $ is the unit vector of $\underline{\BR}$.
For more general elements $(p_d,h_{d+1}) \in P_d \times_{i_d} H_{d+1}$ we define 
\beq
\varphi_{d+1}(p_d, h_{d+1})=(e_0, \varphi_d(p_d)) \cdot \rho_{d+1}(h_{d+1}). 
\eeq
One can check that this definition is well-defined (i.e. the results for $(p_d h_d, h_{d+1})$ and $(p_d, i_d(h_d) h_{d+1})$ are the same)
by using the commutativity of \eqref{eq:Hcomm} as 
\beq
&(e_0, \varphi_d(p_d h_d)) \cdot \rho_{d+1}(h_{d+1}) = (e_0, \varphi_d(p_d) \rho_d(h_d)   ) \cdot \rho_{d+1}(h_{d+1})  \nonumber \\
=&  (e_0, \varphi_d(p_d)  )  j_n (\rho_d(h_d))  \cdot \rho_{d+1}(h_{d+1}) = (e_0, \varphi_d(p_d)  )  \rho_{d+1} (i_d(h_d))  \cdot \rho_{d+1}(h_{d+1})  \nonumber \\
=&(e_0, \varphi_d(p_d)  )   \rho_{d+1}(i_d(h_d)h_{d+1}) .  \label{eq:dcheck}
\eeq
This defines the $H_{d+1}$-structure $(P_{d+1}, \varphi_{d+1})$ on $\underline{\BR} \oplus TX$.

In the same way, given an $H_d$-manifold X, we can canonically define an $H_{d+1}$-structure to $I_{a,b} \times X$ where $I_{a,b}=[a,b] \subset \BR$.
We just identify $e_0$ in the above construction with $\partial / \partial t$, where $t \in I_{a,b}$ is the standard coordinate of $\BR$.
Notice that this $H_{d+1}$-structure induces the metric $ds^2_{d+1}=dt^2+ds^2_d$ on the manifold $I_{a,b} \times X$.
We call this $H_{d+1}$-structure as the product $H_{d+1}$-structure on $I_{a,b} \times X$ induced from $X$.

Let $X$ be an $H_d$-manifold whose boundary contains a connected component $Y$. 
On $Y$, the tangent bundle of $X$ is canonically split as $TX|_Y \cong N \oplus TY$, where $N$
is the normal bundle to $Y$. 
There is the choice of whether we take the outward or inward normal vector as the basis of $N$, which give two different trivializations $N \cong \underline{\BR}$
of the normal bundle.
After the choice, we get an $H_{d}$-structure on $\underline{\BR} \oplus TY$.

Once we have the $H_d$-structure on $\underline{\BR} \oplus TY$, 
we can introduce an $H_{d-1}$-structure on $Y$ as follows for $d>1$.
There is a sub-bundle $F_O Y$ of $F_O (\underline{\BR} \oplus TY)$
which consists of elements of the form $(e_0, * )$ where $e_0$ is the unit vector of $\underline{\BR}$,
and the $*$ represents frames for $TY$. 
This $F_O Y$ has the structure of 
a principal $\SO(d-1) $ or $\O(d-1)$ bundle, depending on whether $\rho_d(H_d)=\SO(d)$ or $\O(d)$. 
Then, from the $H_d$-structure $(P_d, \varphi_d)$ on $\underline{\BR} \oplus TY$, 
we define the $H_{d-1}$-structure on $Y$ as 
\beq
P_{d-1} &= \varphi_d^{-1}(e_0, *) \\
\varphi_{d-1} &: P_{d-1} \xrightarrow{\varphi_d} F_O Y \to F Y
\eeq
where we map $(e_0,*)$ to $*$.
This defines the $H_{d-1}$-structure $(P_{d-1}, \varphi_{d-1})$ on $Y$. 
Conversely, an $H_{d-1}$-structure on $Y$ gives the $H_d$-structure on $\underline{\BR} \oplus TY$ by the process described above.
Therefore, an $H_d$-structure on $\underline{\BR} \oplus TY$
and an $H_{d-1}$-structure on $TY$ can be canonically identified for $d>1$.
So we just call them as an $H_{d-1}$-structure on $Y$ even for $d=1$, even though it is better to have in mind
the presence of $\underline{\BR}$ in the case $d=1$.

Combining the previous two paragraphs, we can define an $H_{d-1}$-structure on a component $Y$ of the boundary of $X$
once we specify the outward or inward normal vector.

Given an $H_d$-structure on $X$, we are now going to explain the opposite $H_d$-structure. 
We only explain a definition which is equivalent to the definition in \cite{Freed:2016rqq} up to a canonical isomorphism. 
See Sec.~4 of \cite{Freed:2016rqq} for a more precise definition and canonical isomorphism.

First, we define the $H_{d+1} $-structure on $\underline{\BR} \oplus TX$ in the way described above. 
We can also go back to the $H_d$-structure on $X$ induced from the $H_{d+1}$-structure on $\underline{\BR} \oplus TX $ in the way described above.
This $H_d$-structure is canonically isomorphic to the original $H_d$-structure on $X$. 
Now, instead of considering the sub-bundle of the form $(e_0, * )$ as discussed above,
we can consider the sub-bundle of the form $(-e_0, * )$ by flipping the direction of $e_0$.
Then, taking the inverse image of $(-e_0, * )$ and following the same procedure as above, we get a certain $H_{d}$-structure on $ X$.
We denote the $H_{d}$-manifold with this new $H_d$-structure as $\overline{X}$. 
(In this paper we never use a closure of a topological space. The line over an $H_d$-manifold is always used for the opposite $H_d$-structure.)
\begin{defi}
The opposite $H_d$-structure $( \overline{P_d},\overline{\varphi_d})$ of an $H_d$-manifold $(X, P_d, \varphi_d)$
is defined as follows. Let $(P_{d+1}, \varphi_{d+1})$ be the $H_{d+1}$-structure on $\underline{\BR} \oplus TX$ with $P_{d+1}=P_d \times_{i_d} H_{d+1}$
and $P_d \cong \{ \varphi_{d+1}^{-1}(e_0, *);~ * \in F X \}$.
Then 
\beq
\overline{P_d}=\{ \varphi_{d+1}^{-1}(-e_0,*);~* \in F X   \}
\eeq
and, by regarding $\overline{P_d}$ as a subset of $P_{d+1}=P_d \times_{i_d} H_{d+1}$,
\beq
\overline{\varphi_d}: \overline{P_d} \ni (p_d, h_{d+1}) \mapsto (e_0, \varphi_d(p_d)) \rho_{d+1}(h_{d+1})=(-e_0, *) \mapsto * \in F X.
\eeq
The manifold with the opposite $H_d$-structure is abbreviated as $\overline{X}$.
\end{defi}

The $H_d$-structure $\overline{\overline{X}}$ obtained by taking the opposite twice can be identified with the 
original $H_d$-structure $X$ in the following way.
First, we take the bundle $P_{d+1} = P_d \times_{i_d} H_{d+1}$ and the map $\varphi_{d+1}: P_{d+1} \to F(\underline{\BR} \oplus TX)$ as defined above.
Then take $\overline{P_d} = \varphi^{-1}_{d+1} ( -e_0, *)$. Repeating this procedure again,
we take the bundle $\overline{P_{d+1}} = \overline{P_d} \times_{i_d} H_{d+1}$ and the map 
$\overline{\varphi_{d+1}}: P_{d+1} \to F(\underline{\BR}' \oplus TX)$ (with different $\underline{\BR}'$ from $\underline{\BR}$ which appeared above).
Then take $\overline{\overline{P_d}} = \overline{\varphi_{d+1}}^{-1} ( -e'_0, *)$, where $e'_0$ is the unit vector of $\underline{\BR}'$.
This bundle $\overline{\overline{P_d}}$ can be regarded as a subset of $  P_d \times_{i_d} H_{d+1} \times_{i_d} H_{d+1}$.
Elements of this bundle are of the form 
\beq
(p, h, h') \in \overline{\overline{P_d}} \subset  P_d \times_{i_d} H_{d+1} \times_{i_d} H_{d+1}
\eeq
such that $\rho_{d+1}(h)$ and $\rho_{d+1}(h')$ act as $(-1)$ on the first component of $d+1$ dimensional vectors.
We denote that condition as 
\beq
\rho_{d+1}(h),~\rho_{d+1}(h') \in (-1) \oplus \O(d) 
\eeq
with the obvious notation.

Now let us consider the product $hh'$. The projection $\rho_{d+1}(hh')$ acts trivially on the first component of $d+1$ dimensional vectors, and hence 
by the pullback diagram \eqref{eq:Hcomm}, 
it can be represented as $hh'=i_d(h'')$ for $h'' \in H_d$. 
Now we can identify $\overline{\overline{P_d}}$ and $P_d$ by the map 
\beq
\xi: (p, h, h' ) \mapsto p h''(-1)^F. \label{eq:barbarmap}
\eeq 
One can check that
this identification is consistent with the maps $\varphi_d: P_d \to FX $ and $\overline{\overline{\varphi_d}}: \overline{\overline{P_d}} \to FX$,
where $\overline{\overline{\varphi_d}}$ is constructed in the appropriate way.
In the above identification, we have included the factor $(-1)^F$. There are several motivations for including it, which will be explained later in this paper.

Before closing this subsection, let us also define isomorphisms between two $H_d$-manifolds.
\begin{defi}
An isomorphism between two $H_d$-manifolds $(X, P, \varphi)$ and $(X',P', \varphi')$ is a bundle map $\Phi : P \to P'$
such that it induces a diffeomorphism $\Psi : X \to X'$ of the base manifolds,
and the following diagram commutes:
\beq
\xymatrix{
P \ar[r]^{\Phi} \ar[d]_{\varphi}& P ' \ar[d]^{\varphi'} \\
FX \ar[r]_{\Psi_*} & FX'
}
\eeq
\end{defi}
\begin{ex}\label{ex:time1}
To give a physical motivation for including the factor $(-1)^F$ in \eqref{eq:barbarmap}, we discuss the example of time-reversal ${\mathsf T}$.
(The following discussion can also be done for $\mathsf{CPT}$ after some modification.)
To make the physical meaning clear, here we consider an $H_{d-1}$-manifold $Y$ regarded as a spatial manifold.
Let us consider the case $H_{d} = \Pin^{\pm}(d)$ as an example. 
We pick up an element $\gamma_0 \in \Pin^{\pm}(d) $ such that $\rho_{d}(\gamma_0)=(-1) \oplus (1_{d-1})$.
If $Y$ is orientable, there is an isomorphism  ${\mathsf T}_Y$ (which will be ``time-reversal", as explained later in Example~\ref{ex:time2}) from $Y$ to $ \overline{Y}$
defined as follows. Let $s_\alpha $ be local sections on open cover $\{ U_\alpha \}$ of $Y$, and let $p=s_\alpha p_\alpha$ where 
$p_\alpha$ is the local coordinate of the fiber of $P_d$ on $U_\alpha$. 
Then we define the isomorphism ${\mathsf T}_Y $ from the $H_{d-1}$-manifold $ Y $ to $ \overline{Y}$ as
\beq
{\mathsf T}_Y: P_{d-1} \ni p= s_\alpha p_\alpha \to ( s_\alpha   , \gamma_0 i_{d-1}(p_\alpha)) \in \overline{P_{d-1}} \subset P_{d-1} \times_{i_{d-1}} H_d.
\eeq 
This map ${\mathsf T}_Y$ is well defined on orientable manifolds because $\gamma_0$ commutes with $\Spin(d-1)$
and hence $\gamma_0$ commutes with the transition functions of the bundle $P_{d-1}$. (We assume that $s_\alpha$ is taken to give some orientation to $Y$.)
Also, ${\mathsf T}_Y(p h) = {\mathsf T}_Y(p )h$, so this is an isomorphism of the principal $H_{d-1}$-bundles.
Applying ${\mathsf T}$ twice, we get ${\mathsf T}^2 :={\mathsf T}_{\overline{Y}}{\mathsf T}_Y: Y \to \overline{\overline{Y}}$ as $s_\alpha p_\alpha \to (s_\alpha, \gamma_0, \gamma_0 i_{d-1}(p_\alpha))$.
By the isomorphism $ \overline{\overline{Y}} \cong Y$ introduced in \eqref{eq:barbarmap}, we get 
\beq
{\mathsf T}^2=(\gamma_0)^2 (-1)^F
\eeq 
which is $(-1)^F$ for $\Pin^+$ and $+1$ for $\Pin^{-}$.
This will give the corresponding relation of the time-reversal operator acting on the Hilbert spaces (see Example~\ref{ex:time2} for more details).
This is the standard relation in physics: see e.g. \cite{Kapustin:2014dxa,Witten:2015aba}.
\end{ex}

\subsection{Bordism}

Now we define bordism. Let $Y_0$ and $Y_1$ be $H_{d-1}$-manifolds.
Roughly speaking, a bordism is an $H_d$-manifold $X$ with the boundary given by the disjoint union of $Y_0$ and $Y_1$.
More precise definition may be as follows.
\begin{defi}\label{defi:bordism}
A bordism from $Y_0$ to $Y_1$ is a 7-tuple 
\beq
(X, (\partial X)_0, (\partial X)_1,Y_0, Y_1, \varphi_0, \varphi_1)
\eeq 
as follows.
The $X$ is a compact $H_d$-manifold whose boundary is a manifold 
consisting of the disjoint union of two closed manifolds $(\partial X)_0 \sqcup (\partial X)_1$.
The $(\partial X)_1$ and $(\partial X)_0$ are given the $H_{d-1}$-structures induced from $X$ by using the outward normal vector 
and inward normal vector for $(\partial X)_1$ and $(\partial X)_0$, respectively.
The $\varphi_i~(i=0,1)$ are isomorphisms
$
\varphi_i: (\partial X)_i \to Y_i~(i=0,1)
$.
Moreover, there exists a neighborhood $U_\epsilon$ of the boundary $\partial X \subset U_\epsilon \subset X$ such that 
$U_\epsilon$ is isomorphic to the disjoint union of 
$ [0,\epsilon ) \times Y_0$ and $(-\epsilon,0] \times Y_1 $  for small enough $\epsilon >0$ with the product $H_d$-structure on them. 

Two bordisms $(X, (\partial X)_0, (\partial X)_1, Y_0, Y_1, \varphi_0, \varphi_1)$ and 
$(X' , (\partial X')_0, (\partial X')_1, Y_0, Y_1, \varphi'_0, \varphi'_1)$ are identified if there exists an isomorphism $\Phi : X \to X'$ such that 
$\varphi_i = \varphi'_i \circ \Phi|_{(\partial X)_i}$.
Also, two bordisms are identified under the following homotopy.
If there exists a smooth one-parameter family of $H_d$-structures on $X$ parametrized by $s \in [0,1]$, 
denoted as $\varphi(s): P \to TX$, 
such that $\varphi(s)$ is independent of $s$ in some neighborhood of the boundary $\partial X$, 
then the two bordisms which has the $H_d$-structures 
at $s=0$ and $s=1$ are identified.
\end{defi}

We often abbreviate a bordism $(X, (\partial X)_0, (\partial X)_1, Y_0, Y_1, \varphi_0, \varphi_1)$ as just $(X, Y_0, Y_1)$ or
$X: Y_0 \to Y_1$, or more simply as $X$. 
In some cases, the full structure of the bordism given by the 7-tuple
$(X, (\partial X)_0, (\partial X)_1, Y_0, Y_1, \varphi_0, \varphi_1)$ is important, and in those cases we will write them explicitly. 
The $Y_0$ and $Y_1$ may be called as ingoing and outgoing boundary, respectively.

\begin{rem}
In the above definition of bordism, the identification under the homotopy $\varphi(s)$ is introduced. This is because
we want to discuss ``topological" QFT in which the partition function is invariant under such homotopy.
However, notice that we fix the boundary under the homotopy. Continuous change of the boundary would lead to physical effects
such as Berry phases.
\end{rem}
\begin{rem}
In the above definition,
the assumption about the existence of $U_\epsilon$ which has the product structure $ [0,\epsilon ) \times Y_0 \sqcup (-\epsilon,0] \times Y_1 $
is imposed for technical simplicity~\cite{Freed:2012hx}. 
For TQFT, we do not lose anything by this assumption.
This is because if we want to cut and glue a manifold $X$ along a codimension-1 submanifold $Y$, we can first change the metric
continuously near $Y$ so that it has a neighborhood with the product $H_d$-structure.
Alternatively, it is also possible to require only continuous (i.e., not necessarily smooth) $H_d$-structure, while base manifolds themselves are still smooth.
\end{rem}

If we are given two bordisms 
\beq
(X_0, (\partial X_0)_0, (\partial X_0)_1,Y_0, Y_1, \varphi_{0,0}, \varphi_{0,1}) \nonumber \\
(X_1, (\partial X_1)_0, (\partial X_1)_1,Y_1, Y_2, \varphi_{1,0}, \varphi_{1,1}), 
\eeq
we can glue them together along 
the boundary $Y_1$ as follows. Let $(P_0, \varphi_0)$ and $(P_1, \varphi_1)$ be the $H_d$-structures on $X_0$ and $X_1$,
respectively. Also, let $e_{0,1}$ be the outward normal vector to $(\partial X_0)_1$ and let $e_{1,0}$ be the inward normal vector to $(\partial X_1)_0$.
Then, the induced $H_{d-1}$-structures on these boundaries are given by $P'_0 = \varphi^{-1}_0(e_{0,1}, *)$ and 
$P'_1 =  \varphi^{-1}_0(e_{1,0}, *)$, respectively. Then we identify the boundaries as follows. First, we idenfity $p'_0 \in P'_0 \subset P_0$
and $p'_1 \in P'_1 \subset P_1$ as $p'_1 = \varphi^{-1}_{1,0} \circ \varphi_{0,1}(p'_0)  $. For more general element $p_0 \in P_0|_{(\partial X_0)_1}$,
we represent it as $p_0=p'_0 h$ for a (not unique) $h \in H_d$ where $p'_0 \in P'_0$. 
Then we identify $P_0|_{(\partial X_0)_1}$ and $ P_1|_{(\partial X_1)_0}$ 
by the map $p_0 \mapsto  \varphi^{-1}_{1,0} \circ \varphi_{0,1}(p'_0)h$.
It is easy to check that it is independent of the choice of the pair $(p'_0,h)$, as we have done in \eqref{eq:dcheck}. 
In this way,  $P_0|_{(\partial X_0)_1}$
and $ P_1|_{(\partial X_1)_0}$ are identified, and hence $X_0$ and $X_1$ are glued.
This is possible in a smooth and unique way because of our technical assumption that a neighborhood of the boundary has the product structure.
Then we get a manifold $X_1 \cdot  X_0$ obtained by the gluing. In this way, we get a new bordism $X_1 \cdot X_0 : Y_0 \to Y_2$.
This is the composition of two bordisms.

Suppose that we are given an $H_d$-manifold with a boundary component $(\partial X)_c$ which is given the $H_{d-1}$-structure by using
the outward normal vector, and assume that there is an isomorphism $\varphi_c: (\partial X)_c \to Y$ using that $H_{d-1}$-structure. 
Then it can be seen as outgoing boundary.
However, we can also give $(\partial X)_c$ the $H_{d-1}$-structure by using the inward normal vector,
and regard it as ingoing. 
This is done as follows. Let $P$ and $P'$ be the $H_d$-bundles associated to $FX|_{(\partial X)_c}=F(N \oplus T(\partial X)_c )$ and $F(\underline{\BR} \oplus TY)$,
respectively.
We have an isomorphism between $P$ and $P'$
by trivializing $N \cong \underline{\BR}$ by using the outward normal vector $e_{\rm out}$. 
Now we restrict that isomorphism to the inverse images of 
$(e_{\rm in}, *) \in FX|_{(\partial X)_c}$ 
and $(-e_0, *) \in F(\underline{\BR} \oplus TY)$,
where $e_{\rm in}=-e_{\rm out}$ is the inward normal vector, and $e_0$ is the basis vector of $\underline{\BR}$.
This defines the isomorphism $\varphi'_c: (\partial X)_c \to \overline{Y}$ where now $(\partial X)_c$ is regarded as ingoing. 
Thus, $(\partial X)_c$ can be regarded as outgoing $Y$ or ingoing $\overline{Y}$ by choosing the outward or inward normal vectors, respectively.

In particular, by using the above discussions, we can define the following.
\begin{defi}\label{defi:evcoev}
Let $Y$ be a closed $H_{d-1}$-manifold, and let $I_Y:=[0,1] \times Y$ with the product $H_d$-structure. 
The identity bordism $1_Y$, the evaluation $e_Y$, and coevaluation $c_Y$ are defined as
\beq
1_Y &= ( I_Y, Y, Y), \label{eq:idB} \\
e_Y& = ( I_Y , Y \sqcup \overline{Y}, \varnothing), \label{eq:evB} \\
c_Y &= ( I_Y , \varnothing, \overline{Y} \sqcup Y). \label{eq:coevB}
\eeq
\end{defi}

Another operation we can do to a bordism $X : Y_0 \to Y_1$ is to take the opposite $\overline{X}$.
Then it can be seen as a bordism $\overline{X} : \overline{Y_0} \to \overline{Y_1}$.
This requires some careful consideration of definitions. 
To define $\overline{X}$, we first uplift the $H_d$-structure to the $H_{d+1}$-structure
$\varphi_{d+1}:P_{d+1} \to F(\underline{\BR} \oplus TX)$.
Then the bundle on $\overline{X}$ is defined as $\overline{P_d} = \varphi_{d+1}^{-1}(-e_0, *) $ for $ * \in FX $.

The restriction of $F(\underline{\BR} \oplus TX)$ to $(\partial X)_0$ is canonically isomorphic to $F( \underline{\BR} \oplus N \oplus T (\partial X)_0)$,
where $N$ is the normal bundle to $(\partial X)_0$ and we identify it with a copy of the trivial bundle $ \underline{\BR} $
by using the inward normal vector $e_{\rm in}$. Then the $H_{d-1}$-structure on the boundary $(\partial \overline{X})_0$ is described by the
bundle $\overline{P_{d-1}} := \varphi^{-1}_{d+1}( -e_0, e_{\rm in},*)$ for $* \in F (\partial X)_0 $. Then we want to define an isomorphism from $\overline{P_{d-1}} $ 
to the bundle on $\overline{Y_0}$. Notice that what we have is the isomorphism from the bundle 
${P}_{d-1} :=\varphi^{-1}_{d}( e_{\rm in}, *)$ to $Y_0$, and what we want is an isomorphism from 
$\overline{P_{d-1}} =\varphi^{-1}_{d+1}( -e_0, e_{\rm in}, *)$
to $\overline{Y_{0}}$.

For the purpose of defining the desired isomorphism, we pick up an element 
\beq
r \in H_{d+1}  \text{ such that }  \rho_{d+1}(r) = (-1) \oplus (-1) \oplus (1_{d-1}). 
\eeq
Namely, $\rho_{d+1}(r) $ flips the sign of the first two coordinates of 
$d+1$-dimensional vectors, while it acts trivially on other components.
Such an element can be specified by using the $\Spin(d+1)$ part of the group
$\rho^{-1}_{d+1}(\SO(d+1)) \cong (\Spin(d+1) \times K)/\langle (-1, k_0 ) \rangle$.
We can specify $r \in H_{d+1}$ if we require it to be either the $180^\circ$ rotation or the $- 180^\circ$ rotation on the 
two-dimensional plane spanned by the first two components of $(d+1)$-dimensional vectors. We denote the $180^\circ$ rotation by $r$.
This $r$ commutes with $H_{d-1}$ embedded in $H_{d+1}$.

By using $r$, we define a bundle automorphism of $P_{d+1}=P_d \times_{i_d} H_{d+1}$ (or gauge transformation in physics terminology) near the boundary.
Let $s_\alpha$ be local sections of $P_{d+1}$ 
on open cover $\{ U_\alpha \}$ of a neighborhood of $(\partial X)_0$,  such that 
\beq
\varphi_{d+1}(s_\alpha) = ( -e_0, e_{\rm in},*) . \label{eq:loca}
\eeq
Then let $p=s_\alpha  p_\alpha \in P_{d+1}$
where $ p_\alpha$ is the local coordinate of the fiber of the bundle on $U_\alpha$.
Then we define the automorphism of $P_{d+1}$ as 
\beq
\Phi : p=s_\alpha  p_\alpha \mapsto s_\alpha r p_\alpha. \label{eq:autm}
\eeq
This gives a well-defined bundle automorphism since the transition functions of the bundle on $(\partial X)_0$ is reduced to $H_{d-1}$,
and $r$ commutes with $H_{d-1} \subset H_{d+1}$.

By the isomorphism $\Phi$, the bundle $\overline{P_{d-1}} :=\varphi^{-1}_{d+1}( -e_0, e_{\rm in}, *)$ is mapped to 
the bundle $ \varphi^{-1}_{d+1}( e_0, -e_{\rm in}, *) $ as principal $H_{d-1}$-bundles. 
The bundle $ \varphi^{-1}_{d+1}( e_0, -e_{\rm in}, *)$  in turn can be canonically identified with $ \varphi^{-1}_{d}( -e_{\rm in}, *)$.
From this $ \varphi^{-1}_{d}( -e_{\rm in}, *)$, we can define the canonical isomorphism to $\overline{Y}$ which was already explained before.

We can apply the same consideration to the outgoing component of the boundary by using the outward normal vector.
Therefore, we can regard $\overline{X}$ as a bordism $\overline{X}: \overline{Y_0} \to \overline{Y_1}$.

The $-180^\circ$ rotation $r^{-1}$ can be different from $r$ in the case of spin manifolds or a generalization of spin manifolds
where $(-1)^F:=(-1,1) \sim (1,k_0)$ is nontrivial.
It is free to use either $r$ or $r^{-1}$, but we use the same one for both ingoing and outgoing boundaries.
Then the two choices are equivalent. The reason is that the difference is given by $(-1)^F$,
and the action of $(-1)^F$ can be extended as a bundle automorphism of $P_{d+1}$ not only near the boundary, but also
on the entire $X$ because $(-1)^F$ commutes with any element of $H_{d+1}$.
Thus the two choices of using $r$ or $r^{-1}$ are related by this bundle automorphism $(-1)^F$.
The definition of bordisms given in Definition~\ref{defi:bordism} says that we identify two bordisms if they are related by an isomorphism.
Therefore, they are equivalent. In this way, the new bordism $\overline{X}: \overline{Y_0} \to \overline{Y_1}$ is defined independent of the choice of $r$ or $r^{-1}$.

The reason that we use the same choice ($r$ or $r^{-1}$) for both ingoing and outgoing boundary components is to obtain the properties: $\overline{X_1 \cdot X_0} = \overline{X_1} \cdot \overline{X_0}$ for
the composition of two bordisms $X_0: Y_0 \to Y_1$ and $X_1: Y_1 \to Y_2$, and also $\overline{1_Y}=1_{\overline{Y}}$.
In summary, we have
\begin{lem}\label{lem:involution}
Taking the opposite $H_d$-structure gives a map 
from the class of bordisms to itself which map $X: Y_0 \to Y_1$ to $\overline{X}: \overline{Y_0} \to \overline{Y_1}$ such that $\overline{1_Y} = 1_{\overline{Y} }$
and $\overline{X_1 \cdot X_0} = \overline{X_1} \cdot \overline{X_0}$.
\end{lem}
In the next section, we will consider the bordism category. Then, this lemma means that taking the opposite $H_d$-structure is 
a functor from the bordism category to itself.
\begin{defi}\label{defi:double}
The double $\Delta X$ of a bordism $X: \varnothing \to Y$ is defined as $\Delta X = e_Y \cdot (X \sqcup \overline{X})$.
\end{defi}
\begin{rem}\label{rem:transp}
Given a bordism $X: Y_0 \to Y_1$, we may instead regard it as a bordism $\overline{Y_1} \to \overline{Y_0}$ by changing the interpretation of ingoing and outgoing boundary
components.
Let us denote that bordism as $^t X$. Then $\overline{^t X}$ is a bordism $Y_1 \to Y_0$. In particular, if $X$ is a bordism $X : \varnothing \to Y$,
the double in Definition~\ref{defi:double} is also given as $\Delta X = \overline{^t X} \cdot X$.
\end{rem}
Now we can also define the bordism group $\Omega^d_H$ which appears in the main classification theorem.
Let us consider an equivalence relation $X \sim X'$ between closed $H_d$-manifolds if there exists a bordism $C: X \to X'$.
From the above discussions, this equivalence relation is symmetric and transitive, while reflectivity is obvious. 
\begin{defi}\label{defi:bor}
The $d$-dimensional bordism group $\Omega^d_H$ with $H$-structure is, as a set, given by
the set of equivalence classes of closed $H_d$-manifolds $X$ under the equivalence relation $X \sim X'$ if there exists a 
$d+1$-dimensional bordism $C$ from $X$ to $X'$.
The abelian group structure is given by taking the disjoint union $\sqcup$ as the multiplication, the empty manifold $\varnothing$ as the unit element,
and the opposite $H_d$-manifold $\overline{X}$ as the inverse of $X$.
\end{defi}

\subsection{Subtle spin structure}\label{sec:subtle}
In addition to $\overline{1_Y}=1_{\overline{Y}}$, 
we would like to study the opposites of the evaluation and coevaluation defined in Definition~\ref{defi:evcoev}.
Namely, we are going to compare $\overline{c_Y}$ and $c_{\overline{Y}}$, and similarly for $e_Y$.
However, the spin structure is extremely subtle,
so let us study it carefully.
In the following, we discuss the case of ${c_Y}$, but ${e_Y}$ is completely parallel.

As an $H_d$-manifold, $c_Y$ is given by $I_Y:=[0,1] \times Y$.
First, let us neglect the boundary and establish $\overline{I_Y}=I_{\overline{Y}}$ in the interior. 
Let $P_d =P_{d-1} \times_{i_{d-1}} H_d$ be the $H_d$-bundle on $I_Y$
where $P_{d-1}$ is the $H_{d-1}$-bundle on $Y$. As usual, it is equipped with the map $\varphi_d: P_d \to F( \underline{\BR} \oplus TY)$ where $\underline{\BR}$
is the tangent bundle to $[0,1]$.
We denote the unit vector on $[0,1]$ as $e_0$.
To define the opposite $H_d$-structure $\overline{I_Y}$, we take 
\beq
P_{d+1}=P_d \times_{i_d} H_{d+1} = P_{d-1} \times_{i_{d-1}} H_d \times_{i_d} H_{d+1}. \label{eq:tripleproduct}
\eeq
with $\varphi_{d+1} : P_{d+1} \to F(\underline{\BR}' \times \underline{\BR} \times Y)$ as usual, where $\underline{\BR}'$ is introduced to define the opposite 
$H_d$-structure. We denote the unit vector of $\underline{\BR}' $ as $e'_0$.

We can represent the inverse image $\varphi^{-1}_{d+1}(-e'_0, e_0, *)$ by elements of the form
\beq
(p_{d-1}, h_d, h_{d+1} ) \in P_{d+1}=P_{d-1} \times_{i_{d-1}} H_d \times_{i_d} H_{d+1} \label{eq:represent}
\eeq
such that 
\beq
\rho_{d+1}(h_{d+1}) \in (-1) \oplus (-1) \oplus \O(d-1) , \qquad
\rho_d(h_d) \in (-1) \oplus \O(d-1)
\eeq
in the obvious notation.
Then, by acting $r$, we get $h_{d+1}r = r h_{d+1}$ which now satisfies $\rho_{d+1}(h_{d+1}r) \in (+1) \oplus (+1) \oplus \O(d)$ and hence it can be represented as
an image of some element $h'_d \in H_d$ as $h_{d+1}r = i_d( h'_d )$. Therefore, 
\beq
(p_{d-1}, h_d, h_{d+1}r  ) \sim (p_{d-1}, h_d h'_d)
\eeq
with $\rho_d( h_d h'_d) \in (-1) \oplus \O(d-1)$. Such elements precisely represent the opposite $\overline{Y}$.

For more general elements of $\varphi^{-1}_{d+1}(-e'_0,*)$ where $* \in F I_Y$,
we represent them as $(p_{d-1}, h_d, h_{d+1})$ such that
\beq
\rho_{d+1}(h_{d+1}) \in (-1)  \oplus \O(d),  \qquad   \rho_d(h_d) \in (-1) \oplus \O(d-1). \label{eq:barY}
\eeq 
In this case, $r$ and $h_{d+1}$ do not commute. We act $r$ as a gauge transformation, meaning that it acts from the left (instead of right) as $r h_{d+1}$.
Then, $rh_{d+1} = i_d(h'_d)$ for some $h'_d \in H_d$ and we can represent the elements of $\varphi^{-1}_{d+1}(-e'_0,*)$ 
by the triple $(p_{d-1}, h_d, h'_d) \in P_{d-1} \times_{i_{d-1}} H_d \times_{H_{d-1}} H_d$ such that $\rho_d(h_d) \in (-1) \oplus \O(d-1)$ and $h'_d$
is arbitrary. This is precisely the form of the $H_d$-bundle on $I_{\overline{Y}}$.
This establishes $\overline{I_Y}=I_{\overline{Y}}$ if we neglect the boundary.

Now we are going to study the boundary.
Let $t \in [0,1]$ be the standard coordinate of the interval.
The $c_Y: \varnothing \to  \overline{Y} \sqcup Y$ has the two boundary components: $t=1$ corresponding to $Y$
and $t=0$ corresponding to $\overline{Y}$. At $t=1$, it is straightforward that the above identification $\overline{I_Y}=I_{\overline{Y}}$
holds including the boundary isomorphism. In fact, we have chosen the action of $r$ so that 
this identification holds at $t=1$.

Let us consider the boundary component at $ t =0$. The boundary component here is isomorphic to $\overline{\overline{Y}} \cong Y$.
The outward normal vector at this boundary is given by $e^0_{\rm out} = - e_0$.
The inverse image $\varphi^{-1}_{d+1}(-e'_0, -e_0, *)$ is represented
by elements of the form
$(p_{d-1}, h_d, h_{d+1}) \in P_{d-1} \times_{i_{d-1}} H_d \times_{i_d} H_{d+1}$
such that 
\beq
\rho_{d+1}(h_{d+1}) \in (-1) \oplus (+1) \oplus \O(d-1), \qquad \rho_d(h_d) \in (-1) \oplus \O(d-1).
\eeq
We act $r$ on $h_{d+1}$, but here comes the subtlety. 
In the description of $c_{\overline{Y}}$, we act $r$ from the left as $r h_{d+1}$ as discussed in the paragraph containing \eqref{eq:barY}.
However, to define $\overline{c_{Y}}$, we must follow the prescription discussed around \eqref{eq:autm}.
We take a local section $s=(p_{d-1}, h_d, h_{d+1}) $ such that $\varphi_{d+1}(p_{d-1}, h_d, h_{d+1}) = ( -e'_0, -e_0,*)$,
and then act $r$ from the right as $sr$. This means that we use $h_{d+1}r$.
The fact that $\rho_{d+1}(h_{d+1}) \in (-1) \oplus (+1) \oplus \O(d-1)$ implies that 
\beq
r h_{d+1} =h_{d+1}r (-1)^F
\eeq
 due to the properties of the spin group.
Thus they differ by $(-1)^F$.

The above fact can also be understood as follows. In the definition of $\overline{c_{Y}} : \varnothing \to \overline{\overline{Y}} \sqcup \overline{Y}$ from 
${c_{Y}} : \varnothing \to {\overline{Y}} \sqcup {Y}$, we need to use $r$ in the following way.
At $t=1$, we need to consider the $180^\circ$ rotation $r$ of the plane spanned by $(-e'_0, e^1_{\rm out} )$
where $e^1_{\rm out}=e_0$ is the outward normal vector at $t=1$. On the other hand, at $t=0$, 
we need to consider the $180^\circ$ rotation $r$ of the plane spanned by $(-e'_0, e^0_{\rm out} )$
where $e^0_{\rm out}=-e_0$ is the outward normal vector at $t=1$. These two rotations differ by a $360^\circ$ rotation due to the sign change in $\pm e_0$.
This is the reason that we get $(-1)^F$ above.

Therefore, we conclude that $\overline{c_{Y}}$ and $c_{\overline{Y}}$ differ by the additional boundary isomorphism $(-1)^F$ at one of the boundary components.
In the above discussion, the $(-1)^F$ acted on the component $t=0$. But it is not relevant which boundary component is acted by $(-1)^F$
because they can be exchanged by the overall action of $(-1)^F$ on the entire manifold $I_Y$.

Let us further compare $\overline{c_Y}$ and $c_Y$. We have seen that there is an additional $(-1)^F$ in $\overline{c_Y}$.
However, in the isomorphism $\overline{\overline{Y}} \cong Y$ given in \eqref{eq:barbarmap}, we have included the factor $(-1)^F$.
These two $(-1)^F$ cancels with each other when we regard $\overline{c_Y}$ as a bordism $\varnothing \to Y \sqcup \overline{Y}$.
More precisely, from the triplet $(p_{d-1}, h_d, h'_d)$ where $i_d(h'_d)=rh_{d+1}$, we further map
\beq
(p_{d-1}, h_d, h'_d) \mapsto (p_{-1}, h_d h'_d) \in P_{d-1} \times_{i_{d-1}} H_d
\eeq
with the projection defined as
\beq
\varphi'_d: (p_{d-1}, h_d h'_d) \mapsto (-e_0, \varphi_{d-1}(p_{d-1}))\rho_d(h_d h'_d).
\eeq
Combining this with the differomorphism $[0,1] \ni t \mapsto (1-t) \in [0,1]$, one can see that
we actually have $\overline{c_Y} = \tau_{\overline{Y},Y} \cdot c_Y  $ as a bordism $\varnothing \to Y \sqcup \overline{Y}$.
Here, $\tau_{Y, Y'}$ is the bordism defined as follows. As an $H_d$-manifold, it is just $([0,1] \times Y) \sqcup ([0,1] \times Y')$.
However, it is regarded as a bordism $Y \sqcup Y' \to Y' \sqcup Y$, i.e. it exchanges $Y$ and $Y'$. See Figure~\ref{fig:exchange}.
\begin{figure}
\centering
\includegraphics[width=.5\textwidth]{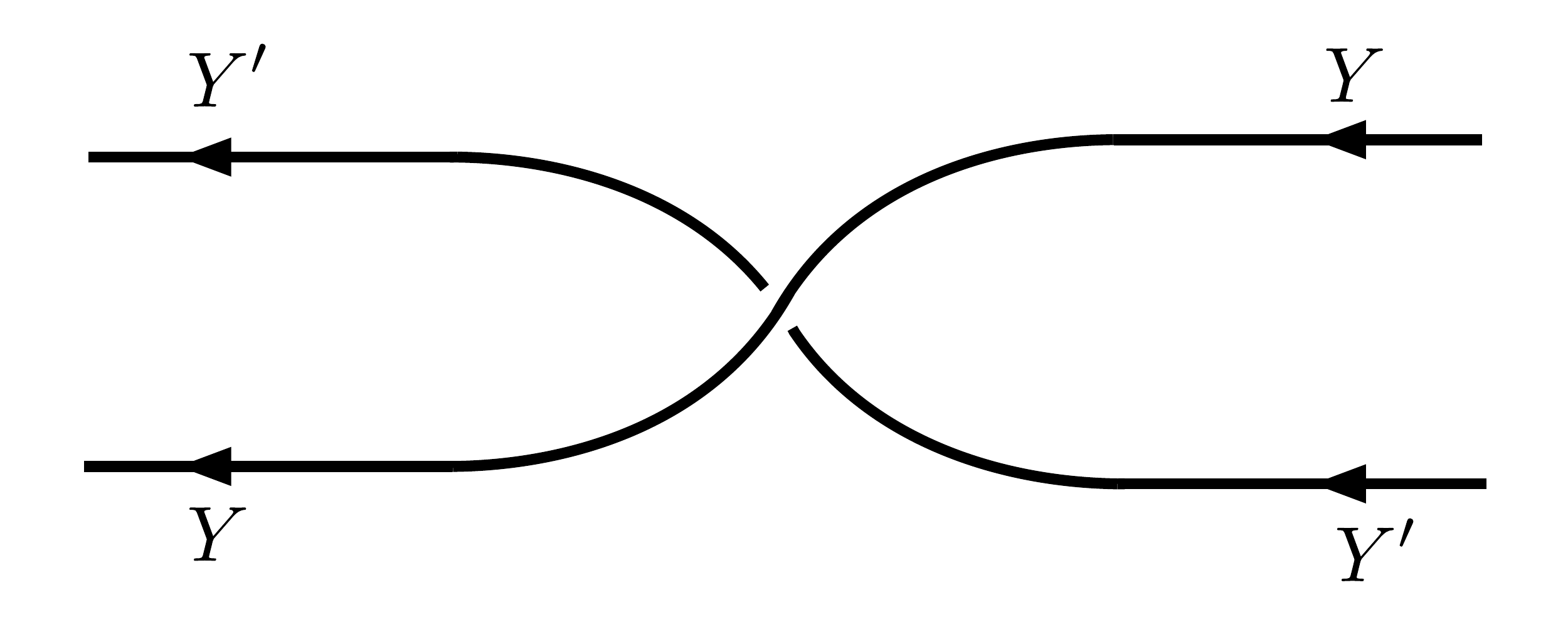}
\caption{The bordism $\tau_{Y,Y'} $ from $Y \sqcup Y'$ to $Y' \sqcup Y$. 
As an $H_d$-manifold, it is just $([0,1] \times Y) \sqcup ([0,1] \times Y') $.  \label{fig:exchange}}
\end{figure}
In summary,
\begin{lem}\label{lem:anti}
There are relations
\beq
( (-1)^F \sqcup 1_{\overline{Y}} ) \cdot c_{\overline{Y}}   = \overline{c_{Y}}  =  \tau_{\overline{Y},Y} \cdot c_Y, \\
e_{\overline{Y}} \cdot (1_{\overline{Y}} \sqcup (-1)^F) = \overline{e_{Y}}  =   e_Y \cdot \tau_{\overline{Y},Y} 
\eeq
where $(-1)^F$ acts on $\overline{\overline{Y}} \cong Y$ and it means that the boundary isomorphism contains $(-1)^F$.
\end{lem}
\begin{rem}\label{rem:anti}
We have defined the double in Definition~\ref{defi:double}. 
Let us consider the double of $c_Y$, 
\beq
\Delta c_Y = e_{ \overline{Y} \sqcup Y} \cdot (c_Y \sqcup \overline{c_Y}).
\eeq
This is topologically $S^1 \times Y$.
In fact, the $S^1$ has the anti-periodic spin structure which we denote as $S^1_{\rm A}$. 
The reason is as follows. Up to permutation of the boundary components by $\tau_{Y,Y'}$, 
$e_{Y \sqcup \overline{Y} }$ is the disjoin union of $e_Y$ and $e_{\overline{Y}}$. Then, 
up to permutation of the boundary components, $e_{\overline{Y}}$ has the extra factor of $(-1)^F$ compared to $e_Y$ as shown in Lemma~\ref{lem:anti} above.
Therefore, the $S^1$ has the anti-periodic spin structure due to the insertion of $(-1)^F$.
\end{rem}

\subsection{Axioms of TQFT}\label{sec:Atiyah}

Here we introduce the axioms of TQFT.
The axioms are described by Atiyah~\cite{Atiyah:1989vu}, 
and they can be explained in the language of category (see e.g. \cite{Freed:2016rqq,Freed:2012hx}).
See \cite{Baez} for categorical notions needed for TQFT which we will reproduce in appendix~\ref{sec:app}. For completeness, we recall a few elementary definitions.

Let us recall the basic terminology from category theory. 
A category $\sC$ consists of 
(1) a class of objects ${\rm obj}(\sC)$, (2) a set of morphisms ${\rm Hom}(A,B)$ for each pair of objects $A, B \in {\rm obj}(\sC)$,
and (3) composition ${\rm Hom}(A,B) \times {\rm Hom}(B,C) \to {\rm Hom}(A,C)$ denoted as $(f, g) \to g \circ f$.
They satisfy the following axioms; 
(i) composition is associative, i.e., for $f \in {\rm Hom}(A,B)$, $g \in {\rm Hom}(B,C)$ and $h \in {\rm Hom}(C,D)$, we have
$(h \circ g) \circ f = h \circ (g \circ f)$, (ii) for each $A$, there exists an element $1_A \in {\rm Hom}(A,A)$ called the identify morphism,
with the properties that $f \circ 1_A = 1_B \circ f = f$ for any $f \in {\rm Hom}(A,B)$.
A morphism $f \in {\rm Hom}(A,B)$ is also written as $f: A \to B$.

A functor $F: \sC \to \sD$ from a category $\sC$ to a category $\sD$ is a function as follows.
For each object $A \in {\rm obj}(\sC)$ it gives an object $F(A) \in {\rm obj}(\sD)$. For each morphism $f \in {\rm Hom}(A,B)$
it gives a morphism $F(f) \in {\rm Hom}(F(A),F(B))$ such that it preserves composition, $F( g \circ f) = F(g) \circ F(f)$.
Also, for $1_A$ it gives $F(1_A)=1_{F(A)}$.

A natural transformation $\eta$ between two functors $F: \sC \to \sD$ and $G: \sC \to \sD$ consists of
a family of morphisms $\eta_A \in {\rm Hom}(F(A), G(A))$ parametrized by objects $A \in {\rm obj}(\sC)$ such that
for each morphism $f : A \to B$ the diagram
\beq
\xymatrix{
F(A) \ar[d]_{F(f)}  \ar[r]^{\eta_A} & G(A)  \ar[d]^{G(f)} \\
F(B) \ar[r]_{\eta_B} & G(B)
}
\eeq
commutes, i.e., $G(f) \circ \eta_A = \eta_B \circ F(f)$. A natural isomorphism is a natural transformation such that 
there exists an inverse $\eta^{-1}$ i.e., $ \eta^{-1}_A \circ \eta_A = 1_{F(A)}$ and $\eta_A \circ \eta^{-1}_A = 1_{G(A)}$.
If there exists a natural isomorphism between two functors $F$ and $G$, they are called isomorphic.

A category can have additional structures.
One of them is the symmetric monoidal structure as follows.
Instead of writing the full definition (which can be found in appendix~\ref{sec:app}), we only describe the key points for our purposes.
For each pair of objects $A,B$ we have a product $A \otimes B \in {\rm obj}(\sC)$, and 
for each pair of morphisms $f \in {\rm Hom}(A,C)$ and $g \in {\rm Hom}(B,D)$ we have $f \otimes g \in {\rm Hom}(A \otimes B, C \otimes D)$.
We can regard this operation $\otimes$ as a functor from the product category $\sC \times \sC$ 
(with the obvious definition) to $\sC$,  $\otimes : \sC \times \sC \to \sC$. We can also define another functor 
such that $(A,B) \to B \otimes A$ and $(f,g) \to g \otimes f$. Then, symmetric monoidal structure requires, among other things,
that there exists a natural isomorphism  $\tau_{A,B}: A \otimes B \to B  \otimes A$ such that $\tau_{B , A} \circ \tau_{A,B} = 1_{A \otimes B}$.
It is also required that there is a distinguished object $1_\sC \in {\rm obj}(\sC)$ called the unit object. 

Roughly speaking, $\otimes$ is a kind of abelian semi-group multiplication, although it is not strictly a semi-group.
The usual equalities of an abelian semi-group such as $A \otimes B = B \otimes A$, $(A \otimes B) \otimes C= A \otimes (B \otimes C)$,
and $1 \otimes A = A \otimes 1=A$
are replaced by natural isomorphisms; we have the natural isomorphism $\tau_{A,B}: A \otimes B \to B  \otimes A$ as well as
$(A \otimes B) \otimes C \to A \otimes (B \otimes C)$, $1 \otimes A \to A$, and $A \otimes 1 \to A$, with several conditions (given in appendix~\ref{sec:app}). 
However, in our applications, the natural isomorphisms of $1 \otimes A \to A$, $A \otimes 1 \to A$, and
$(A \otimes B) \otimes C \to A \otimes (B \otimes C)$
are canonical and hence they are straightforward.
The only point which we need to be careful about is the $\tau_{A,B}$ as will be discussed later.

For TQFT, we need two symmetric monoidal categories.
One of them is the category of complex vector spaces ${\rm Vect}_\BC$ (or the category of super vector spaces 
${\rm sVect}_\BC$ as we discuss later).
The objects of ${\rm Vect}_\BC$ are complex vector spaces $V$, the set of morphisms ${\rm Hom}(V,W)$ is the space of linear maps from $V$ to $W$,
and composition is just the composition of linear maps. The $1_V \in {\rm Hom}(V,V)$ is the identity map.
The symmetric monoidal structure is given by the tensor product 
of two vector spaces $V \otimes W$. 
We can identify $\BC \otimes V \cong V \otimes \BC \cong V$, and $(V \otimes W) \otimes U \cong V \otimes (W \otimes U)$, 
and we can have maps $\tau'_{V,W}: V \otimes W \xrightarrow{\sim} W \otimes V$.
Thus this is a symmetric monoidal category with the unit object given by $\BC$.

The other category is the bordism category ${\rm Bord}_{\langle d-1, d \rangle}(H)$. The objects are closed $H_{d}$-manifolds $Y$,
morphisms in ${\rm Hom}(Y_0, Y_1)$ are bordisms $(X , Y_0 , Y_1)$, and composition is defined as $(X_1, Y_1, Y_2) \circ (X_0,Y_0,Y_1) =
(X_1 \cdot X_0, Y_0, Y_2)$ by gluing $X_0$ and $X_1$ at $Y_1$ as mentioned in the previous subsection. 
The identity morphism $1_Y \in {\rm Hom}(Y,Y)$ is the bordism given by $[0, \epsilon] \times Y$ with the product $H_d$-structure induced from $Y$.
Due to homotopy invariance, the value of $\epsilon$ does not matter.
The symmetric monoidal structure is given by the disjoint union of manifolds $Y \sqcup Y' $ and $X \sqcup X'$, so we
replace $\otimes$ by $\sqcup $ in this case.
The unit object is given by the empty manifold $\varnothing$.
The natural isomorphism $\tau_{Y,Y'} : Y \sqcup Y' \to Y' \sqcup Y $ is given by the bordism depicted in Figure~\ref{fig:exchange}, which 
is $([0,1] \times Y) \sqcup ([0,1] \times Y' )$ as an $H_d$-manifold, but regarded as a bordism from $Y \sqcup Y'$ to $Y' \sqcup Y$.

If the categories $\sC$ and $\sD$ are symmetric monoidal categories,
we can require a functor $F : \sC \to \sD$ to preserve the symmetric monoidal structure. 
For our purposes this means that we can identify $F(1_{\sC}) \cong 1_{\sD}$, and we also have a natural isomorphism
\beq
F( A \otimes B) \xrightarrow{\sim} F(A) \otimes F(B)~~~\text{under which}~~~F(f \otimes g) = F(f) \otimes F(g) \label{eq:productID}
\eeq
where $\xrightarrow{\sim}$ means the natural isomorphism.
This identification is required to satisfy the conditions that
\beq
&F( (A\otimes B) \otimes C) \to F(A\otimes B) \otimes F(C) \to (F(A)\otimes F(B)) \otimes F(C) \nonumber \\
&F( A\otimes (B \otimes C) ) \to F(A) \otimes F( B \otimes C) \to F(A)\otimes (F(B) \otimes F(C) ) \label{eq:SMCcondition1}
\eeq
gives the same result under $(A \otimes B ) \otimes C \cong A \otimes (B \otimes C)$ etc. If $F$ satisfies these conditions (as well as 
other more simple conditions involving the unit objects), it is called a monoidal functor. If we also require that
\beq
F(\tau^{\sC}_{A,B}) = \tau^{\sD}_{F(A),F(B)}~~\text{under the identification}~~ F( A \otimes B) \xrightarrow{\sim} F(A) \otimes F(B), \label{eq:SMCcondition2}
\eeq
then it is called a symmetric monoidal functor.
See appendix~\ref{sec:app} for more details. 

\begin{rem}
As mentioned above, we use just the standard canonical monoidal structure in the category of vector spaces, meaning that
we just identify $\BC \otimes V \cong V \otimes \BC \cong V$ under which $1 \otimes v = v \otimes 1 = v$, and
identify $(V \otimes W) \otimes U \cong V \otimes (W \otimes U)$ under which $( v \otimes w) \otimes u = v \otimes (w \otimes u)$
where $v, w, u$ are vectors in $V,W, U$ respectively. So we do not bother to distinguish them, and just write e.g.,
$V \otimes W \otimes U$ and $v \otimes w \otimes u$, and so on.
\end{rem}
\begin{rem}\label{rem:svect}
However, the symmetric structure 
$\tau'_{V,W}: V \otimes W \xrightarrow{\sim} W \otimes V$ requires a bit care due to the fact that state vectors in the physical Hilbert spaces can be bosonic or fermionic,
and fermionic state vectors anti-commute with each other. Namely, we can have $v \otimes w \mapsto \pm w \otimes v$.
This sign is actually determined by the topological spin-statistics theorem (Sec.~11 of \cite{Freed:2016rqq}) as we will briefly review later.
Then it is more appropriate to consider super vector spaces (i.e., $\BZ_2$-graded vector spaces) 
$V= V^0 \oplus V^1$, where $V^0$ and $V^1$ contain ``bosonic" and ``fermionic" states of $V$.
We denote the category of super vector spaces as ${\rm sVect}_\BC$.
Let ${\rm deg}(v)=0,1$ when $v \in V^0$ and $v \in V^1$, respectively.
Then $\tau'_{V,W}$ in ${\rm sVect}_\BC$ is defined as 
\beq
v \otimes w \mapsto (-1)^{{\rm deg}(v) {\rm deg}(w)} w \otimes v.
\eeq
 We assume this symmetric monoidal structure of ${\rm sVect}_\BC$ in the following.
 In particular, for invertible TQFT, every Hilbert space is one-dimensional, and hence $\tau'_{V,W} = \pm 1$
 depending on whether both of $V$ and $W$ are fermionic or not.
 We also assume that morphisms in ${\rm sVect}_\BC$ preserve the degrees.
\end{rem}

\begin{rem}
Also in the bordism category,
we do not bother to distinguish $\varnothing \sqcup Y$ and $Y \sqcup \varnothing$ from $Y$,
and identify them all. Also, we identify
$(Y \sqcup Y') \sqcup Y''$ and $Y \sqcup (Y' \sqcup Y'')$ and regard them as $Y \sqcup Y' \sqcup Y''$.
However, we distinguish $Y \sqcup Y'$ from $Y' \sqcup Y$, and consider the natural isomorphism $\tau_{Y,Y'}$ of Figure~\ref{fig:exchange} explicitly.
\end{rem}

Now we give the axioms of TQFT. First we do not incorporate unitarity.
\begin{defi}
A topological quantum field theory (TQFT) with the symmetry group $H_d$ 
is a symmetric monoidal functor $F$ from the bordism category ${\rm Bord}_{\langle d-1, d \rangle}(H)$
to the category of super vector spaces ${\rm sVect}_\BC$.
\end{defi}
 Let us spell out the functor more explicitly and introduce some notations. For each closed $H_{d-1}$-manifold $Y$, we assign a vector space $F(Y)$.
 We denote this vector space as 
 \beq
 \CH(Y):=F(Y)
 \eeq
 and call it the Hilbert space on $Y$. 
 For each bordism $X$ from $Y_0$ to $Y_1$, we 
 assign a linear map $F(X)$ from $\CH(Y_0)$ to $\CH(Y_1)$. We also denote it as 
 \beq
 Z(X):=F(X).
 \eeq
The composition of two bordisms $X_0$ and $X_1$ gives $Z(X_1 \cdot X_0)= Z(X_1) Z(X_0)$ where 
the right hand side is the composition as linear maps. For the bordism $1_Y=[0,1] \times Y$ from $Y$ to $Y$, we have 
$Z(1_Y) = 1_{\CH(Y)}$.
For a disjoint union $Y_0 \sqcup Y_1$, we have $\CH(Y_0 \sqcup Y_1) \xrightarrow{\sim} \CH(Y_0) \otimes \CH(Y_1)$,
and for $X_0 \sqcup X_1$ we have $Z(X_0 \sqcup X_1) = Z(X_0) \otimes Z(X_1)$ under the above identification of 
$\CH(Y_0 \sqcup Y_1)$ and $ \CH(Y_0) \otimes \CH(Y_1)$.
For the empty manifold we have $\CH(\varnothing) \cong \BC$, and hence
if $X$ is a closed $H_d$-manifold, i.e., a bordism from $\varnothing $ to $\varnothing$, we get $Z(X) \in \BC$. This is called the partition function on $X$.

The $F(\tau_{Y,Y'})$, where $\tau_{Y,Y'}$ is the bordism $([0,1] \times Y \sqcup [0,1] \times Y' , Y \sqcup Y', Y' \sqcup Y)$ of Figure~\ref{fig:exchange},
is identified as the linear map $\tau'_{\CH(Y),\CH(Y') } : \CH(Y) \otimes \CH(Y') \xrightarrow{~} \CH(Y') \otimes \CH(Y)$ 
of the super vector spaces under the isomorphism $\CH(Y) \otimes \CH(Y') \xrightarrow{\sim} \CH(Y \sqcup Y')$.
This will need a care due to bosonic/fermionic statistics.

Next, we are going to define unitarity.
In the categories ${\rm Bord}_{\langle d-1, d \rangle}(H)$ and  ${\rm sVect}_\BC$,
we can define functors 
\beq
\beta_{\rm B}:~& {\rm Bord}_{\langle d-1, d \rangle}(H) \to {\rm Bord}_{\langle d-1, d \rangle}(H),  \nonumber \\
\beta_{\rm V} :~& {\rm sVect}_\BC \to {\rm sVect}_\BC.
\eeq
These are involution (see Definition~\ref{defi:involution} for a general definition of involution) defined as follows.
Given an $H_d$-manifold $X$, we have a manifold with opposite $H_d$-structure denoted as $\overline{X}$.
If $X: Y_0 \to Y_1$ is a bordism, we have $\overline{X}: \overline{Y_0} \to \overline{Y_1}$ as shown in Lemma~\ref{lem:involution}.
So we define $\beta_{\rm B}$ as $\beta_{\rm B}(Y) = \overline{Y}$ and $\beta_{\rm B}(X) = \overline{X}$ which is a functor by Lemma~\ref{lem:involution}.
On the other hand, in the category of vector spaces  ${\rm sVect}_\BC$, 
we can define the complex conjugate vector space $\overline{V}$ of
any vector space $V$. The complex conjugate of a map $f \in {\rm Hom}(V,W)$ is a map $\overline{f} \in {\rm Hom}(\overline{V},\overline{W})$.
Thus we define $\beta_{\rm V}(V) = \overline{V}$ and $\beta_{\rm V}(f) = \overline{f}$, which is clearly a functor.
We have isomorphisms $\overline{ Y \sqcup Y'} \xrightarrow{\sim} \overline{Y} \sqcup \overline{Y'}$
and $\overline{V \otimes V'} \xrightarrow{\sim} \overline{V} \otimes \overline{V'}$ under which $\overline{\tau_{Y,Y'}}=\tau_{\overline{Y},\overline{Y'}}$
and $\overline{\tau'_{V,V'}} = \tau'_{\overline{V} , \overline{V'} }$.
Thus $\beta_{\rm B}$ and $\beta_{\rm V}$ are symmetric monoidal functors. 
\begin{rem}\label{rem:involutionsign}
In the category of super vector spaces ${\rm sVect}_\BC$, the isomorphism  $\overline{V \otimes V'} \xrightarrow{\sim} \overline{V} \otimes \overline{V'}$ 
is subtle. We require that it is given as 
\beq
\overline{v \otimes w} \mapsto {(-1)^{{\rm deg}(v) {\rm deg}(w)}} \overline{v} \otimes \overline{w}.
\eeq
A physical motivation is that if we have two grassmannian numbers $\eta$ and $\eta'$, then the complex conjugate of their product
is, according to the standard physics rule, given as $\overline{\eta \cdot \eta'} = \overline{\eta'} \cdot \overline{\eta}=- \overline{\eta} \cdot \overline{\eta'}$.
A more precise argument of why it is necessary in order for unitary theories with nontrivial $(-1)^F$ to exist will be discussed at the end of Sec.~\ref{sec:property}.
\end{rem}

We can impose a TQFT functor $F: {\rm Bord}_{\langle d-1, d \rangle}(H) \to  {\rm sVect}_\BC$ to satisfy the following condition.
We have two functors $F \beta_{\rm B}$ and $\beta_{\rm V} F$, both of which are symmetric monoidal functors from $ {\rm Bord}_{\langle d-1, d \rangle}(H)$
to $ {\rm sVect}_\BC$. Then we want to identify them. More precisely, we require that there exists a natural isomorphism
between these functors which preserves symmetric monoidal structure (i.e., symmetric monoidal natural isomorphism; see appendix~\ref{sec:app}).
This essentially means that we can identify 
\beq
F(\overline{Y}) \xrightarrow{\sim} \overline{F(Y)}~~~\text{under which}~~~F(\overline{X}) = \overline{F(X)} \label{eq:invo}
\eeq
such that the diagram
\beq
\xymatrix{
F(\overline{Y}) \otimes F(\overline{Y'})  \ar[d]_{}  \ar[r]^{  }  & \overline{F(Y)} \otimes \overline{F(Y')}  \ar[d]^{} \\
F(\overline{ Y \sqcup Y'})   \ar[r]_{} &   \overline{F( Y \sqcup Y')} 
} \label{eq:involutioncondition}
\eeq
commutes. 

We also need to require a condition (see Definition~\ref{defi:equivariant} for details) which essentially 
says that the above natural isomorphism $F \beta_{\rm B} \Rightarrow \beta_{\rm V} F$ is consistent
with the identification $Y \cong \overline{\overline{Y}} $ and $V \cong \overline{\overline{V}} $,
where $Y \cong \overline{\overline{Y}} $ has been explained in the previous subsection and $V \cong \overline{\overline{V}} $ is just the canonical one for vector spaces.
Combining several identification maps, we have
\beq
F(Y) \to F(\overline{\overline{Y}}) \to \overline{F(\overline{Y}) } \to \overline{\overline{F(Y)}} \to F(Y), \label{eq:twiceinvo}
\eeq
where the first step uses $Y \cong \overline{\overline{Y}} $ and the last step uses $V \cong \overline{\overline{V}} $.
The composition of the above chain of maps is required to be the identity.

If $F$ satisfies the above conditions, $F$ is called an equivariant functor for the involution pair $(\beta_{\rm B}, \beta_{\rm V})$.
One of the requirements of unitarity is that $F$ is an equivariant functor.
Another requirement is that the evaluation $e_Y$ defined in \eqref{eq:evB} gives a positive definite hermitian inner product. 
\begin{defi}\label{defi:unitary}
A unitary TQFT is a TQFT such that $F$ is an equivariant functor for the involution pair $(\beta_{\rm B}, \beta_{\rm V})$,
and the evaluation $F(e_Y) : F(Y) \otimes \overline{F(Y)} \to \BC $ gives a positive definite hermitian metric on $F(Y)$.
\end{defi}
For more details, see Sec.~4 of \cite{Freed:2016rqq}. 
By using the positive definite sesquilinear form defined by $F(e_Y)$, we regard the vector spaces $\CH(Y)=F(Y)$ as finite dimensional Hilbert spaces.
Namely, $\CH(Y)$ has the positive definite hermitian inner product determined by $F(e_Y)$.

\begin{rem}
The statement that $F(e_Y) $ is hermitian is described in more detail as follows.
The conjugate $F(\overline{e_Y})=\overline{F(e_Y) }$ gives a map $ \overline{F(Y)} \otimes F(Y)  \to \BC$.
Then, by acting $\tau'_{ F(Y), \overline{F(Y)} }$, we get $\overline{F(e_Y) } \tau'_{ F(Y), \overline{F(Y)} }: F(Y) \otimes \overline{F(Y)} \to \BC$.
The hermiticity means that 
\beq
\overline{F(e_Y) } \tau'_{ F(Y), \overline{F(Y)} } = F(e_Y).
\eeq
There are sign factors in $\tau'_{ F(Y), \overline{F(Y)} }$  as $v \otimes w \mapsto (-1)^{{\rm deg}(v) {\rm deg}(w)} w \otimes v$ (Remark~\ref{rem:svect}) and in 
the involution as $\overline{v \otimes w} \mapsto {(-1)^{{\rm deg}(v) {\rm deg}(w)}} \overline{v} \otimes \overline{w}$ (Remark~\ref{rem:involutionsign}).
These two sign factors precisely cancel each other to give $\overline{v \otimes w} \mapsto  \overline{w} \otimes  \overline{v}$.
\end{rem}

Finally we define invertible TQFT as follows.
We impose the unitarity from the beginning in this paper. 
\begin{defi}
A unitary invertible TQFT is a unitary TQFT 
in which the Hilbert space $\CH(Y)=F(Y)$ for every $Y$ is one dimensional, $\dim \CH(Y) =1$.
\end{defi}
For an invertible TQFT, the $Z(X)$ for any bordism $X$ takes values in one dimensional Hilbert spaces.
First of all, we identify a dual space $\CH(Y)^*$ of $\CH(Y)$ with $\overline{\CH(Y)}$ by using the hermitian metric.
Then for a bordism $X: Y_0 \to Y_1$, we regard $Z(X)$ to be an element of the one-dimensional Hilbert space $\CH(Y_1) \otimes \overline{\CH(Y_0)}$.
Therefore, there is no problem in writing $Z(X_1 \cdot X_0)= Z(X_1 )Z( X_0) =  Z(X_0 )Z( X_1)$,
$Z(X_0 \sqcup X_1) = Z(X_1 )Z( X_0) =  Z(X_0 )Z( X_1)$, and so on. %Namely, ``taking inner product over some Hilbert spaces" is trivial for one-dimensional Hilbert spaces. 
This simplifies the notation.

\subsection{A few properties} \label{sec:property}
Here we review a few properties. We consider general unitary TQFTs.

\paragraph{Boundary isomorphism.}
The first property is about an isomorphism $\phi: Y \to Y'$ between two $H_{d-1}$-manifolds.
For each isomorphism $\phi$, we can define an operator $U(\phi) : \CH(Y) \to \CH(Y')$ as follows.
Let $I_Y=[0,1] \times Y$, and let $(I_Y, (\partial I_Y)_0, (\partial I_Y)_1,Y, Y', \varphi_0, \varphi_1)$ be a bordism,
where $\varphi_0:(\partial I_Y)_0 = \{ 0 \} \times Y \to Y $ is just the canonical one and 
$\varphi_1:(\partial I_Y)_1 = \{1 \} \times Y \to Y' $ is given by $\phi$.
This bordism is just $[0,1] \times Y$ as an $H_d$-manifold, but the boundary isomorphism is nontrivial.

Then we set 
\beq
U(\phi) : = Z(I_Y, (\partial I_Y)_0, (\partial I_Y)_1,Y, Y', \varphi_0, \varphi_1). \label{eq:unitaryop}
\eeq
One can check that $U$ satisfies $U( \phi' \circ \phi) = U(\phi')U(\phi)$. 
One can also see that the hermitian metric $Z(e_Y)$ satisfies $Z(e_{Y'})(U(\phi) \otimes  \overline{U(\phi)} )= Z(e_Y)$ (in the obvious notation),
so this $U(\phi)$ is a unitary operator. Moreover, let 
$(X, (\partial X)_0, (\partial X)_1,Y_0, Y_1, \varphi_0, \varphi_1)$ be a bordism and
$(X, (\partial X)_0, (\partial X)_1,Y'_0, Y'_1, \varphi'_0, \varphi'_1)$ be another bordism
which differs from the first one only by the boundary isomorphisms $\phi_i = \varphi'_i \circ \varphi^{-1}_i : Y_i  \to Y'_i$~($i=0,1$). Then we have
\beq
&Z(X, (\partial X)_0, (\partial X)_1,Y_0, Y_1, \varphi_0, \varphi_1) \nonumber \\
=& U(\phi_1)^{-1} Z(X, (\partial X)_0, (\partial X)_1,Y'_0, Y'_1, \varphi'_0, \varphi'_1) U(\phi_0).
\eeq
Thus, changing the boundary isomorphisms corresponds to multiplications of unitary operators $U(\phi)$.
This is familiar in quantum field theory in flat space $\BR^d$ where we consider isometries $\phi_{\rm isometry} : \BR^{d-1} \to \BR^{d-1}$
(e.g., rotation, translation, internal symmetry action) acting on the Hilbert space $\CH(\BR^{d-1})$.

In particular, there is always a distinguished isomorphism of any $H_{d-1}$-manifold which we denote (by abusing notation) as $(-1)^F$,
by using the fermion parity $(-1)^F \in (\Spin(d-1) \times K)/\langle (-1, k_0) \rangle \subset H_{d-1}$.
We can consider the corresponding operator $U((-1)^F)$. By further abusing the notation, we denote it as $(-1)^F$.
\begin{ex}\label{ex:time2}
One nontrivial example is about the case of the time-reversal. Let us consider an orientable manifold $Y$ for the case of $H_d = \Pin^\pm (d)$ as discussed in Example~\ref{ex:time1}.
There is an isomorphism ${\mathsf T}_Y: Y \to \overline{Y}$. This gives a linear map $U({\mathsf T}_Y): \CH(Y) \to \overline{\CH(Y)}$.
On the other hand, there is the canonical anti-linear map $\sigma:  \overline{\CH(Y)} \to \CH(Y)$ which is the identity as a map between the underlying real vector spaces.
Then we get the anti-linear map 
\beq
\sigma  U({\mathsf T}_Y): \CH(Y) \to \CH(Y) . 
\eeq
This is the usual time-reversal symmetry.
We have 
\beq
(\sigma  U({\mathsf T}_Y))^2=\overline{U({\mathsf T}_Y)}U({\mathsf T}_Y)=U(\overline{{\mathsf T}_Y})U({\mathsf T}_Y)=U({\mathsf T}_{\overline{Y}}   {\mathsf T}_Y)
\eeq
where the opposite $\overline{\phi} : \overline{Y} \to \overline{Y}'$ of an isomorphism ${\phi} : {Y} \to {Y}'$ is defined in the straightforward way.
Thus, $(\sigma  U({\mathsf T}_Y))^2 $ is $ (-1)^F$ for $\Pin^+$ and $+1$ for $\Pin^-$ by using the result of Example~\ref{ex:time1}.
This is one of the motivations of the factor $(-1)^F$ in \eqref{eq:barbarmap}.
\end{ex}

\paragraph{Reflection positivity.}
The next property is about reflection positivity. 
Let $(X, \varnothing, Y)$ be a bordism from $\varnothing$ to $Y$,
and $(\overline{X}, \varnothing, \overline{Y})$ be the opposite one. Then, from the positive definiteness of $Z(e_Y)$
we have $Z(e_Y)( Z(X) \otimes Z(\overline{X})) \geq 0$. 
On the other hand, this quantity can be evaluated as 
the partition function on the closed $H_d$-manifold $\Delta X$, the double of $X$ in Definition~\ref{defi:double}. 
Thus we obtain
\beq
Z(\Delta X)=F(\Delta X)  \geq 0.\label{eq:doublepositive}
\eeq
Geometrically, the double $\Delta X$ is also isomorphic to the $H_d$-manifold obtained by gluing $X : \varnothing \to Y$ and 
$\overline{^tX}: Y \to \varnothing$ as in Remark~\ref{rem:transp}, so we have $Z(\overline{^tX})Z(X) \geq 0$.

\paragraph{Relation between evaluation and coevaluation.}
The evaluation and coevaluation in Definition~\ref{defi:evcoev} give
\beq
E_Y: &=Z(e_Y) : \CH(Y) \otimes \overline{\CH(Y) } \to \BC , \\
C_Y: &=Z(c_Y) : \BC \to \overline{\CH(Y) } \otimes \CH(Y).
\eeq
The $C_Y$ can be identified with an element of $\overline{\CH(Y) } \otimes \CH(Y)$ by putting $1 \in \BC$ in the argument of the map.
Now consider the composition 
\beq
Y \xrightarrow{ 1_Y \sqcup c_Y } Y \sqcup \overline{Y} \sqcup Y \xrightarrow{ e_Y \sqcup 1_Y } Y, \label{eq:Scomposition}
\eeq
where $1_Y$ is the identity morphism defined in \eqref{eq:idB}. See Figure~\ref{fig:S}.
\begin{figure}
\centering
\includegraphics[width=.7\textwidth]{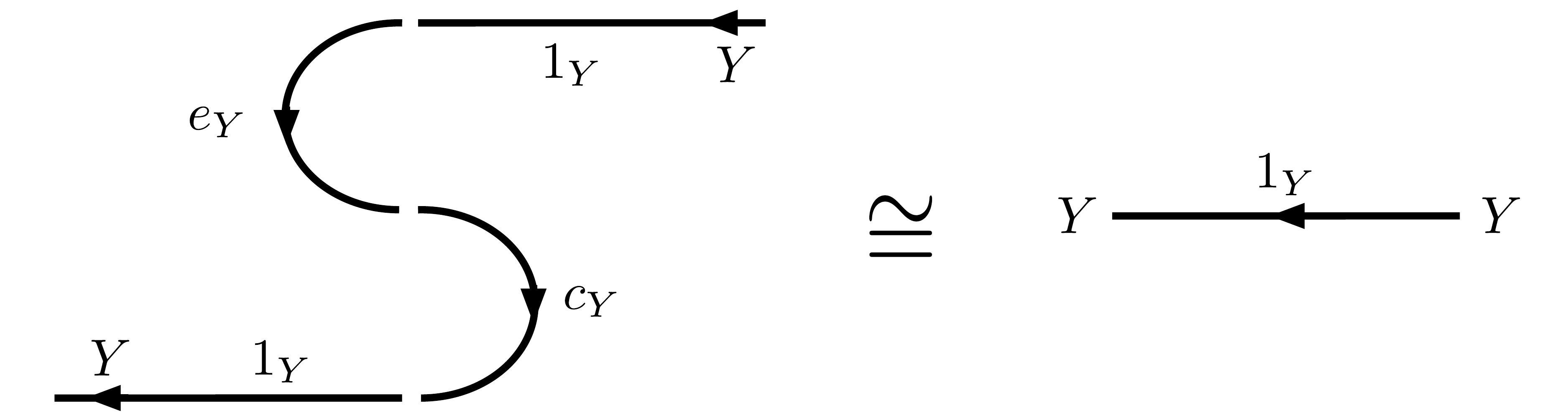}
\caption{The composition \eqref{eq:Scomposition}. \label{fig:S} }
\end{figure}
From the figure, it is clear that it is just equivalent to $1_Y$.
Thus we get
\beq
(E_Y \otimes 1_{\CH(Y) } ) ( 1_{\CH(Y) }  \otimes C_Y ) = 1_{\CH(Y) } .\label{eq: evcev}
\eeq
By taking explicit basis for $\CH(Y)$ and $\CH(\overline{Y})$, $E_Y$ and $C_Y$ can be written as
$(E_Y)_{i \bar{j}}$ and $(C_Y)^{\bar{j} i}$, where upper $i$ and $\bar{j}$ are indices for the vector spaces $\CH(Y)$ and $\CH(\overline{Y})$,
respectively, and lower $i$ and $\bar{j}$ are indices for the dual vector spaces. Then \eqref{eq: evcev}
is just saying that 
\beq
\sum_{\bar{j}} (E_Y)_{i \bar{j}} (C_Y)^{\bar{j} k} = \delta_i^k . \label{eq:evcev1}
\eeq

\paragraph{Topological spin-statistics theorem.}
Let us consider $\overline{e_Y}$.
We have $Z(\overline{e_Y})=\overline{Z( e_Y ) }$. 
The right hand side requires care because of the rule of the involution mentioned in Remark~\ref{rem:involutionsign}.
(We will momentarily see why the rule there is necessary.) 
Namely, if $v_i$ are the basis of $\CH(Y)$, $v_i \otimes \overline{v_j}$ are the basis of $\CH(Y) \otimes \overline{\CH(Y)}$ and the involution acts
as $\overline{v_i \otimes \overline{v_j} } \mapsto (-1)^{{\rm deg}(i){\rm deg}(j)} \overline{v_i} \otimes v_j$.
Then $\overline{Z( e_Y ) }$ in the basis dual to $\overline{v_i} \otimes v_j$
is given by $ \overline{(E_Y)_{i \bar{j}}} (-1)^{ {\rm deg}(i) {\rm deg}(j)}$. % where ${\rm deg}(i)$ is the degree of the vector component labelled by $i$.
On the other hand, Lemma~\ref{lem:anti} states that $\overline{e_Y}=e_{\overline{Y}} \cdot (1_{\overline{Y}} \sqcup (-1)^F)$.
The $(-1)^F$ squares to the identity, and hence it has eigenvalues $\pm 1$. We assume that it is already diagonalized in the basis of the vector spaces we are using.
Then in the same basis as above, the $Z(\overline{e_Y})$ is given as $(E_{\overline{Y} } )_{\bar{i} j} (-1)^{{\rm fp}(i) }$
where $(-1)^{{\rm fp}(i) }$ is the eigenvalue of $(-1)^F$ acting on $v_i$. %the vector component labelled by $i$.
By using the fact that both $\overline{(E_Y)_{i \bar{j}}} $ and $(E_{\overline{Y} } )_{\bar{i} j} $ must be positive definite by unitarity,
and also the fact that ${\rm fp}(i) = {\rm fp}(j) \mod 2$ if $(E_Y)_{i \bar{j}} \neq 0$ (see Sec.~\ref{sec:subtle}),
we get 
\beq
(-1)^{{\rm deg}(i) }=(-1)^{{\rm fp}(i) }, \qquad \overline{(E_Y)_{i \bar{j}}} =(E_{\overline{Y} } )_{\bar{i} j}, \qquad \overline{(C_Y)^{\bar{j} i} } =(C_{\overline{Y} } )^{j \bar{i} }\label{eq:EbarE}
\eeq
where the first equation comes from $\overline{(E_Y)_{i \bar{i}}} (-1)^{{\rm deg}(i) } = (E_{\overline{Y} } )_{\bar{i} i} (-1)^{{\rm fp}(i) }$ and $(E_Y)_{i \bar{i}} > 0$,
the second one from ${\rm deg}(i)={\rm fp}(i) = {\rm fp}(j)={\rm deg}(j)$ and $\overline{(E_Y)_{i \bar{j}}} (-1)^{ {\rm deg}(i) {\rm deg}(j)} = (E_{\overline{Y} } )_{\bar{i} j} (-1)^{{\rm fp}(i) }$,
%two equations come from $\overline{(E_Y)_{i \bar{j}}} (-1)^{{\rm deg}(i) } = (E_{\overline{Y} } )_{\bar{i} j} (-1)^{{\rm fp}(i) }$ and unitarity,
and the third one comes from the fact that $E_{\overline{Y}}$ and $C_{\overline{Y}}$ are inverse matrices of each other.
Now we can see why the rule in Remark~\ref{rem:involutionsign} is necessary. 
If we had not included the factor $(-1)^{{\rm deg}(i) }$, we would have gotten $(-1)^{{\rm fp}(i) }=+1$ which would have contradicted with physical experience
(e.g. results in free fermion theories). 

The equation $(-1)^{{\rm deg}(i) }=(-1)^{{\rm fp}(i) }$ implies the topological spin-statistics theorem.
Namely, the $\BZ_2$-grading of the ${\rm sVect}_\BC$ is determined by the eigenvalues of $(-1)^F$.

\paragraph{The dimension of the Hilbert space.}
Here we study the partition function on a manifold $S^1_{\rm A} \times Y$ where $S^1_{\rm A}$ is the one-dimensional circle which
has the anti-periodic spin structure in the sense of $(-1)^F$.

What we have found above is that $C_Y=Z(c_Y)$ and $C_{\overline{Y}}=Z(c_{\overline{Y}} )$
are complex conjugates of each other if they are regarded as matrices by using explicit basis vectors.
On the other hand, $E_Y=Z(e_Y)$ is a hermitian matrix and hence $C_Y$ 
(which is the inverse matrix of $E_Y$ as shown above) is also hermitian matrix, $\overline{(C_Y)^{\bar{i} j} }=C_Y^{\bar{j} i }$.
Thus we get 
\beq
(C_{\overline{Y}} )^{i \overline{j}} = (C_Y)^{{\overline j} i}.
\eeq
In particular, we obtain 
\beq
\sum_{\bar{j}} (E_Y)_{i \bar{j}} (C_{\overline{Y}})^{ k \bar{j}}   = \sum_{\bar{j}} (E_Y)_{i \bar{j}} (C_Y)^{\bar{j} k} = \delta^i_k. \label{eq:evcev}
\eeq

Now let $X$ be a bordism from $Y$ to $Y$. We can take the trace $\tr Z(Y)$ of $Z(X) =F(X)$ by using the compositions as follows:
\beq
 \tr Z(X)= E_Y \cdot (Z(X) \otimes 1_{\CH(\overline{Y}) } ) \cdot {C}_{\overline{Y}}, \label{eq:traceformula}
\eeq
where we have used \eqref{eq:evcev}. In the bordism, this trace is given as in Figure~\ref{fig:trace}.
\begin{figure}
\centering
\includegraphics[width=.6\textwidth]{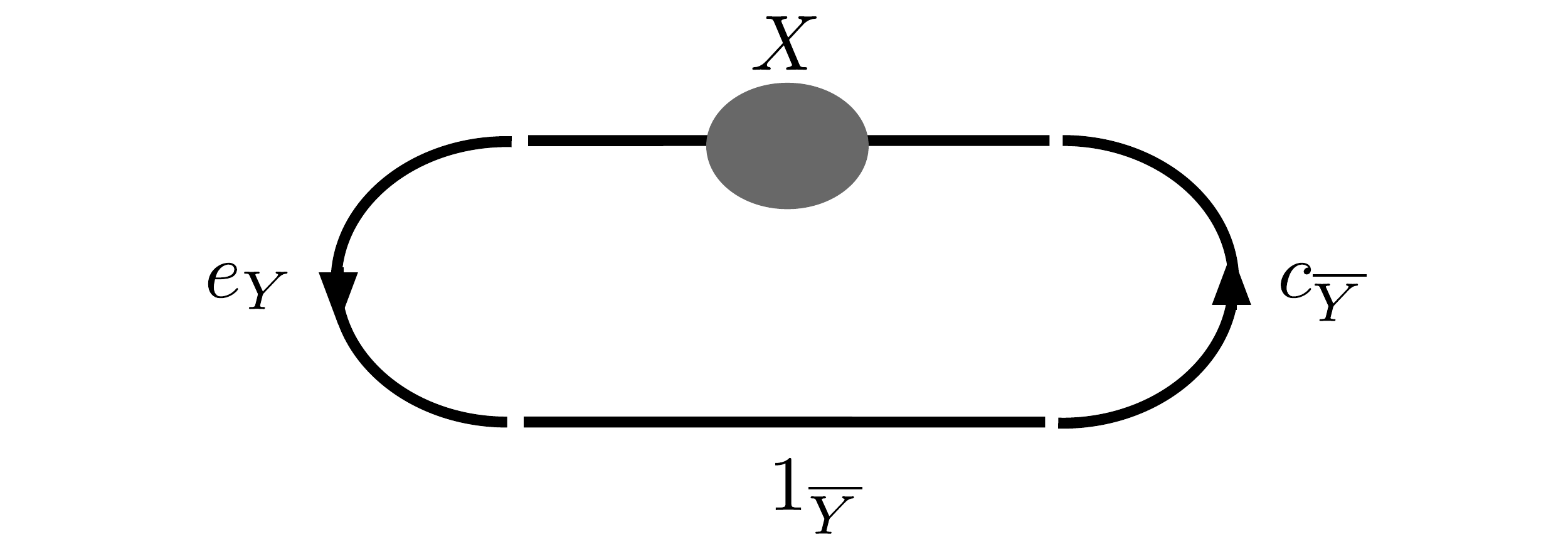}
\caption{The trace \eqref{eq:traceformula}. Notice that we have to use ${c}_{\overline{Y}}$ instead of $c_Y$.
\label{fig:trace} }
\end{figure}

In the special case $X=1_Y$, we get $\tr Z(Y) = \dim \CH(Y)$. On the other hand, the above composition is
the partition function $Z(S^1_{\rm A} \times Y)$, where $S^1_{\rm A}$ is a copy of $S^1$ with the anti-periodic 
spin structure. The appearance of the anti-periodic spin structure is due to the factor of $(-1)^F$ in 
$( (-1)^F \sqcup 1_{\overline{Y}} ) \cdot c_{\overline{Y}}   =  \tau_{\overline{Y},Y} \cdot c_Y$
as shown in Lemma~\ref{lem:anti}. Therefore, we get
\beq
Z(S^1_{\rm A} \times Y) =  \dim \CH(Y). \label{eq:dimformula}
\eeq
This equation is the most significant motivation for including the factor $(-1)^F$ in \eqref{eq:barbarmap}.
The equation \eqref{eq:dimformula} is derived from Lemma~~\ref{lem:anti}, which in turn was based on \eqref{eq:barbarmap}.

\paragraph{Summary of signs.}
There are subtle factors of $(-1)^F$ in the category of bordism ${\rm Bord}_{\langle d-1,d \rangle}$, and $(-1)^{{\rm deg}(v) {\rm deg}(w)}$ in the category of super vector spaces ${\rm sVect}_\BC$.
The reasons why they are necessary have been explained above, but let us summarize them here because they are very subtle.
For simplicity of presentation, we restrict our attention to invertible TQFTs in which the Hilbert spaces are one-dimensional.
\begin{itemize}

\item The $E_Y=Z(e_Y)$ is positive by unitarity. Since $C_Y=Z(c_Y)$ is the inverse of $E_Y$ as shown in \eqref{eq: evcev} (or more explicitly in \eqref{eq:evcev1}), 
$C_Y$ is also positive definite.

\item Applying the positivity for $\overline{Y}$, we conclude that $C_{\overline{Y}}$ is positive definite and hence $E_Y  C_{\overline{Y}}=Z(e_{Y} \cdot c_{\overline{Y}})$ is positive definite.
The question is whether $e_{Y} \cdot c_{\overline{Y}}$ should be $S^1_{\rm A} \times Y$ or $S^1_{\rm P} \times Y$, where $S^1_{\rm A}$ and $S^1_{\rm P}$ are circles with the anti-periodic spin structure
and the periodic spin structure, respectively. It is known in physics that $Z(S^1_{\rm A} \times Y) = \tr_{\CH(Y)} e^{-H}$ where $H$ is the Hamiltonian,
and hence it must always be positive definite. Therefore we require $e_{Y} \cdot c_{\overline{Y}} = S^1_{\rm A} \times Y$.
There are concrete physical examples in which $Z(S_{\rm P} \times Y) = -1$ for some $Y$.

\item The Lemma~\ref{lem:anti} was derived under the assumption that we use \eqref{eq:barbarmap}.
If we had not included $(-1)^F$ in \eqref{eq:barbarmap}, then a careful inspection of the proof of Lemma~\ref{lem:anti}
implies that we would have gotten $ c_{\overline{Y}}   =  \tau_{\overline{Y},Y} \cdot c_Y$ and hence $e_{Y} \cdot c_{\overline{Y}}=S^1_{\rm P} \times Y$
which contradicts with the above requirement. Therefore, we have to use \eqref{eq:barbarmap}.

\item We have $e_Y \cdot \tau_{\overline{Y},Y} \cdot c_Y=S^1_{\rm P} \times Y$. Both $Z(e_Y)$ and $Z(c_Y)$ are positive definite (in the sense
that the hermitian matrices $(E_Y)_{i\bar{j}}$ and $ (C_Y)^{{\overline j} i}$ are positive definite).
%so the sign of $Z(\tau_{\overline{Y},Y})$ must be the same as that of $Z(S^1_{\rm P} \times Y)$.
For invertible TQFTs, $Z(S^1_{\rm P} \times Y)$ is just the value of $(-1)^F$ acting on $\CH(Y)$ (times $Z(S^1_{\rm A} \times Y)=\dim \CH(Y)=1$).
Thus the equality $e_Y \cdot \tau_{\overline{Y},Y} \cdot c_Y=S^1_{\rm P} \times Y$ implies that $Z(\tau_{\overline{Y},Y})$ is given by
$Z(\tau_{\overline{Y},Y}): \overline{v} \otimes v \mapsto (-1)^{{\rm fp}(Y)} v \otimes \overline{v}$ where $(-1)^{{\rm fp}(Y)} =Z(S^1_{\rm P} \times Y)$.
Therefore, we must consider super vector spaces instead of the ordinary vector spaces as stated in Remark~\ref{rem:svect}, and the $\BZ_2$-grading of $\CH(Y)$ is determined by $(-1)^{{\rm fp}(Y)}$.
This is the topological spin-statistics theorem \cite{Freed:2016rqq}.

\item Lemma~\ref{lem:anti} states $\overline{c_{Y}}  =  \tau_{\overline{Y},Y} \cdot c_Y$ and hence $e_Y \cdot \overline{c_{Y}} = S^1_{\rm P} \times Y$.
Therefore, we have $E_Y \overline{C_Y} = Z(S^1_{\rm P} \times Y) = (-1)^{{\rm fp}(Y)}$. However, $C_Y$ was positive definite.
To avoid contradiction between these two facts, we have to include a sign factor in the involution as $\overline{ \overline{v} \otimes v} = (-1)^{{\rm fp}(Y)}  v \otimes \overline{v}$.
Therefore, the rule of the involution mentioned in Remark~\ref{rem:involutionsign} is necessary.
\end{itemize}

%%%%%%%%%%%%%%%%%%%%%%%%%%%%%%%%%%%%%%%%%%%%%%%%%%%%%%%%%%%%%%%%%%%%%%%%%%%%%%%%%%%%%%%%%%%%%%%%%%%%%%%%%%%%%%%%%%%%%%%%%%%%%%%%%%%%%%%%%%%%%%%%%%%%%%%%%%%%%%%%%%%%%%%%%%%%%%%%%%%%%%%%%%%%%%%%%%%%%%%%%%%%%%%%%%%%%%%%%%%%%%%%%%%%%%%%%%%%%%%%
\section{Cobordism invariance of partition functions}\label{sec:cobinv}
%%%%%%%%%%%%%%%%%%%%%%%%%%%%%%%%%%%%%%%%%%%%%%%%%%%%%%%%%%%%%%%%%%%%%%%%%%%%%%%%%%%%%%%%%%%%%%%%%%%%%%%%%%%%%%%%%%%%%%%%%%%%%%%%%%%%%%%%%%%%%%%%%%%%%%%%%%%%%%%%%%%%%%%%%%%%%%%%%%%%%%%%%%%%%%%%%%%%%%%%%%%%%%%%%%%%%%%%%%%%%%%%%%%%%%%%%%%%%%%%

In this section we prove that partition functions of invertible TQFTs are cobordism invariants.
Essentially this was already done by Freed and Moore~\cite{Freed:2004yc} under slightly different axioms. 
We review their proof and supply a little more details to account for 
the $H_d$-structure and unitarity.

\subsection{Sphere partition functions}
The proof of cobordism invariance requires some knowledge of the partition functions of the form $Z(S^\lambda \times S^{d-\lambda})$.
Thus we study them in this subsection.

Let us first notice the following simple fact.
Let $S^n$ be an $n$-dimensional sphere and $D^n$ be an $n$-dimensional ball with boundary $\partial D^n =S^{n-1}$.
Consider $D^{p} \times S^{d-p}$ and $S^{p-1} \times D^{d-p+1}$ for an integer $p$ with $1 \leq p \leq d$. 
The boundaries of both of them are given by $S^{p-1} \times S^{d-p}$,
and hence we can glue them together along the boundaries. (We take the anti-periodic spin structure for $S^1$.)
Then we get a sphere $S^d$. To see this,
consider $S^d$ embedded in $\BR^{d+1}$,
\beq
\sum_{i=1}^{d+1} (x^i)^2 =1,
\eeq
where $x^i$ are the coordinates of $\BR^{d+1}$.
On this $S^d$, we take a submanifold as
\beq
\sum_{i=1}^{p}(x^i)^2=  \sum_{i=p+1}^{d+1}(x^i)^2 =\frac{1}{2}  . \label{eq: divsphere}
\eeq
This submanifold is clearly $S^{p-1} \times S^{d-p}$, and it separates $S^d$ into two pieces given by
\beq
D^{p} \times S^{d-p} &: \sum_{i=1}^{p}(x^i)^2 \leq \frac{1}{2} \leq  \sum_{i=p+1}^{d+1}(x^i)^2, 
\nonumber \\
S^{p-1} \times D^{d-p+1} &:  \sum_{i=1}^{p}(x^i)^2 \geq \frac{1}{2} \geq  \sum_{i=p+1}^{d+1}(x^i)^2.  
\eeq

Notice also that on a sphere $S^d$, there exists a trivial $H_d$-structure in the following sense.
The $S^d$ is orientable (and oriented if $\rho_d(H_d)=\SO(d)$). Then the structure group of the $H_d$-bundle
can be reduced to $(\Spin(d) \times K)/\langle (-1, k_0 ) \rangle$ as in \eqref{eq:orientableH}.
Then we take a trivial $H_d$-structure as the one induced from the homomorphism $\Spin(d) \to (\Spin(d) \times K)/\langle (-1, k_0 ) \rangle$
and using the unique spin structure on $S^d$ (which we take to be anti-periodic for $d=1$). This means that the bundle of the internal symmetry $K$ is trivial.
In the above process, we have picked up an orientation of $S^d$ (if it is not oriented from the beginning),
but spheres with two different orientations are isomorphic by orientation-changing diffeomorphism and we do not need to distinguish them when we consider 
the partition function on them.
So we will not be careful about the orientation.

Using the above facts, we have~\cite{Freed:2004yc}
\begin{lem} \label{lem:sphere1}
The partition function $Z(S^d )$ with the trivial $H_d$-structure on $S^d$ is positive, and in particular
$Z(S^d )=1$ if the dimension $d$ is odd.
\end{lem}
\begin{proof}
Because the $H_d$-structure is trivial, we omit it. First, notice that the $S^d$ can be constructed as the double $\Delta D^d$
of the disk $D^d$, and hence $Z(S^d) =Z(\Delta D^d) \geq 0$ by \eqref{eq:doublepositive}. If this were zero, that would mean 
that $Z(D^d) = 0 $ by unitarity. If this were the case, the partition functions of arbitrary manifolds would be zero because
any manifold $X$ can be constructed by gluing $X \setminus B^\circ $ and $B$, where $B$ is a small closed subspace of $X$ isomorphic to $D^d$
and the superscript $^\circ$ means the interior. Then $Z(X)=Z(X \setminus B^\circ )Z(B)=0$.
However, $Z(S^1_{\rm A} \times Y)=\dim \CH(Y)$ by \eqref{eq:dimformula} and $\dim \CH(Y)=1$ by definition of invertible TQFT, a contradiction.
Therefore, $Z(D^d) \neq 0$ and $Z(S^d) >0$.

Consider the case of odd dimensions $d=2n+1$. We can represent it as
\beq
Z(S^{2n+1} ) = Z({D^{n+1} \times S^n})  Z( S^n \times D^{n+1} ).
\eeq
The boundary of $ S^n \times D^{n+1}$ is $S^n \times S^n$, and we can consider an isomorphism $\phi : S^n \times S^n \to S^n \times S^n$ 
which exchanges the two $S^n$'s (up to an isomorphism which flips the orientation). 
We have the corresponding unitary operator $U(\phi)$ as defined in \eqref{eq:unitaryop}.
Because the Hilbert space is one dimensional, the action of $U(\phi)$ is just a phase, 
$ Z( S^n \times D^{n+1} )= U(\varphi)  Z(  D^{n+1} \times S^n) = e^{i\alpha} Z(  D^{n+1} \times S^n)$ for some phase $e^{i\alpha} \in \U(1)$.
Therefore, we have
\beq
Z(S^{2n+1} ) 
&=e^{i\alpha} Z({D^{n+1} \times S^n})  Z(  D^{n+1} \times S^n) \nonumber \\
&=e^{i\alpha} Z(S^{n+1} \times S^n).
\eeq
But $Z(S^{2n+1} ) >0$ and $Z(S^{n+1} \times S^n) = Z(\Delta( D^{n+1} \times S^n))  \geq 0$,
and hence we get $e^{i\alpha}=1$.

Now, either $n$ or $n+1$ is odd. Let the odd one be denoted as $2\ell+1$. The other one is $d-2\ell-1$. Then, completely in the same way 
as above, we get 
\beq
Z(S^{2\ell +1} \times S^{d-2\ell-1} ) = Z(S^\ell \times S^{\ell +1} \times S^{d-2\ell-1} ) .
\eeq
Repeating this procedure, we eventually get
\beq
Z(S^{2n+1} ) =Z(S^1_{\rm A} \times M_{2n})
\eeq
where $M_{2n}$ is a $2n$-dimensional manifold (which is a product of spheres of various dimensions), and $S^1_{\rm A}$ has the anti-periodic spin structure.
The right hand side is the dimension of the Hilbert space on $M_{2n}$ by \eqref{eq:dimformula}, and this is just 1 by definition of invertible TQFT. 
Therefore we get $Z(S^{2n+1} )=1$.
\end{proof}

If the dimension $d$ is even, then $Z(S^d)$ need not be 1.
The reason behind it is that we can add a local term to the Lagrangian which is proportional to the Euler density
\beq
\CL &\supset a E \\
E &= \frac{1}{(2\pi)^{d/2}2^d} R^{\mu_1\mu_2}_{~~~~~\nu_1\nu_2} \cdots R^{\mu_{d-1}\mu_d}_{~~~~~~~\nu_{d-1}\nu_d}\epsilon_{\mu_1 \cdots \mu_d} \epsilon^{\nu_1\cdots \nu_d}
\eeq
where $R^{\mu\nu}_{~~\rho\sigma}$ is the Riemann curvature tensor, and $a$ is a coefficient. We define the contribution of such an Euler term 
to any compact manifold $X$ (possibly with boundary) to be $\lambda^{{\rm Euler}(X) }$ where $\lambda =e^{-a} \in \BC$ is a constant.
Such a term is consistent with all the axioms of TQFT
as long as $\lambda$ is positive, $\lambda \in \BR_{>0}$.
We allow ourselves to freely add an Euler term with $\lambda>0$ to set $Z(S^d)=1$.
\begin{rem}\label{rem:euler}
In fact, if we are given an arbitrary unitary invertible TQFT $\CI$, we can always factor it uniquely as $\CI = \hat{\CI} \times \CI^{\rm Euler}_\lambda$,
where $\hat{\CI}$ is an invertible TQFT which has the unit sphere partition function $Z(S^d)=1$ for the sphere $S^d$ with the trivial $H_d$-structure, 
and $\CI^{\rm Euler}_\lambda$ is the theory
whose partition function on any manifold is given as $\lambda^{{\rm Euler}(X) }$ for a positive $\lambda>0$. Therefore,
there is no loss of generality to focus on the case that the sphere partition function is unity $Z(S^d)=1$
even if the spacetime dimension $d$ is even.
\end{rem}

\begin{lem}\label{lem:sphere2}
If $Z(S^d)=1$, then $Z(S^\lambda \times S^{d-\lambda})=1$ for any integer $\lambda$ with $0 \leq \lambda \leq d$
where all the spheres are assumed to have the trivial $H_d$-structure.
\end{lem}
\begin{proof}
If either $\lambda$ or $d-\lambda$ is odd, then the proof is completely the same as in the proof of $Z(S^{2n+1})=1$ in Lemma~\ref{lem:sphere1}.
So we assume $\lambda$ is even. The case $\lambda=0$ follows from the assumption itself, so we assume $\lambda $ is an even integer larger than 1.

The $S^d$ is obtained by gluing $D^\lambda \times S^{\lambda-d}$ and $S^{\lambda-1} \times D^{\lambda+1-\lambda}$.
Then we have
\beq
 Z({D^{\lambda} \times S^{d-\lambda}})  Z( S^{\lambda-1} \times D^{d+1-\lambda} )
 =  Z(S^d) = Z ( {S^{\lambda-1} \times D^{d+1-\lambda}} )  Z({D^{\lambda} \times S^{d-\lambda}}) .
\eeq
By combining them 
we get
\beq
Z(S^d)^2=Z(S^{\lambda} \times S^{d-\lambda})Z( S^{\lambda-1} \times S^{d+1-\lambda} ).
\eeq
Now $Z( S^{\lambda-1} \times S^{d+1-\lambda} )=1$ because $\lambda-1$ is odd. Also $Z(S^d)=1$ by assumption.
Therefore we get $Z(S^{\lambda} \times S^{d-\lambda})=1$.
\end{proof}

\subsection{Cobordism invariance}
Now we prove the main theorem of this section which was essentially proved in \cite{Freed:2004yc}.
\begin{thm}\label{thm:inv}
Let $X_0$ and $X_1$ be closed $H_d$-manifolds such that there exists a $(d+1)$-dimensional bordism $C$ from $X_0$ to $X_1$.
Then $Z(X_0)=Z(X_1 )$ automatically for odd $d$ and if the Euler term is chosen such that $Z(S^d)=1$ for even $d$.
\end{thm}
\begin{proof}
The key point of the proof is to decompose the bordism $C$ into elementary building blocks by using Morse theory.
So let us review basic facts from Morse theory~\cite{Milnor}.

On the manifold $C$, there exists a smooth function $f: C \to \BR$ (called Morse function)
with the following properties. On the components of the boundary $\partial C =X_0 \sqcup X_1 $, we have $f|_{X_0}=0$ and $f|_{X_1}=1$, 
and $0 < f < 1$ in the interior $C^\circ$ of $C$.
The $f$ has only finitely many points $\{ p_a \}_{a=1,2,\cdots} $ where 
$df (p_a)=0$, and they are all in the interior of $C$. 
These points are called critical points. In a neighborhood of each critical point $p_a$, there exists
a local coordinate system $x^i~(i=1,\cdots,d+1)$ such that the critical point $p_a$ corresponds to $x^i=0$ and $f$ is given in that neighborhood by 
\beq
f = f(p_a) - \sum_{i=1}^{\lambda_a} (x^i)^2 + \sum_{i = \lambda_a+1}^{d+1} (x^i)^2. 
\eeq
where $\lambda_a$ is an integer. Furthermore, $f(p_a) \neq f(p_b)$ if $p_a \neq p_b$.
The $f(p_a)$ are called critical levels.
See Sec.~2 of \cite{Milnor} for a proof of the existence of $f$.
Essentially, this $f$ is just a generic enough function with the conditions $f|_{X_0}=0$ and $f|_{X_1}=1$, 
and $0 < f|_{C^\circ} < 1$.

Also, there exits a vector field $\xi$ with the following properties (see Sec.~2 of \cite{Milnor}). We have $ \xi \cdot f >0$ except at the critical points $\{ p_a \}$.
In a neighborhood of the critical points, $\xi$ is given by
\beq
\xi =  - \sum_{i=1}^{\lambda_a} x^i \partial_i + \sum_{i = \lambda_a+1}^{d+1} x^i \partial_i  
\eeq
where $\partial_i = \partial/ (\partial x^i)$, and we are using the same coordinate system as above. 
This $\xi$ may be taken to be a gradient vector of $f$ by using some Riemann metric, but that is not necessary.

Now, suppose that $t \in [0,1]$ is not a critical level, i.e., $t \neq f(p_a) $ for any $p_a$. 
Define $X_t = f^{-1}(t)$. Then, on $X_t$ we have $d f \neq 0$ and hence the implicit function theorem tells us that $X_t$
is a submanifold of $C$. See Figure~\ref{fig:morse} for the situation.
This $X_t$ is given the $H_d$-structure by using the upward normal vector to $X_t$. 
The bordism $C$ can be decomposed into $C' = f^{-1}([0,t])$ and $C'' = f^{-1}([t,1]) $
with $\partial C' = X_0 \sqcup X_t$ and $\partial C'' =X_t \sqcup X_1$.
(Up to continuous deformation, we can assume that the $H_{d+1}$-structure of $C$ near $X_t$ is a product type, which we always assume to be the case.)
Conversely, $C$ is constructed from composing the two bordisms $C'$ and $C''$
along $X_t$. By using such decomposition repeatedly, we can decompose $C$ as $C = C_1 C_2 \cdots $
where each of $C_1, C_2, \cdots$ contains at most a single critical point. This is possible because $f(p_a) \neq f(p_b)$ for $p_a \neq p_b$.
Thus we can just consider each piece $C_1, C_2, \cdots$ separately. Therefore, from the beginning,
we can assume that $C$ has only a single critical point $p$ without loss of generality. 

\begin{figure}
\centering
\includegraphics[width=.7\textwidth]{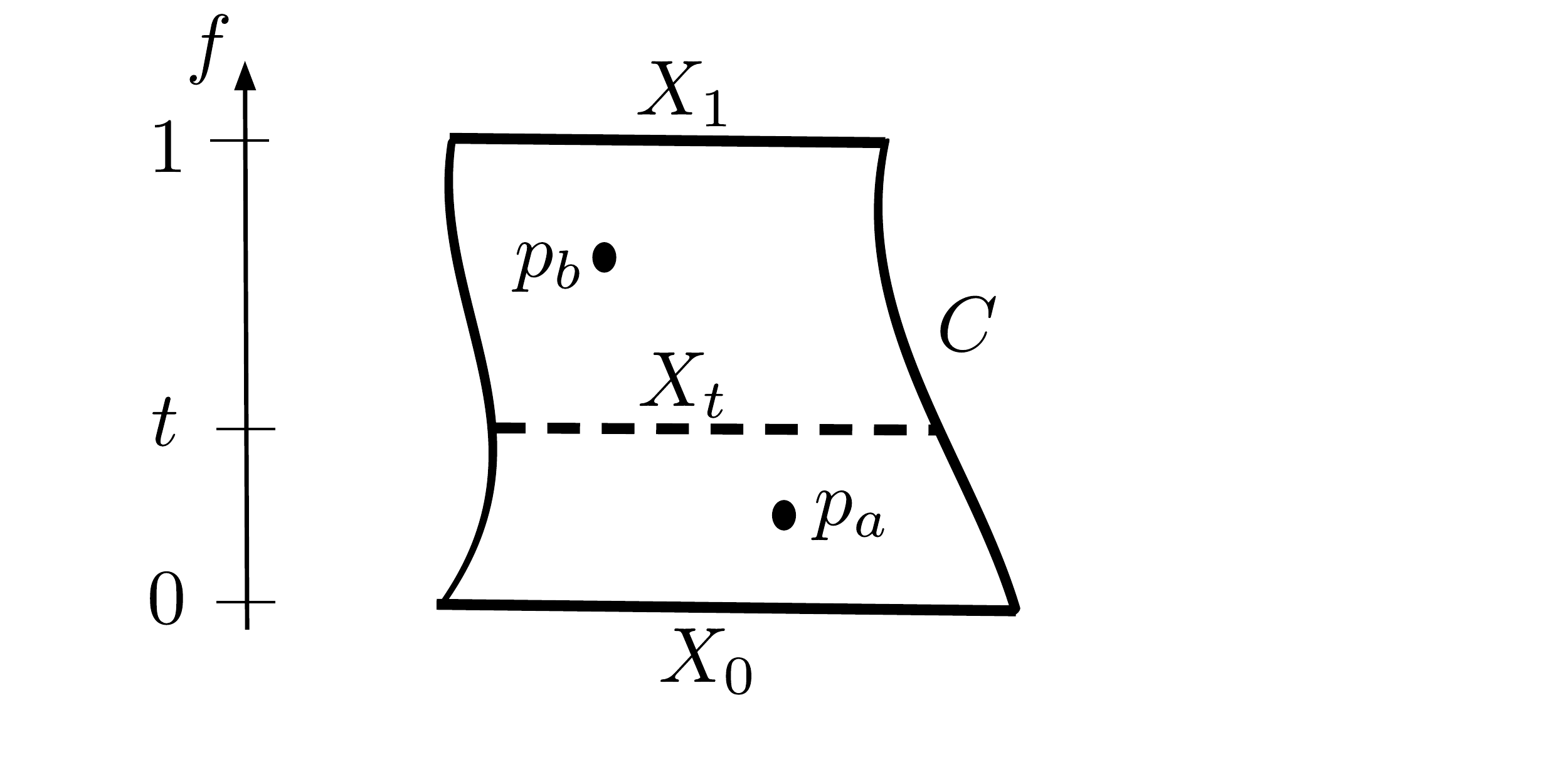}
\caption{The bordism $C$ from $X_0$ to $X_1$. The $C$ is sliced by the values of the function $f$.
Here it is sliced at $f^{-1}(t)$ which is assumed not to contain any critical points $p_a$. \label{fig:morse} }
\end{figure}

Suppose that in a region $f^{-1}( [t_1, t_2 ] )$ for some $t_1$ and $t_2$, there is no critical point. 
Then this region is diffeomorphic to $[t_1, t_2] \times X_{t_1}$. This can be seen as follows.
Define $\hat{\xi} = (\xi \cdot f)^{-1} \xi$ which is well-defined because $\xi \cdot f >0$ in the absence of critical points.
This new vector field $\hat{\xi}$ has the property that $\hat{\xi} \cdot f=1$. Now
we take a point $q \in X_{t_1}$, and solve the differential equation 
\beq
\frac{d \psi^i}{ d t}  = \hat{\xi}^i(\psi(t))  \label{eq:diffeq}
\eeq
with the initial condition $\psi(t=t_1) =q$ for $\psi: [t_1,t_2]  \to C$. Let $\psi (t,q)$ be the unique solution of the equation with the initial condition $q$.
Because $\hat{\xi} \cdot f=1$, we have $f(\psi(t ,q))=t$. Then the map $(t, q) \mapsto \psi(t,q)$
gives the diffeomorphism $ [t_1, t_2] \times X_{t_1} \to f^{-1}([t_1, t_2]) $.
In such a region, $X_t$ for different values of $t \in [t_1, t_2]$ are continuous deformation of each other, and
hence they are identified when they are viewed as $d$-dimensional bordisms from $\varnothing$ to $\varnothing$ 
(see the statement about the homotopy at the last sentence of Definition~\ref{defi:bordism}).
This means that $Z(X_{t_1}) = Z(X_{t_2})$. So we can forget about the region $f^{-1}([t_1, t_2])$ which does not contain any critical point.

Let $v = f(p)$.
From the above considerations, we can restrict our attention to a region $C_\epsilon=f^{-1}( [v - \epsilon, v+\epsilon ] ) $
for arbitrarily small $\epsilon$ with the boundary $\partial C_\epsilon = X_{v+\epsilon} \sqcup X_{v-\epsilon}  $. 
By taking $\epsilon$ small enough, we can cover $C_\epsilon$ by two open 
sets $U$ and $V$ with the following properties. 

The $U$ is a small neighborhood of the critical point $p$. It has the coordinate system in which $f$ and $\xi$ are given by 
\beq
f &= v - \sum_{i=1}^{\lambda} (x^i)^2 + \sum_{i = \lambda+1}^{d+1} (x^i)^2. \label{eq: normalform1} \\
\xi &=  - \sum_{i=1}^{\lambda} x^i \partial_i + \sum_{i = \lambda+1}^{d+1} x^i \partial_i.   \label{eq: normalform2}
\eeq
In the following, we focus on the cases $1 \leq \lambda \leq d$,
but the cases $\lambda=0,d+1$ can be treated much more easily.
In this coordinate system, $U$ is explicitly defined as
\beq
U=  A_\epsilon \cap B_{2\epsilon}
\eeq
where
\beq
A_\epsilon &= \left\{ -\epsilon \leq - \sum_{i=1}^{\lambda} (x^i)^2 + \sum_{i = \lambda+1}^{d+1} (x^i)^2 \leq \epsilon  \right\} \\
B_{2\epsilon} &=  \left\{  \sum_{i=1}^{d} (x^i)^2 < 2\epsilon  \right\}. \label{eq:smallball}
\eeq
The $A_\epsilon$ means the condition $v-\epsilon \leq f \leq v+\epsilon $, while $B_{2\epsilon}$ just means that
$U$ is a small enough neighborhood of $p$. There is no particular meaning to the value $2\epsilon$ used in $B_{2\epsilon}$. 
The following discussions are unchanged as long as it is larger than $\epsilon$.

The solutions of \eqref{eq:diffeq} which flow to or from the critical point $p$ are contained entirely in the region $U$, 
and are given by $D^\lambda_L \cup D^{d+1-\lambda}_R$, where
\beq
D_L^{\lambda} &= \{ x^{\lambda+1}= \cdots = x^{d+1}=0 \},  \\
D_R^{d+1-\lambda} &=  \{ x^{1}= \cdots = x^{\lambda}=0 \} .
\eeq

Next, we define the open set $V$ as $C_{\epsilon} \setminus D_L^{\lambda} \cup D_R^{d+1-\lambda}$.
It is diffeomorphic to $[v - \epsilon, v+\epsilon] \times W$ for an open set $W = X_{v - \epsilon} \setminus ( X_{v - \epsilon} \cap D_L^{\lambda})$ , where the diffeomorphism is
given by the same argument as above by using \eqref{eq:diffeq}.

In particular, let
\beq
S_L^{\lambda-1} &= D_L^{\lambda} \cap X_{v-\epsilon}, \\
S_R^{d-\lambda} & = D_R^{d+1-\lambda} \cap X_{v+\epsilon}.
\eeq
They are diffeomorphic to a sphere.
Then, the $X_{v+\epsilon} \setminus S_R^{d-\lambda}$ is diffeomorphic to $W=X_{v-\epsilon} \setminus S_L^{\lambda-1} $ 
since $C_\epsilon \setminus D_L \cup D_R \cong [v - \epsilon, v+\epsilon] \times W$.
This suggests (which we will discuss more explicitly later) that $X_{v+\epsilon}$ is obtained from $X_{v-\epsilon}$ by the following surgery procedure.
From $X_{v-\epsilon}$, we eliminate a tubular neighborhood of $S_L^{\lambda-1}$. 
This tubular neighborhood is diffeomorphic to $S^{\lambda-1} \times D^{d+1 - \lambda} $ as can be seen by using the above explicit coordinate system.
The manifold obtained by eliminating this tubular neighborhood of $S_L^{\lambda-1}$ from $X_{v-\epsilon}$
has the boundary diffeomorphic to $S^{\lambda-1} \times S^{d - \lambda}$. Then we glue
$D^{\lambda} \times S^{d - \lambda}$ to this boundary. The result is the $X_{v+\epsilon}$, and 
the $D^{\lambda} \times S^{d - \lambda}$ is the tubular neighborhood of $S_R^{d-\lambda}$.

Let us study it in more detail.
The space $X_v=f^{-1}(v)$ for the critical level $v=f(p)$ is singular at $p$ because $df(p)=0$.
We can  avoid this singularity slightly above $p$ or below $p$, where ``above" and ``below" are in terms of the function $f$.
We denote these submanifolds as $X_+$ and $X_-$, respectively. More explicitly, we may take them as follows.
Consider a small disk $D_{\epsilon/2}^{d+1}$ around $p$ which is described by using the above explicit coordinate system on $U$ as
\beq
D_{\epsilon/2}^{d+1} = \left\{  \sum_{i=1}^{d+1} (x^i)^2 \leq \epsilon/2  \right\}.
\eeq
We also denote the boundary as $\partial D_{\epsilon/2}^{d+1} = S_{\epsilon/2}^d$.
The intersection of this sphere $S_{\epsilon/2}^d$ and $f^{-1}(v)$ is given by
\beq
S_{\epsilon/2}^d \cap f^{-1}(v) =  \left\{ \sum_{i=1}^{\lambda} (x^i)^2 = \sum_{i = \lambda+1}^{d+1} (x^i)^2 = \frac{\epsilon}{4}  \right\}
\eeq
which is diffeomorphic to $S^{\lambda-1} \times S^{d-\lambda}$. As explained around \eqref{eq: divsphere},
this intersection divides the sphere $S_{\epsilon/2}^d$ into $X'_+ \cong D^{\lambda} \times S^{d-\lambda}$ and $X'_- \cong S^{\lambda-1} \times D^{d+1-\lambda}$.
On $X'_+$ we have $f \geq v$, while on $X'_-$ we have $f \leq v$. Let us also take
$
X''_v = f^{-1}(v) \setminus (D_{\epsilon/2}^{d+1} )^\circ
$
where $^\circ$ means the interior. Now we have two subspaces $ X'_+ \cup X''_v$ and $X'_- \cup X''_v$ which are piecewise smooth.
By deforming $X'_\pm$ so that they become smooth, we get $X_+$ and $X_-$. 
See Figure~\ref{fig:critical} for the case $d+1=2$ and $\lambda=1$.
\begin{figure}
\centering
\includegraphics[width=.95\textwidth]{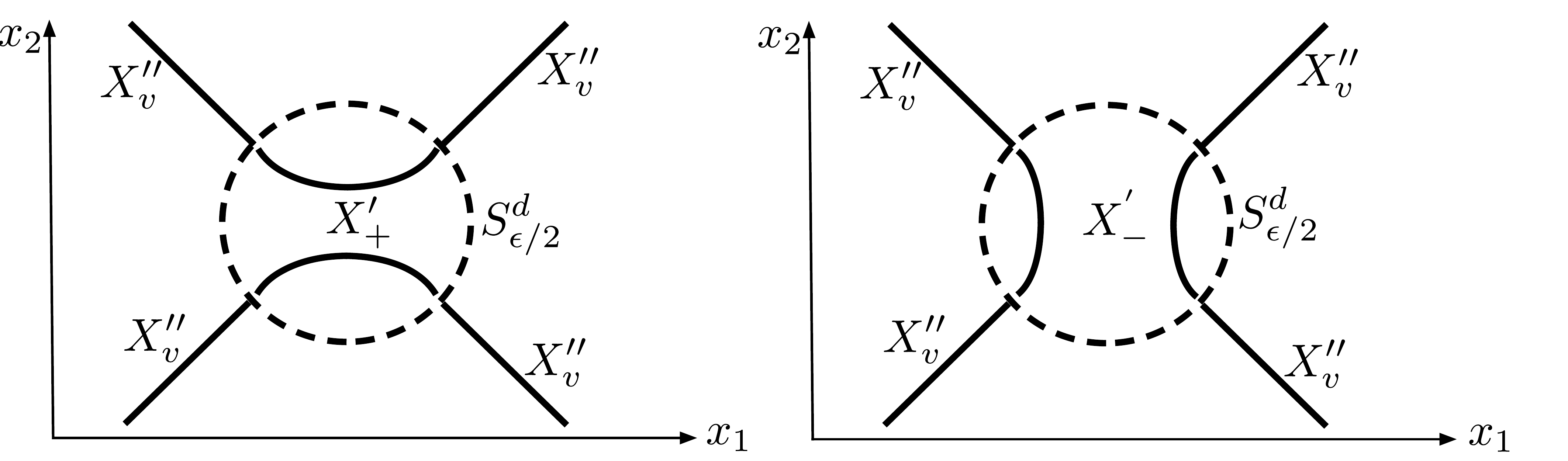}
\caption{
The situation near the critical point for the case $d+1=2$ and $\lambda=1$.
The dashed line represents the sphere $S_{\epsilon/2}^d$. The solid lines outside of $S_{\epsilon/2}^d$
represent $X''_v = f^{-1}(v) \setminus (D_{\epsilon/2}^{d+1} )^\circ$ where $f^{-1}(v)$ is just $x_1 =\pm x_2$.
The solid lines inside the sphere represent $X'_+ \cong D^{1} \times S^{0}$ (Left) and $X'_- \cong S^{0} \times D^{1}$ (Right).
The $X'_\pm$ are taken in such a way that they are smoothly glued to $X''_v$. \label{fig:critical}
}
\end{figure}

By an argument based on the flow equation \eqref{eq:diffeq} with a slight modification, 
the region between $X_{v+\epsilon}$ and $X_{+}$ is diffeomorphic to $[0,1] \times X_{v+\epsilon}$.
In particular, $X_{+}$ is a continuous deformation of $X_{v+\epsilon}$ as an $H_d$-manifold and hence $Z(X_+)=Z(X_{v+\epsilon})$.
In the same way, we get $Z(X_-)=Z(X_{v-\epsilon})$.
 So we only need to study the relation between $X_-$ and $X_+$.

The $X_+$ is obtained from $X_-$ by eliminating $X'_- \cong S^{\lambda-1} \times D^{d+1-\lambda}$ and gluing $X'_+ \cong D^{\lambda} \times S^{d-\lambda}$.
This is the surgery operation mentioned above. This is true including the $H_d$-structure defined on these submanifolds induced from the one on $C$
by using the upward normal vector.

Among the partition functions, we have the relations
\beq
Z(X_+) &= Z(X''_v) Z(X'_+) = Z(X_-) Z(X'_-)^{-1}    Z(X'_+) \nonumber \\
& =Z(X_-) Z( \overline{^tX'_-} \cdot X'_- )^{-1} Z( \overline{^tX'_-} \cdot X'_+  ) .
\eeq
Here, for a given bordism $X: \varnothing \to {Y} $, the notation $\overline{^tX} $ is the one introduced in Remark~\ref{rem:transp}
which is a bordism $ Y \to \varnothing$.
The $X'_\pm$ are both regarded as a bordism of the form $ \varnothing \to Y$.

Now, $\overline{^tX'_-} \cdot X'_+ \cong S^d$ and $\overline{^tX'_-} \cdot X'_- \cong S^{\lambda-1} \times S^{d+1-\lambda}$ as manifolds.
Moreover, the $H_d$-structure on them is trivial by the following reason. 
The submanifolds $X'_\pm$ are all inside the ball $B_{2\epsilon}$ given by \eqref{eq:smallball},
which is contractible to a point. Then the $H_{d+1}$-bundle $P$ can be trivialized in $B_{2\epsilon}$,
and the $H_d$-structure $\varphi_{d+1}: P \to TC$ can be trivialized inside $B_{2\epsilon}$ up to continuous deformation.
Then $X'_\pm$ are just $H_d$-submanifolds of the trivial $H_{d+1}$-manifold. In particular the $K$-part of the 
$(\Spin(d+1) \times K)/\langle (-1,k_0) \rangle$-bundle
is completely trivialized.
The only subtlety is the spin structure when $S^1$ appears.
(i) When $S^1$ appears in $\overline{^tX'_-} \cdot X'_+  \cong S^d$ with $d=1$,
gluing evaluation $e_{Y}$ to $\overline{X'_-} \sqcup X'_+$
on the two dimensional plane in Figure~\ref{fig:critical} shows that it is isomorphic to $S_{\epsilon/2}^d$ which is the boundary of $D_{\epsilon/2}^{d+1}$
and hence it is anti-periodic. 
(Alternatively it can be shown by the argument in the proof of Lemma~\ref{lem:gluing} in the next section.) 
(ii) When $S^1$ appears as the $S^1$ in 
${X'_-} \cong S^1 \times D^{d-1}$ with $\lambda=2$, the $S^1$ is the boundary of the disk $D^2$ which appears in $X'_+=D^2 \times S^{d-2}$.
(iii) When $S^1$ is realized in $\overline{^tX'_-} \cdot X'_-$ from $D^1$ in $X'_- \cong S^{d-1} \times D^1$ with $\lambda=d$, the $S^1$ is constructed 
as a part of the double $\overline{^tX'_-} \cdot X'_-   = \Delta X'_-$ and hence has the anti-periodic spin structure by Remark~\ref{rem:anti}. 
In all the cases, $S^1$ has the anti-periodic spin structure.

By Lemma~\ref{lem:sphere2}, we have 
\beq
Z( \overline{^tX'_-} \cdot X'_- )&=Z(S^{\lambda-1} \times S^{d+1-\lambda})=1 \\
Z( \overline{^tX'_-} \cdot X'_+  ) &= Z(S^d)=1,
\eeq
automatically for odd $d$ and by the appropriate choice of the Euler term for even $d$.
Therefore, we get $Z(X_+) = Z(X_-) $. This completes the proof.
\end{proof}
 
\begin{cor}\label{thm:unit}
The partition function of an invertible TQFT takes values in $\U(1)$ automatically for odd $d$ and if the Euler term is chosen such that $Z(S^d)=1$ for even $d$,
and hence the partition function is an element of ${\rm Hom}(\Omega^H_d , \U(1))$.
\end{cor}
\begin{proof}
By Theorem~\ref{thm:inv}, the partition functions are invariant under bordism if $Z(S^d)=1$.
Now consider a closed $H_d$-manifold $X$ and regard $C = [0,1] \times X$ as a bordism from $\varnothing $ to $X \sqcup \overline{X}$. 
From the invariance under bordism, we have $Z(X)Z(\overline{X})=1$. 
On the other hand, from unitarity of TQFT, we have $Z(\overline{X})=\overline{Z(X)}$. 
Combining them, we get $|Z(X)|=1$.
\end{proof}

\begin{rem}
If we drop the assumption of unitarity, there exist TQFTs whose partition functions are elements of ${\rm Hom}(\Omega^H_d, \BC^\times)$.
However, more nontrivially, there may be non-unitary counterexamples to the claim that the partition function is invariant under bordisms.
Such an example was discussed in Sec.~11 of \cite{Freed:2016rqq}. Consider the case of $d=1$ and $H_{d=1}=\Spin(1)$.
On a point ${\rm pt}_\pm$ (where the subscript means the orientation), we have the Hilbert space $\CH({\rm pt}_\pm)$ with a unique state vector $v_\pm \in \CH({\rm pt}_\pm)$.
If the $\BZ_2$-grading of the vector $v_\pm$ in the super vector space category ${\rm sVect}_\BC$ is assigned in such a way that it does not agree with the action of $(-1)^F$ on $v$,
the topological spin-statistics theorem is violated. Such a violation is allowed in non-unitary theory.
In this case we get $Z(S^1_{\rm A}) = -1$ by using the results of Sec.~\ref{sec:review}. 
However, $S^1_{\rm A} $ is zero as an element of $\Omega^{\Spin}_1$ since it is the boundary of a disk $D^2$.
Therefore, unitarity may be essential for Theorem~\ref{thm:inv}
\end{rem}

%%%%%%%%%%%%%%%%%%%%%%%%%%%%%%%%%%%%%%%%%%%%%%%%%%%%%%%%%%%%%%%%%%%%%%%%%%%%%%%%%%%%%%%%%%%%%%%%%%%%%%%%%%%%%%%%%%%%%%%%%%%%%%%%%%%%%%%%%%%%%%%%%%%%%%%%%%%%%%%%%%%%%%%%%%%%%%%%%%%%%%%%%%%%%%%%%%%%%%%%%%%%%%%%%%%%%%%%%%%%%%%%%%%%%%%%%%%%%%%%
\section{Construction of TQFT from cobordism invariant}\label{sec:explicit}
%%%%%%%%%%%%%%%%%%%%%%%%%%%%%%%%%%%%%%%%%%%%%%%%%%%%%%%%%%%%%%%%%%%%%%%%%%%%%%%%%%%%%%%%%%%%%%%%%%%%%%%%%%%%%%%%%%%%%%%%%%%%%%%%%%%%%%%%%%%%%%%%%%%%%%%%%%%%%%%%%%%%%%%%%%%%%%%%%%%%%%%%%%%%%%%%%%%%%%%%%%%%%%%%%%%%%%%%%%%%%%%%%%%%%%%%%%%%%%%%

In this section, we construct a TQFT from a given cobordism invariant $z$ which takes values in $\U(1)$.
This is a function from the class of closed $H_d$-manifolds  $X$ to $\U(1)$, $z: X \mapsto z(X) \in \U(1)$,
such that $z(X_0) = z(X_1)$ if there exists a $d+1$-dimensional bordism $C$ from $X_0$ to $X_1$.
Two such $H_d$-manifolds $X_0$ and $X_1$ are called bordant. 
The identification of two bordant manifolds gives the bordism group $\Omega_d^{H}$ in Definition~\ref{defi:bor},
on which the group muptiplication is given by the disjoint union $\sqcup$, the unit is given by $\varnothing$, and the inverse of $X$ is given by $\overline{X}$. 
Then $z$ is required to be a homomorphism $z: \Omega_d^{H} \to \U(1)$.
In particular $z(X \sqcup X')=z(X)z(X')$, $z(\varnothing)=1$ and $z(\overline{X} ) = \overline{z(X)}$.

\subsection{Cutting and gluing law and reflection positivity}
First, we prove two fundamental lemmas which are the essential ingredients of the construction of TQFT.
The one is the reflection positivity and the other is the cutting and gluing law. 
The derivations of them (in particular the reflection positivity) have been already sketched in \cite{Witten:2015aba}.

\begin{lem}\label{lem:reflection}
{\rm (Reflection positivity)} 
Let $X$ be a compact $H_d$-manifold with boundary $Y$, regarded as a bordism $\varnothing \to Y$, and 
let $\Delta X $ be the double of $X$. 
Then $\Delta X$ is bordant to the empty $H_d$-manifold and in particular $z( \Delta X)=1$. 
\end{lem}
\begin{proof}
Consider $C'=[0,1] \times X$. This is a manifold with corner, and the boundary consists of 
$X$, $\overline{X}$ and $[0,1] \times  Y$. By making the corner smooth, the $C'$ is deformed to a smooth manifold $C$
whose boundary is isomorphic to the closed $H_d$-manifold obtained by composing 
$X \sqcup \overline{X}$ (regarded as a bordism $\varnothing \to Y \sqcup \overline{Y}$)
and $[0,1] \times Y$ (regarded as a bordism $ Y \sqcup \overline{Y} \to \varnothing $). This is the double $\Delta X$ by definition.
This shows that $\Delta X \cong \partial C$, i.e., it is bordant to the empty manifold, and hence $z(\Delta X)=1$.
\end{proof}

Before going to the next lemma, we define the fermion parity ${\rm fp}(Y) \in \BZ_2$ of a closed $H_{d-1}$-manifold $Y$ as follows.
Consider $S^1_{\rm P} \times Y$ where $S^1_{\rm P}$ is a copy of $S^1$ with the periodic spin structure. 
This manifold has the property that $\overline{S^1_{\rm P} \times Y} \cong S^1_{\rm P} \times Y$
which is shown by ``flipping the direction of $S^1_{\rm P}$''. Then $z(S^1_{\rm P} \times Y)$ must be real, and hence it is $\pm 1$.
We define the fermion parity of $Y$ as
\beq
(-1)^{{\rm fp}(Y)} : = z(S^1_{\rm P} \times Y). \label{eq:fpdef}
\eeq
If two $H_{d-1}$-manifolds $Y_0$ and $Y_1$ are bordant as $X: Y_0 \to Y_1$, then their fermion parities are the same
because $S^1_{\rm P} \times Y_0$ and $S^1_{\rm P} \times Y_1$ are bordant by $S^1_{\rm P} \times X$ and $z$ is a cobordism invariant. 

\begin{lem}\label{lem:gluing}
 {\rm (Cutting and gluing law)} Let $X$ be a compact $H_d$-manifold whose boundary is given by
\beq
\partial X = Y_1 \sqcup Y_2 \sqcup  {Y_3} \sqcup {Y_4},
\eeq
where each of $Y_a~(a=1,2,3,4)$ is isomorphic to a closed $H_{d-1}$-manifold $Y$ when $Y_1, Y_2$ are given the $H_{d-1}$-structures by using the outward normal vector
while $Y_3, Y_4$ are given the $H_{d-1}$-structures by using the inward normal vector.
Let $X_{(13)(24)}$ be the closed $H_d$-manifold obtained by gluing the pair $Y_1 $ and $ {Y_3} $ and the pair 
$Y_2 $ and ${Y_4} $. 
Also, define $X_{(14)(23)}$ in the same way by gluing the pair $(Y_1, {Y_4})$ and the pair 
$(Y_2, {Y_3})$. Then, 
\beq
z(X_{(13)(24)}  ) = (-1)^{{\rm fp}(Y)}  z(X_{(14)(23)}  ).
\eeq
\end{lem}
\begin{proof}
A neighborhood of the boundary of $X$ is isomorphic to
\beq
 (J_1 \sqcup J_2  \sqcup J_3  \sqcup J_4)  \times Y,
\eeq
where $J_a~(a=1,2)$ are isomorphic to $(-\epsilon,0] $ and $J_b~(b=3,4)$ are isomorphic to $[0,\epsilon)$, and the product $H_d$-structure is given 
on  $(J_1 \sqcup J_2  \sqcup J_3  \sqcup J_4)  \times Y$. 
Gluing $J_a~(a=1,2)$ and $~J_b~(b=3,4)$ at $\{0\}$ gives $K_{ab} \cong (-\epsilon,\epsilon) $.

For the moment, let us neglect the issue of the spin structure.
Then a bordism from $K_{13} \sqcup K_{24} $ to $K_{14} \sqcup K_{23}$ is simply given (up to some metric deformation) by 
\beq
C = \left\{ (x,y) \in \BR^2; ~ -1 \leq -x^2+y^2  \leq 1 ,~~x^2+y^2 <2 \right\}
\eeq
where $K_{13} \sqcup K_{24} $ is given by $-x^2+y^2=-1$ and $K_{14} \sqcup K_{23}$ is given by $-x^2+y^2=+1$.
The situation is analogous to that of Figure~\ref{fig:critical} under some minor reinterpretation. Take the Morse function as $f = -x^2+y^2$.
Consider $C \times Y$ and take the union $\cup$ with $[-1,1] \times X$ in the straightforward way, where $[-1,1]$ is parametrized by $f$.
Then we get the bordism from $X_{(13)(24)}   $ to $X_{(14)(23)} $ up to a possible inclusion of $(-1)^F$ in the gluing process.

The only subtlety that can arise in the above argument is the spin structure on $K_{13} \sqcup K_{24} $ and $K_{14} \sqcup K_{23}$
associated to the inclusion of $(-1)^F$. This bordism is independent of $X$, so
we can see what factors of $(-1)^F$ are necessary by taking a simple example of $X$. Consider the case $X=I_1\sqcup I_2$,
where both $I_1$ and $I_2$ are isomorphic to $[0,1]$, and regard $Y_1$ and $Y_4$ as the boundary components of $I_1$
and $Y_2$ and $Y_3$ as the boundary components of $I_2$. 
Then $X_{(13)(24)}$ is isomorphic to $S^1$.
After the above bordism, the $X_{(13)(24)}$ goes to $X_{(14)(34)}$ which is isomorphic to the disjoint union of two circles $S^1 \sqcup S^1$.
Thus this bordism is the ``pair of pants" bordism from $S^1$ to $S^1 \sqcup S^1$; see Figure~\ref{fig:pants}.
In that bordism, there is a topological constraint on the spin structures of the three $S^1$'s.
By representing the periodic and anti-periodic spin structure as $s=\pm 1$ respectively, the three spin structures $s_1$, $s_2$ and $s_3$
associated to the three $S^1$'s are constrained as $s_1s_2s_3=-1$. If
we take the initial spin structure of $X_{(13)(24)}$ to be the periodic spin structure $S^1_{\rm P}$,
the above pair of pants bordism must be $S^1_{\rm P} \to S^1_{\rm P} \sqcup S^1_{\rm A}$
where $S^1_{\rm A}$ is the $S^1$ with the anti-periodic spin structure.

\begin{figure}
\centering
\includegraphics[width=.5\textwidth]{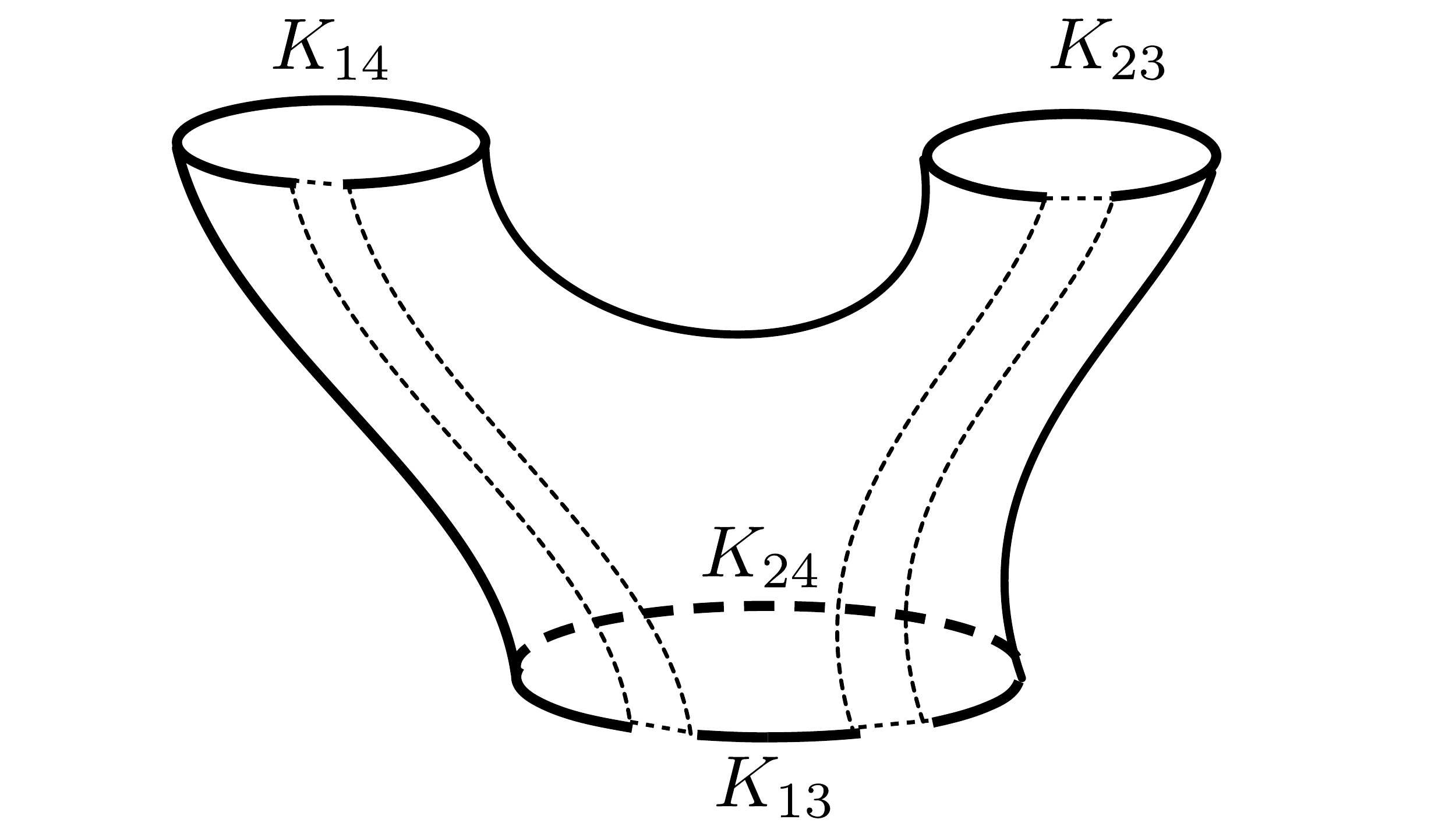}
\caption{
The pair of pants bordism from $S^1$ to $S^1 \sqcup S^1$. By removing the small strips inside the regions bounded by the thin dotted lines,
we get the bordism from $K_{13} \sqcup K_{24} $ to $K_{14} \sqcup K_{23}$.
\label{fig:pants}}
\end{figure}

Therefore, we get the following conclusion. In the bordism from $K_{13} \sqcup K_{24} $ to $K_{14} \sqcup K_{23}$,
if $K_{13} \sqcup K_{24} $ is given the trivial spin structures, then one of $K_{14}$ or $K_{23}$ (let's say $K_{23}$) must have
additional $(-1)^F$. Therefore, 
$X_{(13)(24)}  $ is bordant to $X'_{(14)(23)}  $, where $X'_{(14)(23)} $ is obtained from $X $ in almost the same way as $X_{(14)(23)} $
except for the fact that we include the action of $(-1)^F$ when $Y_2$ and $Y_3$ are glued. 

Let $X_{(14)}$ be the manifold which is obtained from $X$ by gluing $Y_1$ and $Y_4$ only.
Then $X_{(14)} \sqcup ([0,1] \times Y)$ has the boundary isomorphic to the disjoint union of four copies of $Y$.
By using the same kind of bordism as above applied to $X_{(14)} \sqcup ([0,1] \times Y)$, 
we can show that $X'_{(14)(23)} $ is bordant to $X_{(14)(23)} \sqcup (S^1_{\rm P} \times Y )$.
Therefore, we conclude that $X_{(13)(24)}$ is bordant to $X_{(14)(23)} \sqcup (S^1_{\rm P} \times Y)$ and hence 
\beq
z(X_{(13)(24)}  ) =   z(S^1_{\rm P} \times Y)z(X_{(14)(23)}  ).
\eeq
By definition, $(-1)^{{\rm fp}(Y)}  = z(S^1_{\rm P} \times Y)$. This completes the proof.
\end{proof}

The cutting and gluing law means the following.
Given a closed $H_d$-manifold $X$, we cut it at two places along submanifolds which are both isomorphic to $Y$.
Then, we glue them together in the different way from the start. The cobordism invariant is unchanged under this manipulation 
up to the sign factor $ (-1)^{{\rm fp}(Y)} $.  See Figure~\ref{fig:cut} for the situation.

\begin{figure}
\centering
\includegraphics[width=0.9\textwidth]{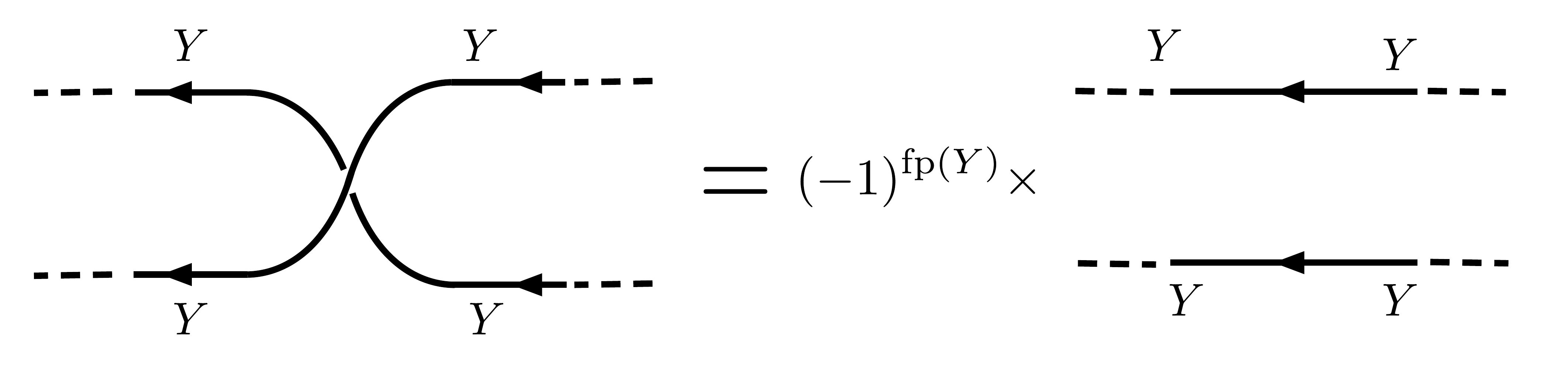}
\caption{
The implication of cutting and gluing law. The dashed lines are connected to the rest of the manifold.
The cobordism invariant $z$ is unchanged under the
process up to the sign $ (-1)^{{\rm fp}(Y)}$.
\label{fig:cut}}
\end{figure}

Roughly speaking, the cutting and gluing law corresponds to locality and the reflection positivity corresponds to unitarity.
However, the actual situation might be a bit more subtle. If we drop the reflection positivity, the 
construction of a TQFT given below fails at some points even if we do not care about unitarity. %The author does not know whether this is a fundamental problem or not.
%As a comparison, we remark that the Osterwalder-Schrader theorem \cite{Osterwalder:1973dx,Osterwalder:1974tc} which construct
%quantum field theories from Euclidean correlation functions also requires reflection positivity in the construction, probably in an essential way.

\subsection{Reconstruction theorem}
Now we construct TQFT. 
%Our theorem might be compared with
%the Osterwalder-Schrader reconstruction theorem~\cite{Osterwalder:1973dx,Osterwalder:1974tc}. This theorem states that
%we can construct QFT satisfying Wightman axioms~\cite{Streater:1989vi} from Euclidean correlation functions.
%Our theorem may be seen as a reconstruction theorem for invertible TQFT.
%Instead of correlation functions, we reconstruct invertible TQFT from the partition function satisfying the cutting and gluing law and reflection positivity.
%In the case of generic QFT, the functional derivatives of the partition function with respect to background field give correlation functions,
%so our reconstruction theorem is conceptually not too different from that of Osterwalder-Schrader.
%However, we only do it for the simplest case of invertible TQFT, although on general $H_d$-manifolds.
For simplicity, we assume the following. 
Consider the $d-1$-dimensional bordism group $\Omega^{H}_{d-1}$ which is one dimension lower than the $\Omega_d^H$ relevant for the partition function.
We assume that this is finitely generated. 
Then, we can pick up a (not canonical) isomorphism 
\beq
\Omega^{H}_{d-1}  \cong (\BZ)^k  \oplus  \BZ_{p_{k+1}} \oplus  \cdots \oplus \BZ_{p_{k+\ell}} ,
\eeq
and also pick up reference $H_{d-1}$-manifolds  
\beq
&Y^{\rm ref}_a~(a=1,\cdots,k)  \\
&Y^{\rm ref}_b ~(b=k+1,\cdots,k+\ell)
\eeq
which represent the generators of $\Omega^{ H}_{d-1}$.
The $Y^{\rm ref}_a~(a=1,\cdots,k)$ generate the free part $(\BZ)^k $ and the $Y^{\rm ref}_b ~(b=k+1,\cdots,k+\ell)$
generate the torsion part $\BZ_{p_{k+1}} \oplus  \cdots \oplus \BZ_{p_{k+\ell}}$. We call them as elementary reference manifolds.
The disjoint union of $p_b$ copies of $Y^{\rm ref}_b$ for $k+1 \leq b \leq k+\ell$ is bordant to empty, and we pick up a bordism 
$X_b^{\rm ref}$ from $\varnothing$ to the disjoint union of $p_b$ copies of $Y^{\rm ref}_b$.

Furthermore, for each element $w \in \Omega^{H}_{d-1}$, we pick up an explicit reference manifold $Y^{\rm ref}_w$
which consists of the disjoint union of copies of $Y^{\rm ref}_a, \overline{Y^{\rm ref}_a}~(1 \leq a \leq k) $ and $Y^{\rm ref}_b ~(b=k+1,\cdots,k+\ell)$ with
the following rules;
(i) Only the minimum number of them appears, i.e., $Y^{\rm ref}_b~(b=k+1,\cdots,k+\ell)$ only appears $w_{b}$ times with $0 \leq w_b \leq p_b-1$,
and only either $Y^{\rm ref}_a$ or $\overline{Y^{\rm ref}_a}$ appears for $1 \leq a \leq k$.
(ii) We introduce the ordering
\beq
Y^{\rm ref}_1 < \overline{Y^{\rm ref}_1} < Y^{\rm ref}_2 < \cdots < \overline{Y^{\rm ref}_k} < Y^{\rm ref}_{k+1}< \cdots < Y^{\rm ref}_{k+\ell} , \label{eq:ordering}
\eeq
and then oder the elementary reference manifolds in $Y^{\rm ref}_w$ according to the order.

The reason that we need to care about the ordering is the sign factor $(-1)^{{\rm fp}(Y)}$ in the cutting and gluing law as in Figure~\ref{fig:cut}.
Due to this sign factor, it is convenient to consider the following sign rule (Koszul sign rule) which will repeatedly appear in the proof of Theorem~\ref{thm:construction}.
Let $Y_1, \cdots, Y_m$ be $H_{d-1}$-manifolds, and let $\tau \in \mathfrak{S}_m$ be an element of the symmetric group of permutation of $m$ objects.
We can consider a bordism, which by abuse of notation we denote by $\tau$, as follows.
As a manifold, it is just $[0,1] \times (Y_1 \sqcup \cdots \sqcup Y_m) $. However, we take the boundary isomorphism at $\{1\}$ 
such that we regard it as a bordism,
\beq
\tau : Y_1 \sqcup \cdots \sqcup Y_m \to Y_{\tau(1)} \sqcup \cdots \sqcup Y_{\tau(m)}. \label{eq:permB}
\eeq
Now, the $\tau$ can be represented as composition of transpositions of the form $\tau_{A,B}$ in Figure~\ref{fig:exchange}. 
Let $n(\tau)$ be the number of times
transposition of two manifolds $A$ and $B$, both of which have odd fermion parity $(-1)^{{\rm fp}(A) }= (-1)^{{\rm fp} (B)  }  = -1$, appear.
The sign factor $(-1)^{n(\tau)}$ is independent of how we represent $\tau$ as the composition of transpositions, and
is determined only by $\tau$ and hence it is well-defined. This is the Koszul sign.
In particular, if $Y_1, \cdots, Y_m$ are all the same manifolds, then the effect of the bordism $\tau$ in $z$
is equivalent to $(-1)^{n(\tau)}$ times the identity bordism
due to the cutting and gluing law shown in Figure~\ref{fig:cut}.

Before going to Theorem~\ref{thm:construction},
let us also recall the following notation from section~\ref{sec:review}.
We have defined the evaluation $e_Y: Y \sqcup \overline{Y} \to \varnothing$ and the coevalutation $c_Y: \varnothing \to \overline{Y} \sqcup Y$ in Definition~\ref{defi:evcoev}.
Furthermore, we have ${c}_{\overline{Y}}: \varnothing \to Y \sqcup \overline{Y} $.
Notice that the composition $e_Y \cdot c_{\overline{Y}} $ is $S^1_A \times Y$ which coincides with the double $\Delta c_{\overline{Y}} $,
and hence 
\beq
z(e_Y \cdot c_{\overline{Y}} ) =1.
\eeq
More generally, if $X$ is a bordism from $Y_0$ to $Y_1$, then the composition
$e_{Y_1} \cdot (X \sqcup \overline{X}) \cdot {c}_{\overline{Y_0}}$ is actually identified with the double $\Delta( (X \sqcup 1_{\overline{Y_0}}) \cdot c_{\overline{Y_0}} )$,
where $(X \sqcup 1_{\overline{Y_0}}) \cdot c_{\overline{Y_0}}$ is a bordism $\varnothing \to Y_1 \sqcup \overline{Y_0}$.
See Figure~\ref{fig:XXbar} for the situation. Thus we have 
\beq
z(e_{Y_1} \cdot (X \sqcup \overline{X}) \cdot {c}_{\overline{Y_0}})=1. \label{eq:refl0}
\eeq
This is the consequence of the reflection positivity.
\begin{figure}
\centering
\includegraphics[width=.95\textwidth]{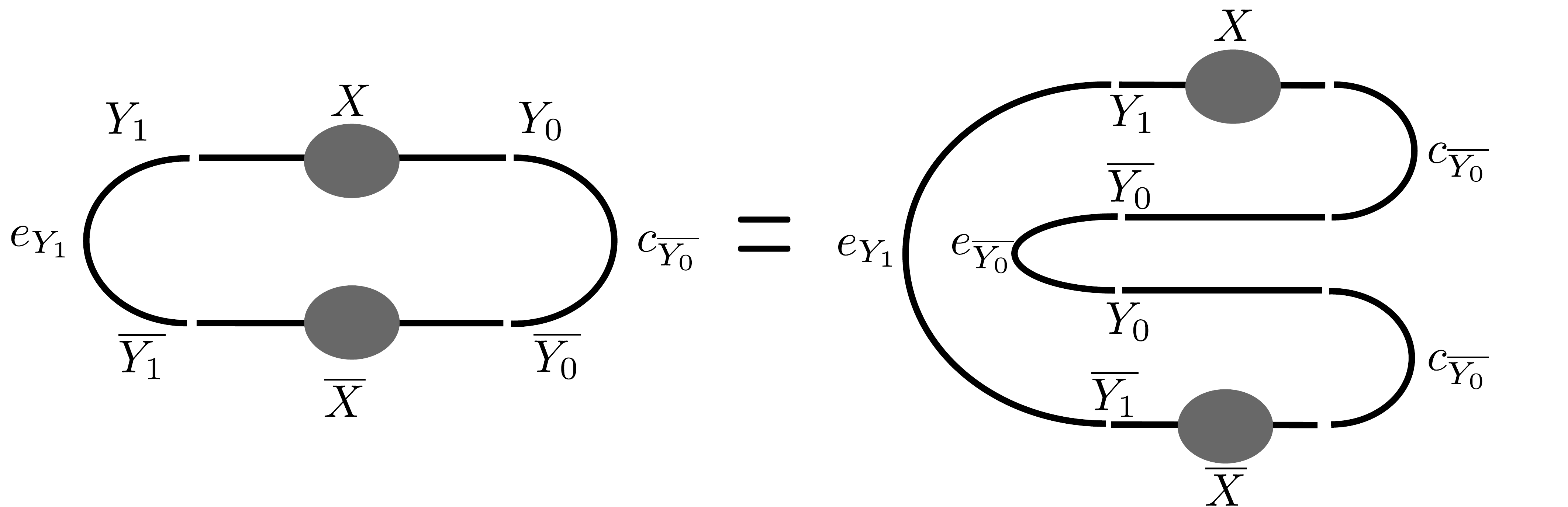}
\caption{The composition $e_{Y_1} \cdot (X \sqcup \overline{X}) \cdot {c}_{\overline{Y_0}}$ (left)
and the double $\Delta( (X \sqcup 1_{\overline{Y_0}}) \cdot c_{\overline{Y_0}} )$ (right). 
They are isomorphic. We have tacitly used Lemma~\ref{lem:anti} to represent $\Delta( (X \sqcup 1_{\overline{Y_0}}) \cdot c_{\overline{Y_0}} )$
as the right figure.  \label{fig:XXbar}}
\end{figure}

\begin{thm}\label{thm:construction}
Given a cobordism invariant $z  \in {\rm Hom}(\Omega^H_d, \U(1))$, there exists an invertible TQFT whose partition function is given by $z$,
assuming that $\Omega^{H}_{d-1}$ is finitely generated.
\end{thm}
\begin{proof}

The proof involves several steps.
\paragraph{Step 1: Definition of vector spaces.} Our first task is to assign a one-dimensional Hilbert space $\CH(Y)$ to each closed $H_{d-1}$-manifold $Y$.
Consider an arbitrary $Y$. If $Y$ represents the element $w \in \Omega^H_{d-1}$,
then it is bordant to the reference manifold $Y^{\rm ref}_w$.
There is a set of morphisms in the bordism category,
${\rm Hom}(Y^{\rm ref}_w, Y)$, which consists of bordisms from $Y^{\rm ref}_w$ to $Y$. 
Let $\CH^{\rm big}(Y)$ be the infinite dimensional vector space whose element is a function from ${\rm Hom}(Y^{\rm ref}_w, Y)$ to $\BC$ 
such that the function is nonzero for only finitely many elements of ${\rm Hom}(Y^{\rm ref}_w, Y)$. The obvious vector space structure is given to this space.
The $\CH^{\rm big}(Y)$ is spanned by basis vectors $\ket{X}$ for each $X \in {\rm Hom}(Y^{\rm ref}_w, Y)$, 
\beq
\CH^{\rm big}(Y) = \bigoplus_{X \in  {\rm Hom}(Y^{\rm ref}_w, Y)} \{ \ket{X} \},
\eeq
where $\ket{X}$ is a function which is 1 on $X$ and zero for others. 

Now, we introduce an equivalence relation between the vectors in $\CH^{\rm big}(Y)$ as follows. Let $X$ and $X'$ be two bordisms.
Then, we have a bordism $X' \sqcup \overline{X} : Y^{\rm ref}_w \sqcup \overline{Y^{\rm ref}_w} \to Y \sqcup \overline{Y}$ and hence
we can obtain a closed manifold as $e_Y \cdot (X' \sqcup \overline{X} ) \cdot {c}_{\overline{Y^{\rm ref}_w}}$ which we denote as $\tr (\overline{X} X')$.
Then we define the equivalence relation as 
\beq
\ket{X'} \sim z(\tr (\overline{X} X')) \ket{X}.
\eeq
We can check the following: 

(i) Reflexivity 
\beq
\ket{X} \sim \ket{X}
\eeq
is a consequence of $z(\tr (\overline{X} X))=1$ which is \eqref{eq:refl0}. 

(ii) Symmetry 
\beq
\ket{X} \sim z(\tr (\overline{X} X'))^{-1} \ket{X'}
\eeq
is a consequence of the fact that $\overline{\tr (\overline{X} X')} = \tr (\overline{X'} X)$
and hence $z(\tr (\overline{X} X'))^{-1} = z(\tr (\overline{X'} X))$. 

(iii) Transitivity 
\beq
\ket{X''} \sim z(\tr (\overline{X'} X'')) z(\tr (\overline{X} X')) \ket{X}
\eeq
is shown by
using the cutting and gluing law (Lemma~\ref{lem:gluing}) as follows. By applying the cutting and gluing law twice as in
Figure~\ref{fig:transitivity}, we get
\beq
z(\tr (\overline{X'} X'')) z(\tr (\overline{X} X')) = (-1)^{{\rm fp}(Y) } (-1)^{{\rm fp}(Y^{\rm ref}_w  ) }  z(\tr (\overline{X} X'')) z(\tr (\overline{X'} X')) . \label{eq:trans}
\eeq
As discussed above, we have $z(\tr (\overline{X'} X'))=1$ by reflection positivity. On the other hand,
the fermion parities of two bordant manifolds, $(-1)^{{\rm fp}(Y) }$ and $(-1)^{{\rm fp}(Y^{\rm ref}_w  ) } $, are the same as shown below \eqref{eq:fpdef}.
Thus we get 
$
z(\tr (\overline{X'} X'')) z(\tr (\overline{X} X')) =  z(\tr (\overline{X} X''))  .
$
\begin{figure}
\centering
\includegraphics[width=1.05\textwidth]{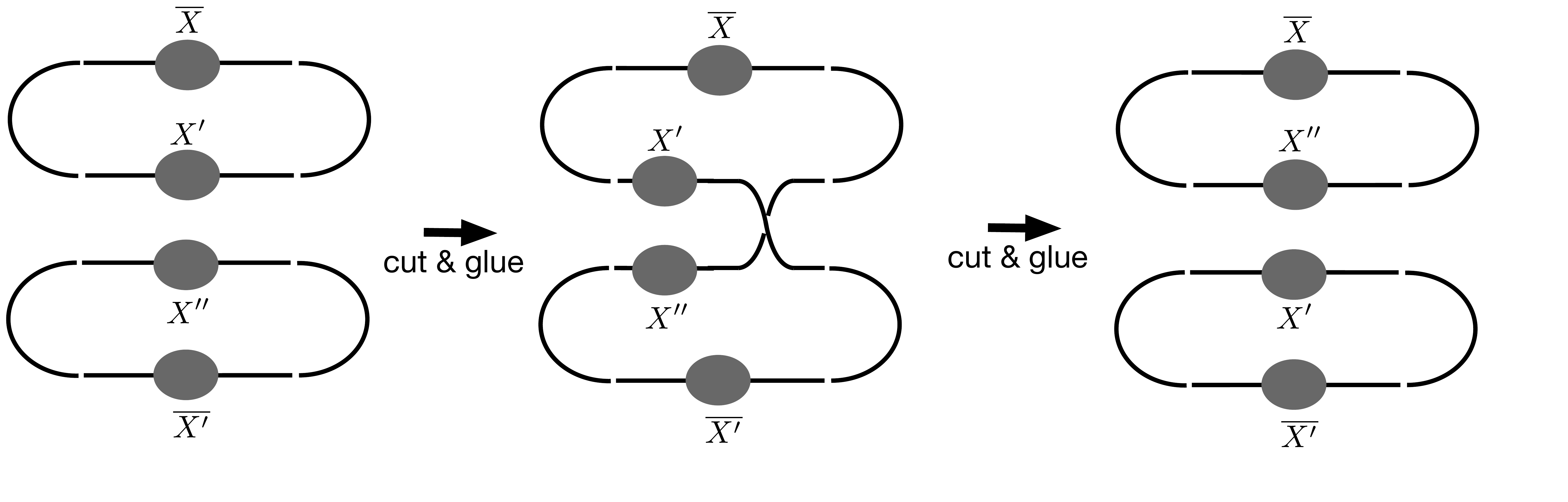}
\caption{ Using the cutting and gluing twice gives the equation \eqref{eq:trans}. \label{fig:transitivity}}
\end{figure}

Therefore, the above definition of $\sim$ gives a well-defined equivalence relation on the vector space $\CH^{\rm big}(Y)$.
Then we define 
\beq
\CH(Y) :=\CH^{\rm big}(Y)/ \sim. 
\eeq
It is clear that this $\CH(Y)$ is a one dimensional vector space. 
We can consistently introduce a hermitian metric on $\CH(Y)$ by imposing that the image of $\ket{X}$ 
under the projection $\CH^{\rm big}(Y) \to \CH(Y)$  has the unit norm.
We will later check that this hermitian metric coincides with the one defined by using the evaluation $e_Y$ as in Definition~\ref{defi:unitary}.

As a special case, if $Y$ is empty, $Y = \varnothing$, then there is a distinguished element in ${\rm Hom}(\varnothing, \varnothing)$
which is the empty bordism which we also denote by $\varnothing$. By using the vector $\ket{\varnothing}$, we can canonically 
identify $\CH(\varnothing)$ with $\BC$. Under this identification, for a closed manifold $X$ we have $\ket{X} \sim z(X)$.
Also, we can canonically identify $\CH(Y \sqcup \overline{Y})$ for arbitrary $Y$ with $\BC$ by using the vector $\ket{c_{\overline{Y}}}$ for the distinguished bordism 
$c_{\overline{Y}} \in {\rm Hom}(\varnothing,Y \sqcup \overline{Y}) $. 
So we have 
\beq
\CH(\varnothing) \cong \BC \cong \CH(Y \sqcup \overline{Y}).
\eeq

\paragraph{Step 2: Identification of vector spaces under products.} 
For the disjoint union of two manifolds $Y$ and $Y'$, we can identify $\CH(Y \sqcup Y')$ and $\CH(Y) \otimes \CH(Y')$ as follows.
We take $X : Y^{\rm ref}_{w} \to Y$ and $X' : Y^{\rm ref}_{w'} \to Y'$. 
Then we have the bordism 
\beq
X \sqcup X': Y^{\rm ref}_{w} \sqcup Y^{\rm ref}_{w'} \to Y \sqcup Y'.
\eeq
Now we take a bordism $X_{w, w'}$ from $Y^{\rm ref}_{w+w'}$ to $Y^{\rm ref}_{w} \sqcup Y^{\rm ref}_{w'}$ as follows.
First, we use ${c}_{\overline{Y^{\rm ref}_a}}$ and $X_b^{\rm ref}$ (tensored with the identity bordism of $Y^{\rm ref}_{w+w'}$)
to get a bordism from $Y^{\rm ref}_{w+w'}$ to the manifold of the form
\beq
\tilde{Y}^{\rm ref} _{w+w'} =
Y^{\rm ref}_{w+w'} \sqcup (Y^{\rm ref}_{a_1} \sqcup \overline{Y^{\rm ref}_{a_1}  }) \sqcup \cdots \sqcup (Y^{\rm ref}_{b_1} \sqcup \cdots \sqcup Y^{\rm ref}_{b_1}) \sqcup \cdots \label{eq:creation}
\eeq
in such a way that $\tilde{Y}^{\rm ref} _{w+w'}$ contains the same number of 
$Y^{\rm ref}_a, \overline{Y^{\rm ref}_a}~(1 \leq a \leq k) $ and $Y^{\rm ref}_b ~(b=k+1,\cdots,k+\ell)$ as 
$Y^{\rm ref}_{w} \sqcup Y^{\rm ref}_{w'}$. Namely, $\tilde{Y}^{\rm ref} _{w+w'}$ and $Y^{\rm ref}_{w} \sqcup Y^{\rm ref}_{w'}$ are the same
up to permutation of their component elementary reference manifolds. See the left of Figure~\ref{fig:creation}.
Next, take
a bordism $\tau$ which permutes the elementary reference manifolds appearing in $\tilde{Y}^{\rm ref} _{w+w'}$
to the order as in $Y^{\rm ref}_{w} \sqcup Y^{\rm ref}_{w'}$. See \eqref{eq:permB}.
Such a permutation is not unique. 
Pick up one of them, and let $(-1)^{n(\tau)}$ be the Koszul sign associated to $\tau$.
By composing the above two bordisms $Y^{\rm ref} _{w+w'} \to \tilde{Y}^{\rm ref} _{w+w'} $ and $\tau$,
we get a bordism $X_{w, w'}(\tau)$ from $Y^{\rm ref}_{w+w'}$ to $Y^{\rm ref}_{w} \sqcup Y^{\rm ref}_{w'}$ which depends on $\tau$.
See the right of Figure~\ref{fig:creation}.
\begin{figure}
\centering
\includegraphics[width=1.05\textwidth]{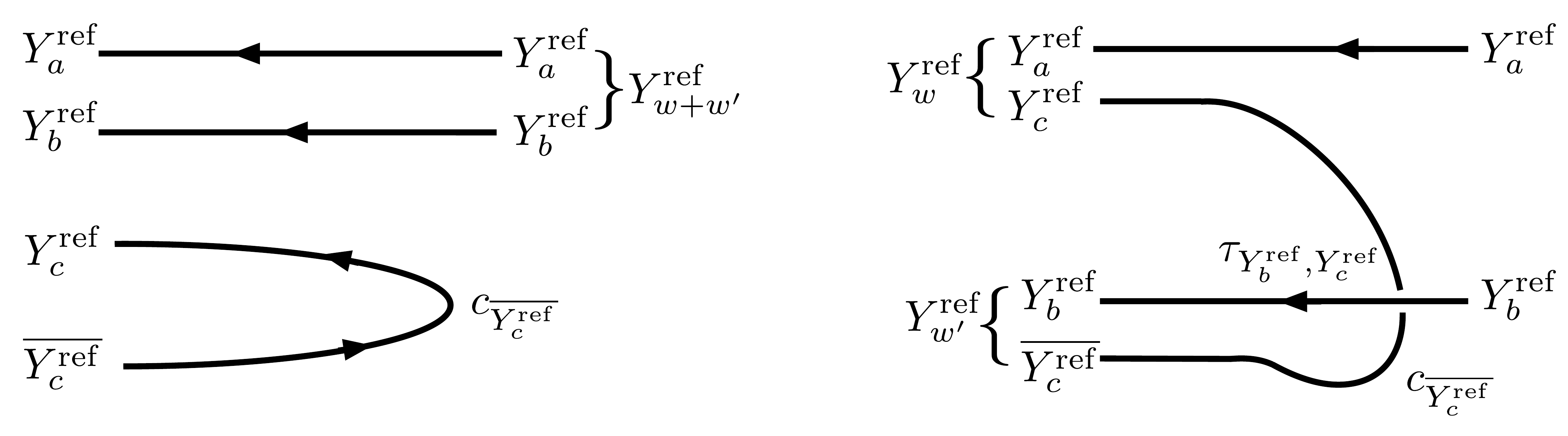}
\caption{ Left: the bordism from ${Y}^{\rm ref} _{w+w'}$ to $\tilde{Y}^{\rm ref} _{w+w'}$. The figure shows an example of the case
$Y^{\rm ref}_{w} =Y^{\rm ref}_{a} \sqcup Y^{\rm ref}_{c}$, $Y^{\rm ref}_{w'} =Y^{\rm ref}_{b} \sqcup \overline{Y^{\rm ref}_{c}}$ and hence 
$Y^{\rm ref}_{w+w'} = Y^{\rm ref}_{a} \sqcup Y^{\rm ref}_{b}$.
Right: the bordism from ${Y}^{\rm ref} _{w+w'}$ to $Y^{\rm ref}_{w} \sqcup Y^{\rm ref}_{w'}$. 
This requires the transposition of $Y^{\rm ref}_{b}$ and $Y^{\rm ref}_{c}$ by $\tau_{Y^{\rm ref}_{b}, Y^{\rm ref}_{c}}$.
\label{fig:creation}}
\end{figure}

Now we consider the following identification. We have a bordism $(X \sqcup X') \cdot X_{w, w'}(\tau)$ from $Y^{\rm ref}_{w+w'}$
to $Y \sqcup Y'$. Then we define the map from $\CH(Y) \otimes \CH(Y')$ to $\CH(Y \sqcup Y')$ as
\beq
 \ket{X} \otimes \ket{X'} \mapsto (-1)^{n(\tau)} \ket{(X \sqcup X') \cdot X_{w, w'}(\tau)}. \label{eq:symmfunctor}
\eeq
This is independent of $\tau$ as follows. Let $\tau'$ be another permutation. Then, $  \tau' \cdot \tau^{-1}$ only permutes
elementary reference manifolds of the same type, i.e., it consists of composition of transpositions such as $\tau_{Y^{\rm ref}_a , Y^{\rm ref}_a }$.  
Then, $\tau$ and $\tau'$ can be changed to each other by the cutting and gluing law up to a sign factor which is exactly $ (-1)^{n(\tau') - n(\tau)}$.
This shows that the above map is independent of the choice of $\tau$.
Also, this map is independent of the choice of $X$ and $X'$ due to the cutting and gluing lemma
and the reflection positivity.
So we get a well-defined map $\CH(Y) \otimes \CH(Y') \to \CH(Y \sqcup Y')$ independent of $\tau$, $X$ and $X'$.

Furthermore, it can be checked that the two chains of maps
\beq
& \CH(Y) \otimes \CH(Y')  \otimes \CH(Y'')   \to \CH(Y \sqcup Y') \otimes \CH(Y'') \to  \CH(Y \sqcup Y'  \sqcup Y''),     \nonumber \\
&\CH(Y) \otimes \CH(Y')  \otimes \CH(Y'')    \to \CH(Y) \otimes \CH ( Y' \sqcup Y'') \to \CH(Y \sqcup Y'  \sqcup Y''),
\eeq
gives the same result which is the associativity condition.
The check of this associativity is done as follows. 
We need to compare two bordisms of the form
\beq
&(X_{w,w'}(\tau_1) \sqcup 1_{Y^{\rm ref}_{w''}} ) \cdot X_{w+w' , w''}(\tau_2) , \label{eq:adj1} \\
&( 1_{Y^{\rm ref}_{w}}  \sqcup X_{w',w''}(\tau_3) ) \cdot X_{w , w' + w''}(\tau_4).\label{eq:adj2}
\eeq
See Figure~\ref{fig:associative}.
\begin{figure}
\centering
\includegraphics[width=1.0\textwidth]{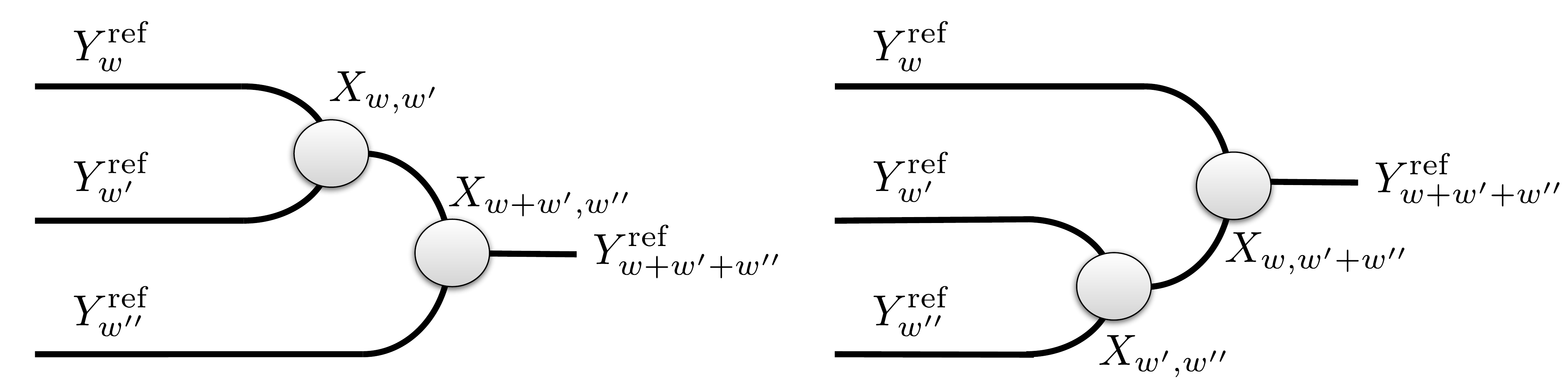}
\caption{  The bordism \eqref{eq:adj1} (Left) and the bordism \eqref{eq:adj2} (Right). \label{fig:associative}}
\end{figure}
Both of them are bordisms from $Y^{\rm ref}_{w+w'+w''}$ to $Y^{\rm ref}_{w} \sqcup Y^{\rm ref}_{w' } \sqcup Y^{\rm ref}_{w''}$.
They both consist of ${c}_{\overline{Y^{\rm ref}_a}}$, $X_b^{\rm ref}$ and $\tau_{A,B}$.
More precisely, ${c}_{\overline{Y^{\rm ref}_a}}$ and $X_b^{\rm ref}$ may appear in a form like
\beq
1_{Y'_{\rm ref}} \sqcup {c}_{\overline{Y^{\rm ref}_{a_1}}} \sqcup \cdots \sqcup X^{\rm ref}_{b_1} \sqcup 1_{Y''_{\rm ref}} \label{eq:CR1}
\eeq
where $Y'_{\rm ref}$ and $Y''_{\rm ref}$ are some disjoint union of elementary reference manifolds.
See the left of Figure~\ref{fig:compare1}.
We can make the above bordism into a bordism
of the form
\beq
\tau_0 \cdot (1_{Y'_{\rm ref} \sqcup Y''_{\rm ref}} \sqcup {c}_{\overline{Y^{\rm ref}_{a_1}}} \sqcup \cdots \sqcup X^{\rm ref}_{b_1} )  \label{eq:CR2}
\eeq
where $\tau_0$ is some permutation. See the right of Figure~\ref{fig:compare1}.
\begin{figure}
\centering
\includegraphics[width=.8\textwidth]{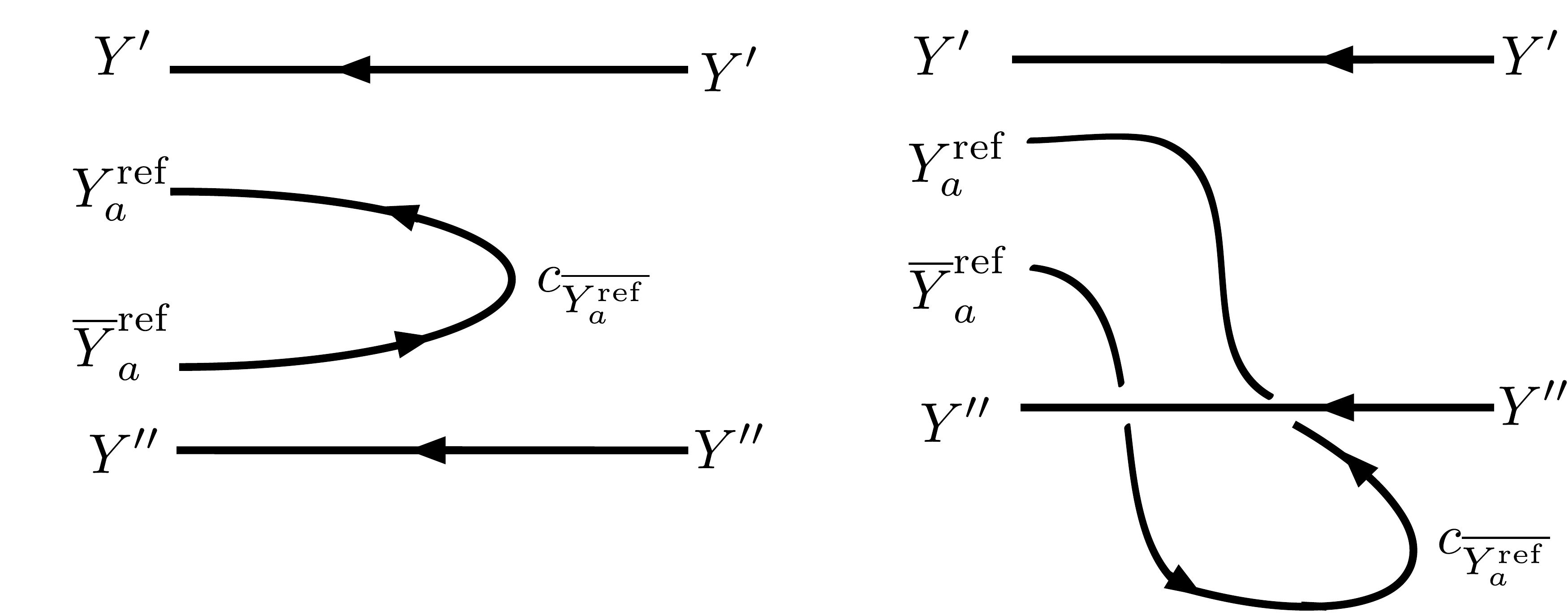}
\caption{  The bordism of the form \eqref{eq:CR1} (Left) is equivalent to the bordism of the form \eqref{eq:CR2} (Right).  
This manipulation makes it possible to put the bordisms in the form \eqref{eq:norderedf}. \label{fig:compare1}}
\end{figure}
By using this fact,
it is always possible to rewrite the above bordisms from $Y^{\rm ref}_{w+w'+w''}$ to $Y^{\rm ref}_{w} \sqcup Y^{\rm ref}_{w' } \sqcup Y^{\rm ref}_{w''}$
in the form 
\beq
\tau \cdot (1_{Y^{\rm ref}_{w+w'+w''}} \sqcup \CC), \label{eq:norderedf}
\eeq
where $\tau$ is a permutation of elementary reference manifolds
and $\CC$ is given by a disjoint union of copies of ${c}_{\overline{Y^{\rm ref}_a}}$ and $X_b^{\rm ref}$.
By further using permutation if necessary, we can assume that this $\CC$ is the same in \eqref{eq:adj1} and \eqref{eq:adj2}.
The part $\tau$ in \eqref{eq:adj1} and \eqref{eq:adj2} differs only by permutation among the same $Y^{\rm ref}_a$ (or $\overline{Y^{\rm ref}_a}$)
 as in the proof that \eqref{eq:symmfunctor} is independent of $\tau$. Also notice that the change from \eqref{eq:CR1} to \eqref{eq:CR2}
 can be done without any Koszul sign, because the boundaries of ${c}_{\overline{Y^{\rm ref}_{a_1}}}$ and $X^{\rm ref}_{b_1} $
 have even total fermion parity.
 Therefore, as in the proof that \eqref{eq:symmfunctor} is independent of $\tau$,
 the two bordisms can be made to be the same by the cutting and gluing. 
 Thus the associativity holds.

The associativity is the condition mentioned around \eqref{eq:SMCcondition1} which is needed for a monoidal functor.
In this way, we can identify $\CH(Y) \otimes \CH(Y') \cong \CH(Y \sqcup Y')$ in a consistent way.

\paragraph{Step 3: Identification of vector spaces under involution.} 
The two identifications $\CH(Y \sqcup \overline{Y}) \cong \BC$ and $\CH(Y \sqcup \overline{Y}) \cong \CH(Y) \otimes \CH(\overline{Y}) $
imply that $\CH(\overline{Y})$ is naturally the dual vector space to $\CH(Y)$.
In addition to this, we have introduced the hermitian metric on $\CH(Y)$ such that the image of $\ket{X}$ for any $X : Y^{\rm ref}_w \to Y$
has the unit norm. 
If we use this hermitian metric, we have the following identifications;
\beq
\overline{\CH(Y)} \cong \CH(Y)^* \cong \CH(\overline{Y}) ,
\eeq
where $\CH(Y)^*$ is the dual vector space to $\CH(Y)$. 

More explicitly, let $X: Y^{\rm ref}_{w} \to Y$ and $\hat{X} : Y^{\rm ref}_{-w} \to \overline{Y}$ be bordisms.
Then, $\ket{X} \otimes \ket{\hat{X}} \in \CH(Y) \otimes \CH( \overline{Y})$ is identified under $\CH(Y) \otimes \CH(\overline{Y}) \cong \CH(Y \sqcup \overline{Y}) \cong \BC$ with 
\beq
 \xi(X, \hat{X}) :=(-1)^{n(\tau_0)}z( e_Y \cdot ( X \sqcup \hat{X}) \cdot X_{w, -w}(\tau_0) ) \in \BC,
\eeq
where $X_{w, -w}(\tau_0)$ is the bordism from $\varnothing$ to $Y^{\rm ref}_{w} \sqcup Y^{\rm ref}_{-w}$ as constructed above.
Therefore, $\ket{\hat{X} } $ is identified with 
\beq
\ket{\hat{X} } \cong \overline{\ket{X}}  \otimes \ket{X} \otimes \ket{\hat{X}} \cong  \xi(X, \hat{X})   \overline{\ket{X}} 
\eeq
where we used $\overline{\ket{X}}  \otimes \ket{X} \cong 1$ under the hermitian metric introduced above.

With the above consideration in mind, we define the map $\CH(\overline{Y}) \to \overline{\CH(Y)}$ as
\beq
\ket{\hat{X}} &\mapsto \xi(X, \hat{X}) \overline{\ket{X}}  \label{eq:involutionmap}
\eeq
where $\overline{\ket{X}} \in \overline{\CH(Y)}$ is the complex conjugate of $\ket{X}$.
This map is independent of $\tau_0$, $X$ and $\hat{X}$ due to the cutting and gluing.

Under the above map, let us check that the following diagram
\beq 
\xymatrix{
\CH (\overline{Y})    \otimes \CH (\overline{Y'})  \ar[d]_{}  \ar[r]^{}  & \overline{  \CH (Y) } \otimes \overline{ \CH (Y') }  \ar[d]^{} \\
\CH (\overline{ Y \sqcup Y'})   \ar[r]_{} &   \overline{\CH ( Y \sqcup Y')} 
} \label{eq:involutioncheck}
\eeq
commutes. This means the commutativity of
\beq
\xymatrix{
\ket{ \hat{X} }   \otimes   \ket{ \hat{X}' } \ar[d]_{}  \ar[r]^{}  
&   \xi(X, \hat{X})  \xi(X', \hat{X'}) \cdot  \overline{\ket{X}} \otimes \overline{\ket{X'}} \ar[d]^{}   \\
(-1)^{n( \hat{\tau} )} \ket{   (\hat{X} \sqcup \hat{X}' ) \cdot X_{-w, -w'}(\hat{\tau}) }  \ar[r]_{} 
& (\text{some factor}) \overline{\ket{   ( X \sqcup X' ) \cdot X_{w, w'}( \tau ) } }
} 
\eeq
The commutativity reduces to the equality
\beq
&(-1)^{n( \hat{\tau} )}\xi(( X \sqcup X' ) \cdot X_{w, w'}( \tau ),  (\hat{X} \sqcup \hat{X}' ) \cdot X_{-w, -w'}(\hat{\tau}))  \nonumber \\
=&(-1)^{n( \tau )+{\rm fp}(Y){\rm fp}(Y')}  \xi(X, \hat{X})  \xi(X', \hat{X'}),
\eeq
where the factor of $(-1)^{{\rm fp}(Y){\rm fp}(Y')}$ in the right hand side appeared 
because of the rule of involution $\overline{v \otimes w} \mapsto {(-1)^{{\rm deg}(v) {\rm deg}(w)}} \overline{v} \otimes \overline{w}$ 
as mentioned in Remark~\ref{rem:involutionsign}, if we assign the $\BZ_2$ grading of the super vector spaces according to ${\rm fp}(Y) \in \BZ_2$.
The check of this equality involves comparing two bordisms appearing in Figure~\ref{fig:compare2}. 
\begin{figure}
\centering
\includegraphics[width=1.1\textwidth]{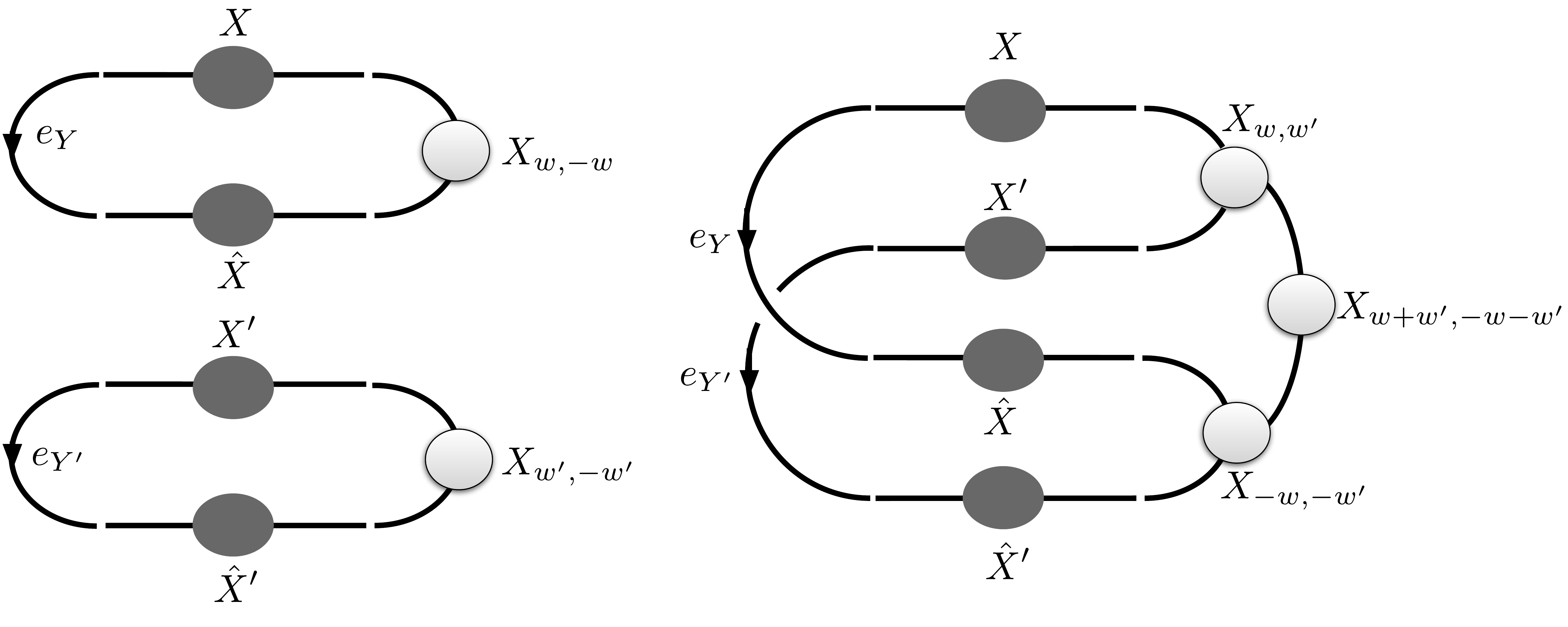}
\caption{ The check of the commutativity \eqref{eq:involutioncheck} requires comparing the left and right bordisms.
\label{fig:compare2}}
\end{figure}
One is
\beq
X_{w, -w}(\tau_0) \sqcup X_{w', -w'}(\hat{\tau}_0) \label{eq:cre1}
\eeq
which is a bordism from $\varnothing$ to $Y^{\rm ref}_{w} \sqcup  Y^{\rm ref}_{-w} \sqcup Y^{\rm ref}_{w'} \sqcup Y^{\rm ref}_{-w'}$ as in the left of Figure~\ref{fig:compare2}. 
The other is 
\beq
(X_{w, w'}( \tau ) \sqcup X_{-w, -w'}( \hat{\tau} ) ) \cdot X_{w+w', -w-w'}(\tau'_0) \label{eq:cre2}
\eeq
which is a bordism from $\varnothing$ to $Y^{\rm ref}_{w} \sqcup Y^{\rm ref}_{w'} \sqcup Y^{\rm ref}_{-w} \sqcup Y^{\rm ref}_{-w'}$
as in the right of Figure~\ref{fig:compare2}.

Both of the bordisms \eqref{eq:cre1} and \eqref{eq:cre2} are constructed as follows.
We create elementary reference manifolds as many as needed by using ${c}_{\overline{Y^{\rm ref}_a}}$ and $X_b^{\rm ref}$,
and then permute them to the desired orders. 
After acting $\tau_{Y^{\rm ref}_{w'} , Y^{\rm ref}_{-w}}$ to \eqref{eq:cre2} so that 
both of them become bordisms from $\varnothing$ to $Y^{\rm ref}_{w} \sqcup  Y^{\rm ref}_{-w} \sqcup Y^{\rm ref}_{w'} \sqcup Y^{\rm ref}_{-w'}$,
these two bordisms are the same up to permutation among the same types of elementary reference manifolds.
Therefore, they give the same values as long as we include the appropriate Koszul sign factors.
We have included almost all the Koszul sign factors $(-1)^{n( {\tau_0} )} $, $(-1)^{n( \hat{\tau}_0 )} $,
$(-1)^{n( {\tau'_0} )} $, $(-1)^{n( {\tau} )} $ and $(-1)^{n( \hat{\tau} )} $.
A nontrivial point is that the Koszul sign associated to $\tau_{Y^{\rm ref}_{w'} , Y^{\rm ref}_{-w}}$ is provided by $(-1)^{{\rm fp}(Y){\rm fp}(Y')}$ 
which appeared because of Remark~\ref{rem:involutionsign}. Thus, we indeed have all the necessary Koszul signs.
In this way the commutativity of \eqref{eq:involutioncheck} holds.
This corresponds to \eqref{eq:involutioncondition}.

We can also check that the diagram
\beq 
\xymatrix{
\CH ( Y )    \ar[rd]_{}  \ar[r]^{ }  &    \CH (\overline{\overline{Y}} )   \ar[d]^{} \\
&   \overline{ \overline{ \CH ( Y) }} 
} 
\label{eq:invinv}
\eeq
commutes, where the horizontal arrow uses $Y \cong \overline{\overline{Y}}$ on the bordism category,
the diagonal arrow uses $\CH ( Y )  \cong \overline{ \overline{ \CH ( Y) }} $ in the category of vector spaces,
and the vertical arrow uses $\CH(\overline{Y}) \to \overline{\CH(Y)}$ twice.
The commutativity is equivalent to 
$
\xi ( \hat{X},  X) \overline{\xi(X, \hat{X})}=1,
$
or more explicitly,
\beq
(-1)^{n(\tau'_0)}z( e_{\overline{Y}} \cdot (  \hat{X} \sqcup  X ) \cdot X_{-w, w}(\tau'_0) )  (-1)^{n(\tau_0)}
\overline{z( e_Y \cdot ( X \sqcup \hat{X}) \cdot X_{w, -w}(\tau_0) ) } =1. \label{eq:zzbar=1}
\eeq
Here, $X_{-w, w}(\tau'_0)$ is a bordism from $\varnothing$ to $Y^{\rm ref}_{-w} \sqcup Y^{\rm ref}_{w}$.
We can take it to be 
\beq
X_{-w, w}(\tau'_0) = \tau_{Y^{\rm ref}_{w} , Y^{\rm ref}_{-w}} \cdot X_{w, -w}(\tau_0) ,~~~~~\tau'_0 =   \tau_{Y^{\rm ref}_{w} , Y^{\rm ref}_{-w}}  \cdot \tau_0.
\eeq
In particular, we have $(-1)^{n(\tau'_0) }= (-1)^{n(\tau_0) + n(\tau_{Y^{\rm ref}_{w}, Y^{\rm ref}_{-w}    }) } =(-1)^{n(\tau_0) + {\rm fp}(Y) }$.

We want to compare $e_{\overline{Y}} \cdot (  \hat{X} \sqcup  X ) \cdot \tau_{Y^{\rm ref}_{w}, Y^{\rm ref}_{-w} } $ 
and $e_Y \cdot ( X \sqcup \hat{X}) $ to check the above equality.
Because $\tau_{\overline{Y},Y} \cdot (  \hat{X} \sqcup  X ) \cdot \tau_{Y^{\rm ref}_{w}, Y^{\rm ref}_{-w} } =( X \sqcup \hat{X})$,
we need to compare $ e_{\overline{Y}} $ and $e_Y \cdot \tau_{\overline{Y},Y}$.
Now, $e_{\overline{Y}}$ has an additional insertion of $(-1)^F$ compared to $e_Y \cdot \tau_{\overline{Y},Y}$ by Lemma~\ref{lem:anti}.
Thus, $e_{\overline{Y}} \cdot (  \hat{X} \sqcup \cdot X ) \cdot \tau_{Y^{\rm ref}_{w}, Y^{\rm ref}_{-w} } $ is isomorphic to 
$e_Y \cdot ( X \sqcup \hat{X}) $ except for the $(-1)^F$ which produces the factor $(-1)^{{\rm fp}(Y) }$.
This proves \eqref{eq:zzbar=1} and hence the diagram \eqref{eq:invinv} commutes.
This corresponds to \eqref{eq:twiceinvo}.

\paragraph{Step 4: Construction of a functor.} 
Our next task is to assign a linear map $Z(X) : \CH(Y_0) \to \CH(Y_1)$ to each bordism $X: Y_0 \to Y_1$.
This is simply done as follows. Let $X_0$ be a bordism $X_0: Y^{\rm ref}_{w} \to Y_0$, 
where $w \in \Omega^{H}_{d-1}$ represents the bordism class of both $Y_0$ and $Y_1$. 
Then the composition $X \cdot X_0$ is a bordism from $Y^{\rm ref}_{w}$ to $Y_1$. 
Then we simply set
\beq
Z(X): \ket{X_0} \mapsto \ket{X \cdot X_0}.
\eeq
By the cutting and gluing and reflection positivity, this does not depend on the choice of $X_0$, and hence this definition
gives a well-defined map from $\CH(Y_0)$ to $\CH(Y_1)$.
Under the composition of two bordisms $X$ and $X'$, we have $Z(X' \cdot X) = Z(X') Z(X)$.
Also, for the identity bordism $1_Y = [0,1] \times Y$ we have $Z(1_Y) = 1_{\CH(Y)}$.
Therefore, it is clear that we get a functor $F: {\rm Bord}_{\langle d-1, d \rangle}(H) \to {\rm sVect}_{\BC}$
by setting $F(Y) = \CH(Y)$ and $F(X) = Z(X)$. 

The fact that $F$ is a monoidal functor is a consequence of the fact that $\CH(\varnothing) \cong \BC$
and $F(Y) \otimes F(Y') \xrightarrow{\sim} F( Y \sqcup Y')$ which satisfy the associativity condition (and other more trivial conditions).
The fact that this identification gives the equality of $Z( X \sqcup X') $ and $Z(X ) \otimes Z(X')$
is immediate from the definition of $Z(X)$. Thus \eqref{eq:productID} is confirmed.
We have already checked the condition about the associativity \eqref{eq:SMCcondition1}.

Let us check that $F$ is symmetric. Let $\tau_{Y, Y'}$ be the bordism from $Y \sqcup Y'$ to $Y' \sqcup Y$ given 
by Figure~\ref{fig:exchange}. 
Let $X_{w,w'}(\tau)$ be the bordism from $Y^{\rm ref}_{w+w'}$ to $Y^{\rm ref}_{w} \sqcup Y^{\rm ref}_{w'}$ introduced above.
The identification $  \CH(Y) \otimes (Y') \xrightarrow{\sim}   \CH(Y \sqcup Y') $ was done by \eqref{eq:symmfunctor}.
Take also the bordism $X_{w', w}(\tau')$ from $Y^{\rm ref}_{w+w'}$ to $Y^{\rm ref}_{w'} \sqcup Y^{\rm ref}_{w}$.
More explicitly, this bordism is taken as 
\beq
X_{w', w}(\tau') := \tau_{Y^{\rm ref}_{w}, Y^{\rm ref}_{w'}    } \cdot X_{w,w'}(\tau),~~~~~\tau' =   \tau_{Y^{\rm ref}_{w} , Y^{\rm ref}_{-w}} \cdot \tau  .
\eeq
We have $(-1)^{n(\tau')} =(-1)^{n(\tau)}  (-1)^{{\rm fp}(Y) {\rm fp}(Y') }$.
Then one can check that the following diagram commutes:
\beq
\xymatrix{
 \ket{X} \otimes \ket{X'} \ar[d]_{}  \ar[r]^{} & (-1)^{{\rm fp}(Y) {\rm fp}(Y') }   \ket{X'} \otimes \ket{X}  \ar[d]^{} \\
 (-1)^{n(\tau)} \ket{(X \sqcup X') \cdot X_{w, w'}(\tau)}   \ar[r]_{Z(\tau_{Y, Y'})} &  (-1)^{n(\tau)} \ket{\cdot (X' \sqcup X) \cdot X_{w', w}(\tau')}
}
\eeq
This proves the symmetry of $F$ given in \eqref{eq:SMCcondition2}
if we assign the $\BZ_2$-grading of the super vector spaces according to ${\rm fp}(Y) $. 
Therefore, $F$ is a symmetric monoidal functor.

We have also defined $F(\overline{Y}) \xrightarrow{\sim} \overline{F(Y)}$ in \eqref{eq:involutionmap}. 
The fact that $F(\overline{X}) =\overline{F(X)}$ under this identification is seen as follows.
Let us take $X_0 :Y^{\rm ref}_w \to Y_0$ and $\hat{X}_0: Y^{\rm ref}_{-w} \to \overline{Y_ 0}$ and consider  \eqref{eq:involutionmap} 
\beq
\ket{ \hat{X}_0 } &\mapsto \xi(X_0,  \hat{X}_0 ) \overline{\ket{X_0}}  .
\eeq
By using $X :Y_0 \to Y_1$, we also have
\beq
\ket{ \overline{X} \cdot \hat{X}_0 } &\mapsto \xi( X \cdot X_0,  \overline{X} \cdot \hat{X}_0  ) \overline{\ket{ X \cdot X_0}}.
\eeq
However, by the cutting and gluing and reflection positivity, we get 
$ \xi( X \cdot X_0,  \overline{X} \cdot \hat{X}_0  )  =    \xi(X_0,  \hat{X}_0 ) $.
Therefore, the comparison of the two equations immediately give $F(\overline{X}) = \overline{F(X)}$. This confirms \eqref{eq:invo}.
We have already checked \eqref{eq:involutioncondition} and \eqref{eq:twiceinvo}.

Finally the evaluation $Z(e_Y)$ is just unity under the chain of identifications 
$\CH(Y) \otimes \overline{ \CH(Y) } \cong \CH(Y) \otimes \CH(\overline{Y}) \cong \CH(Y \sqcup \overline{Y}) \cong \BC$,
because $F(e_Y) \ket{c_{\overline{Y}}} = \ket{ e_Y \cdot c_{\overline{Y}}}= z(e_Y \cdot c_{\overline{Y}})=1$ by reflection positivity.
In other words, the hermitian metric defined by $e_Y$ coincides with the hermitian metric we have already introduced, which is positive definite.

We have confirmed all the nontrivial axioms of unitary invertible TQFT. It is clear from the construction that $F(X)=z(X)$ for every closed $H_d$-manifold $X$.
This completes the proof.
\end{proof}

\subsection{Isomorphism between two theories}
We would like to identify two theories which have the same partition functions on closed $H_d$-manifolds.
Such an identification may be expressed as a natural isomorphism $\eta : F \Rightarrow G$ between
two functors of TQFT.

It is important to have in mind the following physical intuition. 
An element of a Hilbert space $v \in \CH$ does not correspond to physical observables.
Only the ray of $v$ in $\CH$ matters. Still, relative phases can have physical consequences such as Berry phases.
However, if $Y$ and $Y'$ represent different elements of the bordism group $\Omega_{d-1}^H$,
there are no transitions from $Y$ to $Y'$. Therefore, we expect some phase ambiguity 
in the relative phase of $\CH(Y)$ and $\CH(Y')$. This ambiguity appears in the following discussion. 

\begin{thm} \label{thm:identification}
Let $F$ and $G$ be symmetric monoidal functors associated to two invertible TQFTs with the same cobordism invariant partition functions, 
and assume that $\Omega^{ H}_{d-1}$ is finitely generated. 
Then there exits a symmetric monoidal natural isomorphism $\eta$ from $F$ to $G$, meaning that $\eta$ is a natural isomorphism such that
$\eta_{\varnothing}$ regarded as a map $ \BC \to \BC$ is the identity, and
it preserves the symmetric monoidal structure as
\beq
\xymatrix{
F(Y) \otimes F(Y') \ar[d]_{}  \ar[r]^{\eta_Y \otimes \eta_{Y'}} & G(Y) \otimes G(Y')  \ar[d]^{} \\
F( Y \sqcup Y')  \ar[r]_{\eta_{Y \sqcup Y'}} &  G( Y \sqcup Y')  
}\label{eq:SMNT}
\eeq
In addition, the $\eta$ can be taken to preserve the involution as
\beq
\xymatrix{
F( \overline{Y} )  \ar[d]_{}  \ar[r]^{ \eta_{\overline{Y}} }  &  G( \overline{Y} )    \ar[d]^{} \\
\overline{F( Y)}   \ar[r]_{   \overline{\eta_Y}  } &  \overline{G( Y)} 
}\label{eq:INV}
\eeq
\end{thm}
\begin{proof}
To each closed $H_{d-1}$-manifold $Y$, we assign a linear map $\eta_Y: F(Y) \to G(Y)$ as follows.
First, suppose that there exits a bordism $X: \varnothing \to Y$. In this case, we can regard $F(X)$ and $G(X)$ as elements of $F(Y)$ and $G(Y)$, respectively.
Then we define a linear map $\eta_Y: F(Y) \to G(Y)$  as $\eta_Y: F(X) \mapsto G(X)$.
This map does not depend on the choice of $X$ as can be checked from the fact that the partition functions of the two theories on any closed $H_d$-manifolds are the same
by assumption. In particular, $\eta_\varnothing: \BC \to \BC$ is just the identity map.

Next, consider elementary reference manifolds $Y^{\rm ref}_b$ for $ k+1 \leq b \leq k+\ell$ which are the generators of the torsion part of $\Omega^{ H}_{d-1}$.
The disjoint union of $p_b$ copies of them,
$Y^{\rm ref}_b \sqcup \cdots \sqcup Y^{\rm ref}_b$, is bordant to $\varnothing$. Thus we already have a map 
$F(Y^{\rm ref}_b)^{\otimes p_b} \to G(Y^{\rm ref}_b)^{\otimes p_b} $
which was defined above. Then we pick up a map $\eta_{Y^{\rm ref}_b}: F(Y^{\rm ref}_b) \to G(Y^{\rm ref}_b)$ such that 
$\eta_{Y^{\rm ref}_b}^{\otimes p_b}$ coincides with 
the canonical map. There are $p_b$ choices of such a map, and we just pick up one of them.
On the other hand, for elementary reference manifolds $Y^{\rm ref}_a$ with $1 \leq a \leq k$ which are the generators of the freely generated part of $\Omega^{ H}_{d-1}$,
we assign a completely arbitrary map $\eta_{Y^{\rm ref}_a}: F(Y^{\rm ref}_a) \to G(Y^{\rm ref}_a )$ which preserves the hermitian metric.

For any $w \in \Omega^{ H}_{d-1}$, $Y^{\rm ref}_w$ is given as the disjoint union of $Y^{\rm ref}_a, \overline{Y^{\rm ref}_a}$ ($1 \leq a \leq k$)
and $Y^{\rm ref}_b$ ($ k+1 \leq b \leq k+\ell$) in the standard form described around \eqref{eq:ordering}. Then
the map $\eta_{Y^{\rm ref}_w} : F(Y^{\rm ref}_w ) \to G(Y^{\rm ref}_w)$ is defined by reducing $F(Y^{\rm ref}_w)$ (resp. $G(Y^{\rm ref}_w)$)
to the tensor product of $F(Y^{\rm ref}_a)$, $\overline{F(Y^{\rm ref}_a)}$ and $F(Y^{\rm ref}_b)$ 
(reps. $G(Y^{\rm ref}_a)$, $\overline{G(Y^{\rm ref}_a)}$ and $G(Y^{\rm ref}_b)$) 
and then extending the above maps to these tensor products.

Now, consider an arbitrary $Y$. Suppose is $Y$ is bordant to $Y^{\rm ref}_w$ and $X: Y^{\rm ref}_w \to Y$ is a bordism.
Then, we define $\eta_Y: F(Y) \to G(Y)$ as
\beq
\eta_Y = G(X) \eta_{Y^{\rm ref}_w} F(X)^{-1}.
\eeq
The inverse of $F(X)$ exists because it it a unitary operator in invertible TQFT, $F(X)^{-1} =F(\overline{X})$.
This definition is independent of $X$ by the equality of the partition functions of $F$ and $G$.
Also notice that this definition agrees with the one given above when $Y^{\rm ref}_w$ is empty.
By definition, we have a commutative diagram
\beq
\xymatrix{
F(Y^{\rm ref}_w) \ar[d]_{F(X)}  \ar[r]^{\eta_{Y^{\rm ref}_w} } & G(Y^{\rm ref}_w)  \ar[d]^{G(X)} \\
F(Y) \ar[r]_{\eta_Y} & G(Y)
}\label{eq:commWRTref}
\eeq
where all the arrows are isomorphisms.

We have defined $\eta_Y : F(Y) \to G(Y)$ for all $Y$. From the commutativity of the diagram \eqref{eq:commWRTref}
and the fact that all the arrows are isomorphisms, it is easy to see that $\eta$ is a natural isomorphism.

To check the diagrams \eqref{eq:SMNT} and \eqref{eq:INV},
let us pick up unit vectors $f_b \in F(Y^{\rm ref}_b)$, $g_b \in G(Y^{\rm ref}_b)$ ($k+1 \leq b \leq k+\ell$) such that 
$(f_b)^{\otimes p_b} = F(X^{\rm ref}_b )$, $(g_b)^{\otimes p_b} = G(X^{\rm ref}_b)$ and $\eta_{Y^{\rm ref}_b}: f_b \mapsto g_b$.
Also pick up unit vectors $f_a \in F(Y^{\rm ref}_a)$ and $g_a \in G(Y^{\rm ref}_a)$ ($1 \leq a \leq k$) 
such that $\eta_{Y^{\rm ref}_a} : f_a \mapsto g_a$. Then, explicitly we have
\beq
\eta_{Y^{\rm ref}_w} =  g_1^{\otimes w_1}\otimes \cdots \otimes g_{k+\ell}^{\otimes w_{k+\ell}} 
\otimes  (f_1^{\otimes w_1}\otimes  \cdots \otimes  f_{k+\ell}^{\otimes w_{k+\ell}})^{-1} \label{eq:explicitETA}
\eeq
where $(w_1,\cdots, w_{k+\ell})$ are integers such that it represents the element $w$ in 
$\Omega^{H}_{d-1}  \cong (\BZ)^k  \oplus  \BZ_{p_{k+1}} \oplus  \cdots \oplus \BZ_{p_{k+\ell}} $ and $w_b$~($k+1 \leq b \leq k+\ell$)
are constrained as $0 \leq w_b \leq p_b-1$. In the above notation, we have used the fact that all the Hilbert spaces are one dimensional
and hence the inverse like $(f_a)^{-1}$ makes sense.

Let us check the diagram \eqref{eq:SMNT} for the case $Y=Y^{\rm ref}_{w} $ and $Y'= Y^{\rm ref}_{w'}$. 
Let $X_{w,w'}(\tau)$ be the bordism from $Y^{\rm ref}_{w+w'}$ to $Y^{\rm ref}_{w} \sqcup Y^{\rm ref}_{w'}$ 
which consists of ${c}_{\overline{Y^{\rm ref}_a}}$, $X^{\rm ref}_b $ and permutation $\tau$ of elementary reference manifolds. 
This is the same bordism as used in the proof of Theorem~\ref{thm:construction}.
Let $\delta w_b$ be the number of times $X^{\rm ref}_b$ appears in $X_{w,w'}(\tau)$,
and let $\delta w_a$ be the number of times ${c}_{\overline{Y^{\rm ref}_a}}$ appears.
Under the identification $F( Y \sqcup Y') \cong F(Y) \otimes F(Y')$ 
we have
\beq
F(X_{w,w'}(\tau)) = F(\tau) \circ \left(1_{F(Y^{\rm ref}_{w+w'})} \otimes \prod_a( f_a \otimes  \overline{f_a})^{\otimes  \delta w_a} 
\otimes \prod_b f_b^{\otimes  p_b \delta w_b} \right),
\eeq
where $\circ$ means composition of linear maps.
We get a similar formula for $G(X_{w,w'}(\tau)) $ under the identification $G( Y \sqcup Y') \cong G(Y) \otimes G(Y')$. Then
\beq
\eta_{Y^{\rm ref}_{w} \sqcup Y^{\rm ref}_{w'}} &=G( X_{w,w'} ) \circ \eta_{Y^{\rm}_{w+w'} } \circ F( X_{w,w'})^{-1} \nonumber \\
&=  (G(\tau) \circ \CG) \circ (F(\tau) \circ \CF)^{-1}  , \label{eq:etacup}
\eeq
where
\beq
\CF &=f_1^{\otimes w''_1}\otimes  \cdots \otimes  f_{k+\ell}^{\otimes w''_{k+\ell}}  \otimes \prod_a( f_a \otimes  \overline{f_a})^{\otimes  \delta w_a} 
\otimes \prod_b f_b^{\otimes  p_b \delta w_b},  \\
\CG &=g_1^{\otimes w''_1}\otimes  \cdots \otimes  g_{k+\ell}^{\otimes w''_{k+\ell}}  \otimes \prod_a( g_a \otimes  \overline{g_a})^{\otimes  \delta w_a} 
\otimes \prod_b g_b^{\otimes  p_b \delta w_b} .
\eeq
Here $w''_a$ are defined such that they represent the same element of $\Omega^H_{d-1}$ as $w_a+w'_a$.
By construction,
\eqref{eq:etacup} contains the same number of $f_a$, $g_a$, etc, as in $\eta_{Y^{\rm ref}_{w}} \otimes  \eta_{Y^{\rm ref}_{w'}}$.
Also, they are ordered by $F(\tau)$ and $G(\tau)$ in the same way as in $\eta_{Y^{\rm ref}_{w}} \otimes  \eta_{Y^{\rm ref}_{w'}}$.
Any possible sign factors $(\pm )$ are cancelled between the $F(\tau)$-part and the $G(\tau)$-part.
Therefore, we get
\beq
\eta_{Y^{\rm ref}_{w} \sqcup Y^{\rm ref}_{w'}} = \eta_{Y^{\rm ref}_{w}} \otimes  \eta_{Y^{\rm ref}_{w'}},
\eeq
under $F(Y) \otimes F(Y') \cong F( Y \sqcup Y')$  and $G(Y) \otimes G(Y') \cong G( Y \sqcup Y')$.

For general $Y$ and $Y'$, take bordisms $X: Y^{\rm ref}_{w} \to Y$ and $X' : Y^{\rm ref}_{w'} \to Y'$.
Then we have the bordism $(X \sqcup X') : Y^{\rm ref}_{w} \sqcup Y^{\rm ref}_{w'} \to Y \sqcup Y' $.
Then
\beq
\eta_{Y \sqcup Y'} &= G(X \sqcup X') \circ \eta_{Y^{\rm ref}_{w} \sqcup Y^{\rm ref}_{w'}} \circ  F(X \sqcup X')^{-1} \nonumber \\
&=(G(X) \otimes G(X')) \circ (\eta_{Y^{\rm ref}_{w}} \otimes \eta_{Y^{\rm ref}_{w'}}) \circ  (F(X) \otimes F(X') )^{-1} \nonumber \\
&=\eta_Y \otimes \eta_{Y'}.
\eeq
This shows the commutativity of the diagram \eqref{eq:SMNT}.

The diagram \eqref{eq:INV} is also checked in a similar way.
We take a bordism $X_{-w}(\tau): Y^{\rm ref}_{-w} \to \overline{Y^{\rm ref}_{w}}$ which consists of $\overline{X^{\rm ref}_b}$, the evaluations $e_{Y^{\rm ref}_b}$,
as well as permutation $\tau$. Then, under the identification $F(\overline{Y}) \xrightarrow{\sim} \overline{F(Y)} $ and similarly for $G$,
we get
\beq
\eta_{\overline{Y^{\rm ref}_{w}}} = G(X_{-w}(\tau)) \circ \eta_{ Y^{\rm ref}_{-w}  } \circ F(X_{-w}(\tau))^{-1} = \overline{ \eta_{ Y^{\rm ref}_{w}  } }
\eeq
because $\eta_{\overline{Y^{\rm ref}_{w}}}$ and $\overline{ \eta_{ Y^{\rm ref}_{w}  } }$ contains the same number of $f_a$, $g_a$ etc. in the same order, and
any possible sign factors under permutation are cancelled between the $F$-part and the $G$-part.
By using this, the commutativity of \eqref{eq:INV} is checked. This completes the proof.
\end{proof}
\begin{rem}
We call two TQFTs $F$ and $G$ isomorphic if they have a natural transformation $\eta$ of the kind described above.
The above theorem means that two TQFTs are isomorphic if their partition functions on any closed $H_d$-manifolds are the same.
\end{rem}
\begin{rem}
In the context of SPT phases (but not the generalized theta angles), we are concerned with deformation classes of theories rather than isomorphism classes.
Namely, two SPT phases which can be continuously connected to each other are identified.
In our context, this means the following. 
Two theories having the cobordism invariant partition functions $z: \Omega_d^H \to \U(1)$ and $z': \Omega_d^H \to \U(1)$
are identified if $z$ and $z'$ are continuously connected in the obvious sense.
Such an identification can be easily done at the level of TQFTs because we have a TQFT for any $z$ 
which is unique up to isomorphims by Theorems~\ref{thm:construction} and \ref{thm:identification},
and hence the continuous deformation of cobordism invariant $z: \Omega_d^H \to \U(1)$ can be uplifted to the level of TQFTs.
\end{rem}
\begin{rem}
From the above proof, it is clear that the ambiguity of the natural transformation between two TQFTs is classified by $\Omega^{H}_{d-1}$.
More precisely, ${\rm Hom}(\Omega^{H}_{d-1} , \U(1))$ classifies the phase ambiguity of the Hilbert spaces.
This is natural from the following point of view. First, note that by ``Wick rotations",
constant time slices and spatial boundaries may be related to each other. Then, if we put a $d$-dimensional theory on 
a manifold with spatial boundary, the above ambiguity means that
the partition function of the boundary $d-1$-dimensional theory has the ambiguity classified by ${\rm Hom}(\Omega^{H}_{d-1} , \U(1))$.
This is precisely achieved by multiplying the boundary $d-1$-dimensional theory by 
$d-1$-dimensional invertible TQFTs classified by ${\rm Hom}(\Omega^{H}_{d-1} , \U(1))$.
\end{rem}

%%%%%%%%%%%%%%%%%%%%%%%%%%%%%%%%%%%%%%%%%%%%%%%%%%%%%%%%%%%
\acknowledgments
The author would like to thank Y.~Tachikawa and E.~Witten for helpful comments,
and K.~Hori, C.-T.~Hsieh, and Y.~Tachikawa for discussions on related topics.
The work of KY is supported in part by the WPI Research Center Initiative (MEXT, Japan),
and also supported by JSPS KAKENHI Grant-in-Aid (Wakate-B), No.17K14265.

\appendix

%%%%%%%%%%%%%%%%%%%%%%%%%%%%%%%%%%%%%%%%%%%%%%%%%%%%%%%%%%%%%%%%%%%%%%%%%%%%%%%%%%%%%%%%%%%%%%%%%%%%%%%%%%%%%%%%%%%%%%%%%%%%%%%%%%%%%%%%%%%%%%%%%%%%%%%%%%%%%%%%%%%%%%%%%%%%%%%%%%%%%%
\section{Some categorical notions}\label{sec:app}
%%%%%%%%%%%%%%%%%%%%%%%%%%%%%%%%%%%%%%%%%%%%%%%%%%%%%%%%%%%%%%%%%%%%%%%%%%%%%%%%%%%%%%%%%%%%%%%%%%%%%%%%%%%%%%%%%%%%%%%%%%%%%%%%%%%%%%%%%%%%%%%%%%%%%%%%%%%%%%%%%%%%%%%%%%%%%%%%%%%%%%
%\subsection{Symmetric monoidal category, functor, and natural transformation}
For completeness, here we reproduce the definitions of symmetric monoidal categories, functors and natural transformations summarized in \cite{Baez}.
We denote categories by $\sC,\sD,\cdots $, functories by $F,G,\cdots$, and natural transformations by $\eta,\cdots$. 
The definitions of ordinary categories, functors and natural transformations are explained very briefly in Sec.~\ref{sec:Atiyah}.

\paragraph{Symmetric monoidal category.} First we define symmetric monoidal category.
\begin{defi}
A monoidal category is a category equipped with 
\begin{itemize}
\item a functor $\otimes: \sC \times \sC \to \sC$ called the tensor product, 
\item an object $1 \in {\rm obj}(\sC)$ called the unit object,
\item a natural isomorphism $a_{x,y,z}$ ($x,y,z \in {\rm obj}(\sC)$) called the associator
\beq
a_{x,y,z}: (x \otimes y) \otimes z \to x \otimes (y \otimes z)
\eeq
satisfying the pentagon equation
\beq
\xymatrix{
(( w \otimes x) \otimes y) \otimes z  \ar[rr]^{a_{w \otimes x, y, z}}  \ar[d]_{a_{w, x, y} \otimes 1_z }&    & ( w \otimes x) \otimes (y \otimes z)  \ar[d]^{a_{w,x,y \otimes z}}    \\
( w \otimes (x \otimes y)) \otimes z ~ \ar[r]_{a_{w, x \otimes y, z}} & ~~ w \otimes  (  (x \otimes y) \otimes z) ~~ \ar[r]_{~~1_w \otimes a_{x,y,z}~~} &~ w \otimes ( x \otimes ( y \otimes z ))  
}
\eeq
\item natural isomorphisms $\ell_x$ and $r_x$
called the left and right unit laws, 
\beq
\ell_x : 1 \otimes x \to x, \qquad  r_x: x \otimes 1 \to x
\eeq
satisfying the triangle equations
\beq
\xymatrix{
(x \otimes 1) \otimes y \ar[dr]_{r_x \otimes 1_y} \ar[rr]^{a_{x,1,y}}&          & x \otimes (1 \otimes y) \ar[dl]^{1_x \otimes \ell_y} \\
&  x \otimes y  &
}
\eeq
\end{itemize}
\end{defi}
Roughly speaking, the pentagon equation means that ``multiplications can be done in any order", or ``any ways to go from $(( w \otimes x) \otimes y) \otimes z$ to $w \otimes ( x \otimes ( y \otimes z )) $ are the same".
The triangle equation means that ``any ways to eliminate the unit $1$ are the same".

\begin{defi}
A braided monoidal category is a monoidal category with a natural isomorphism $b_{x,y}$ called the braiding,
\beq
b_{x,y}: x \otimes y \to y \otimes x
\eeq
satisfying the hexagon equations
\beq
\xymatrix{
(x \otimes y) \otimes z \ar[r]^{a_{x,y,z}} \ar[d]_{b_{x, y} \otimes 1_z} & x \otimes (y \otimes z) \ar[r] ^{ b_{x, y \otimes z}}&  (y \otimes z) \otimes x \ar[d]^{a_{y,z,x} } \\
(y \otimes x) \otimes z\ar[r]_{ a_{y,x,z} } & y \otimes (x \otimes z) \ar[r]_{1_y \otimes b_{x,z}} & y \otimes (z \otimes x)
}
\eeq
\beq
\xymatrix{
x \otimes (y \otimes z) \ar[r]^{a^{-1}_{x,y,z}} \ar[d]_{1_x \otimes b_{y, z} } & (x \otimes y) \otimes z \ar[r] ^{ b_{x \otimes y, z}}&  z \otimes (x \otimes y)  \ar[d]^{a^{-1}_{z,x,y} } \\
x \otimes (z \otimes y)  \ar[r]_{ a^{-1}_{x,z,y} } & (x \otimes  z) \otimes y \ar[r]_{ b_{x,z} \otimes 1_y} &  (z \otimes x) \otimes y
}
\eeq
\end{defi}
Roughly speaking, the first hexagon equation above means that ``moving $x$ all at once from the left to the right of $y \otimes z$ is the same as moving $x$ step by step
by first going through $y$ and then $z$." The second hexagon equation means a similar thing for $z$.

\begin{defi}
A symmetric monoidal category is a braided monoidal category such that the braiding satisfies $b_{y,x} b_{x,y} = 1_{x \otimes y}$.
\end{defi}

\paragraph{Symmetric monoidal functor.} 
Let us next consider functors between monoidal categories. In the following, if an expression like e.g. $a^\sD_{x,y,z}$ appears with a superscript or subscript $\sD$,
that means (in this particular case) ``the associator in the category $\sD$". The same remark applies to subscripts/superscripts of other quantities. 

\begin{defi}
A monoidal functor $F$ between monoidal categories $\sC$ and $\sD$ is a functor 
with 
\begin{itemize}
\item a natural transformation
\beq
\mu_{x,y}:   F(x) \otimes F(y) \to F(x \otimes y)
\eeq
satisfying the associativity
\beq
\xymatrix{
(F(x)\otimes F(y)) \otimes F(z) ~~ \ar[r]^{~~~~\mu_{x,y} \otimes 1_{F(z)}} \ar[d]_{a^\sD_{F(x),F(y),F(z)}}  &~~ F(x \otimes y) \otimes F(z) \ar[r]^{~~\mu_{x \otimes y, z}} & F( (x \otimes y) \otimes z) \ar[d]^{F(a^\sC_{x,y,z}) } \\
F(x) \otimes (F(y) \otimes F(z))~~ \ar[r]_{~~~~1_{F(x)} \otimes \mu_{y,z}} & F(x) \otimes F( y \otimes z) \ar[r]_{~\mu_{x, y \otimes z}} & F(x \otimes (y \otimes z)) 
}
\eeq
\item an isomorphism
\beq
\epsilon : 1_{\sD} \to F(1_\sC)
\eeq
satisfying
\beq
\xymatrix{ 
1_\sD \otimes F(x)  \ar[r]^{\epsilon \otimes 1_{F(x)}} \ar[d]_{\ell^\sD_{F(x)}} & ~F(1_\sC) \otimes F(x) \ar[d]^{\mu_{1_\sC,x}} \\
F(x)  &  F(1_\sC \otimes x) \ar[l]_{F(\ell^\sC_x)}
}
\eeq
and
\beq
\xymatrix{ 
  F(x) \otimes 1_\sD  \ar[r]^{1_{F(x)} \otimes \epsilon ~~~} \ar[d]_{r^\sD_{F(x)}} & ~ F(x) \otimes F(1_\sC)  \ar[d]^{\mu_{x,1_\sC}} \\
F(x)  &  F( x \otimes 1_\sC ) \ar[l]_{F(r^\sC_x)}
}
\eeq
\end{itemize}
\end{defi}
Roughly speaking, these equations mean that ``the associator $a_{x,y,z}$ and the left, right unit laws $\ell_x$, $r_x$ can be used before or after the application of the functor, giving the same result".

\begin{defi}
A braided monoidal functor between braided monoidal categories is a monoidal functor with the additional condition that
\beq
\xymatrix{
F(x) \otimes F(y) \ar[r]^{~\mu_{x,y}} \ar[d]_{b^\sD_{F(x),F(y)}} & F(x \otimes y) \ar[d]^{F(b^\sC_{x,y} ) } \\
F(y) \otimes F(x) \ar[r]_{~~\mu_{y,x} } & F(y \otimes x) 
}
\eeq
\end{defi}
Again, this roughly means that ``the braiding can be used before or after the functor".
\begin{defi}
A symmetric monoidal functor between symmetric monoidal categories is a braided monoidal functor without extra conditions.
\end{defi}

\paragraph{Symmetric monoidal natural transformation.}Finally, we describe natural transformations.
\begin{defi}
A monoidal natural transformation $\eta$ between monoidal functors $F$ and $G$ is a natural transformation such that the following diagram commutes:
\beq
\xymatrix{
F(x) \otimes F(y) \ar[r]^{\eta_x \otimes \eta_y} \ar[d]_{\mu^F_{x,y}} & G(x) \otimes G(y) \ar[d]^{\mu^G_{x,y}} \\
F(x \otimes y) \ar[r]_{\eta_{x \otimes y} } & G(x \otimes y) 
}
\eeq
\beq
\xymatrix{
& 1_\sD \ar[dl]_{\epsilon_F} \ar[dr]^{\epsilon_G} \\
 F(1_\sC) \ar[rr]_{\eta_{1_\sC}}& & G(1_\sC)
}
\eeq
\end{defi}
\begin{defi}
A braided (reps. symmetric) monoidal natural transformation between braided (reps. symmetric) monoidal functors is a monoidal natural transformation, without extra conditions.
\end{defi}

\paragraph{Involution.}
For completeness, we also describe the notion of involution following \cite{Freed:2016rqq}.
\begin{defi}\label{defi:involution}
An involution on a category $\sC$ is a pair $(\beta, \xi)$ of a functor $\beta: \sC \to \sC$ and a natural isomorphism $\xi: {\rm id}_\sC \to \beta^2$
such that $\beta( \xi_x) = \xi_{\beta(x)}$ as morphisms $\beta(x) \to \beta^3(x)$.
\end{defi}
Roughly speaking, the involution functor $\beta$ ``squares to the identity functor".
\begin{defi}\label{defi:equivariant}
Let $\beta_\sC$ and $\beta_\sD$ be involutions of categories $\sC$ and $\sD$. A functor $F : \sC \to \sD$ is equivariant under the involution pair $(\beta_\sC, \beta_\sD)$ if 
there is a natural isomorphism $\phi: F  \beta_\sC \Rightarrow \beta_\sD  F$ such that the following diagram commutes:
\beq
\xymatrix{
F(x) \ar[rdd]_{\xi^\sD_{F(x)} }  \ar[r]^{F(\xi^\sC_x)~~} &~ F  \beta_\sC^2(x)  \ar[d]^{\phi_{\beta_\sC(x)} } \\
& ~~\beta_\sD F \beta_\sC (x)  \ar[d]^{\beta_\sD(\phi_x) } \\
&  \beta_\sD^2 F(x)
}
\eeq
\end{defi}
Roughly speaking, this means that ``the involution $\beta$ and the functor $F$ commutes with each other in the way consistent with the fact that the involution squires to the identity".

The corresponding notions in symmetric monoidal categories can also be defined by using
symmetric monoidal functors and symmetric monoidal natural isomorphisms.

%%%%%%%%%%%%%%%%%%%%%%%%%%%%%%%%%%%%%%%%%%%%%%%%%%%%%%%%%%%

%%%%%%%%%%%%%%%%%%%%%%%%%%%%%%%%%%%%%%%%%%%%%%%%%%%%%%%%%%%

\bibliographystyle{JHEP}
\bibliography{ref}

\providecommand{\href}[2]{#2}\begingroup\raggedright\begin{thebibliography}{10}

\bibitem{Kapustin:2014tfa}
A.~Kapustin, {\it {Symmetry Protected Topological Phases, Anomalies, and
  Cobordisms: Beyond Group Cohomology}},
  \href{http://arxiv.org/abs/1403.1467}{{\tt arXiv:1403.1467}}.

\bibitem{Kapustin:2014dxa}
A.~Kapustin, R.~Thorngren, A.~Turzillo, and Z.~Wang, {\it {Fermionic Symmetry
  Protected Topological Phases and Cobordisms}},  {\em JHEP} {\bf 12} (2015)
  052, [\href{http://arxiv.org/abs/1406.7329}{{\tt arXiv:1406.7329}}].
  [JHEP12,052(2015)].

\bibitem{Freed:2016rqq}
D.~S. Freed and M.~J. Hopkins, {\it {Reflection positivity and invertible
  topological phases}},  \href{http://arxiv.org/abs/1604.06527}{{\tt
  arXiv:1604.06527}}.

\bibitem{Hasan:2010xy}
M.~Z. Hasan and C.~L. Kane, {\it {Topological Insulators}},  {\em Rev. Mod.
  Phys.} {\bf 82} (2010) 3045, [\href{http://arxiv.org/abs/1002.3895}{{\tt
  arXiv:1002.3895}}].

\bibitem{Qi:2011zya}
X.~L. Qi and S.~C. Zhang, {\it {Topological insulators and superconductors}},
  {\em Rev. Mod. Phys.} {\bf 83} (2011), no.~4 1057--1110.

\bibitem{Freed:2004yc}
D.~S. Freed and G.~W. Moore, {\it {Setting the quantum integrand of M-theory}},
   {\em Commun. Math. Phys.} {\bf 263} (2006) 89--132,
  [\href{http://arxiv.org/abs/hep-th/0409135}{{\tt hep-th/0409135}}].

\bibitem{Ryu:2010ah}
S.~Ryu, J.~E. Moore, and A.~W.~W. Ludwig, {\it {Electromagnetic and
  gravitational responses and anomalies in topological insulators and
  superconductors}},  {\em Phys. Rev.} {\bf B85} (2012) 045104,
  [\href{http://arxiv.org/abs/1010.0936}{{\tt arXiv:1010.0936}}].

\bibitem{Wen:2013oza}
X.-G. Wen, {\it {Classifying gauge anomalies through symmetry-protected trivial
  orders and classifying gravitational anomalies through topological orders}},
  {\em Phys. Rev.} {\bf D88} (2013), no.~4 045013,
  [\href{http://arxiv.org/abs/1303.1803}{{\tt arXiv:1303.1803}}].

\bibitem{Kapustin:2014zva}
A.~Kapustin and R.~Thorngren, {\it {Anomalies of discrete symmetries in various
  dimensions and group cohomology}},
  \href{http://arxiv.org/abs/1404.3230}{{\tt arXiv:1404.3230}}.

\bibitem{Freed:2014iua}
D.~S. Freed, {\it {Anomalies and Invertible Field Theories}},  {\em Proc. Symp.
  Pure Math.} {\bf 88} (2014) 25--46,
  [\href{http://arxiv.org/abs/1404.7224}{{\tt arXiv:1404.7224}}].

\bibitem{Wang:2014pma}
J.~C. Wang, Z.-C. Gu, and X.-G. Wen, {\it {Field theory representation of
  gauge-gravity symmetry-protected topological invariants, group cohomology and
  beyond}},  {\em Phys. Rev. Lett.} {\bf 114} (2015), no.~3 031601,
  [\href{http://arxiv.org/abs/1405.7689}{{\tt arXiv:1405.7689}}].

\bibitem{Hsieh:2015xaa}
C.-T. Hsieh, G.~Y. Cho, and S.~Ryu, {\it {Global anomalies on the surface of
  fermionic symmetry-protected topological phases in (3+1) dimensions}},  {\em
  Phys. Rev.} {\bf B93} (2016), no.~7 075135,
  [\href{http://arxiv.org/abs/1503.01411}{{\tt arXiv:1503.01411}}].

\bibitem{Witten:2015aba}
E.~Witten, {\it {Fermion Path Integrals And Topological Phases}},  {\em Rev.
  Mod. Phys.} {\bf 88} (2016), no.~3 035001,
  [\href{http://arxiv.org/abs/1508.04715}{{\tt arXiv:1508.04715}}].

\bibitem{Witten:2016cio}
E.~Witten, {\it {The "Parity" Anomaly On An Unorientable Manifold}},  {\em
  Phys. Rev.} {\bf B94} (2016), no.~19 195150,
  [\href{http://arxiv.org/abs/1605.02391}{{\tt arXiv:1605.02391}}].

\bibitem{Guo:2017xex}
M.~Guo, P.~Putrov, and J.~Wang, {\it {Time Reversal, SU(N) Yang-Mills and
  Cobordisms: Interacting Topological Superconductors/Insulators and Quantum
  Spin Liquids in 3+1D}},  \href{http://arxiv.org/abs/1711.11587}{{\tt
  arXiv:1711.11587}}.

\bibitem{Atiyah:1975jf}
M.~F. Atiyah, V.~K. Patodi, and I.~M. Singer, {\it {Spectral asymmetry and
  Riemannian Geometry 1}},  {\em Math. Proc. Cambridge Phil. Soc.} {\bf 77}
  (1975) 43.

\bibitem{Dai:1994kq}
X.-z. Dai and D.~S. Freed, {\it {eta invariants and determinant lines}},  {\em
  J. Math. Phys.} {\bf 35} (1994) 5155--5194,
  [\href{http://arxiv.org/abs/hep-th/9405012}{{\tt hep-th/9405012}}]. [Erratum:
  J. Math. Phys.42,2343(2001)].

\bibitem{Yonekura:2016wuc}
K.~Yonekura, {\it {Dai-Freed theorem and topological phases of matter}},  {\em
  JHEP} {\bf 09} (2016) 022, [\href{http://arxiv.org/abs/1607.01873}{{\tt
  arXiv:1607.01873}}].

\bibitem{Fukaya:2017tsq}
H.~Fukaya, T.~Onogi, and S.~Yamaguchi, {\it {Atiyah-Patodi-Singer index from
  the domain-wall fermion Dirac operator}},  {\em Phys. Rev.} {\bf D96} (2017),
  no.~12 125004, [\href{http://arxiv.org/abs/1710.03379}{{\tt
  arXiv:1710.03379}}].

\bibitem{Chen:2011pg}
X.~Chen, Z.-C. Gu, Z.-X. Liu, and X.-G. Wen, {\it {Symmetry protected
  topological orders and the group cohomology of their symmetry group}},  {\em
  Phys. Rev.} {\bf B87} (2013), no.~15 155114,
  [\href{http://arxiv.org/abs/1106.4772}{{\tt arXiv:1106.4772}}].

\bibitem{Gu:2012ib}
Z.-C. Gu and X.-G. Wen, {\it {Symmetry-protected topological orders for
  interacting fermions: Fermionic topological nonlinear ? models and a special
  group supercohomology theory}},  {\em Phys. Rev.} {\bf B90} (2014), no.~11
  115141, [\href{http://arxiv.org/abs/1201.2648}{{\tt arXiv:1201.2648}}].

\bibitem{Wang:2017moj}
Q.-R. Wang and Z.-C. Gu, {\it {Towards a complete classification of fermionic
  symmetry protected topological phases in 3D and a general group
  supercohomology theory}},  {\em Phys. Rev.} {\bf X8} (2018), no.~1 011055,
  [\href{http://arxiv.org/abs/1703.10937}{{\tt arXiv:1703.10937}}].

\bibitem{Kitaev:2013a}
A.~Kitaev, {\it {On the classification of short-range entangled states }}, .
  Talk at Simons Center
  \href{http://scgp.stonybrook.edu/archives/16180}{http://scgp.stonybrook.edu/archives/16180}.

\bibitem{Gaiotto:2017zba}
D.~Gaiotto and T.~Johnson-Freyd, {\it {Symmetry Protected Topological phases
  and Generalized Cohomology}},  \href{http://arxiv.org/abs/1712.07950}{{\tt
  arXiv:1712.07950}}.

\bibitem{Xiong:2016deb}
Z.~Xiong, {\it {Minimalist approach to the classification of symmetry protected
  topological phases}},  \href{http://arxiv.org/abs/1701.00004}{{\tt
  arXiv:1701.00004}}.

\bibitem{Freed:2014eja}
D.~S. Freed, {\it {Short-range entanglement and invertible field theories}},
  \href{http://arxiv.org/abs/1406.7278}{{\tt arXiv:1406.7278}}.

\bibitem{Freed:2017rlk}
D.~S. Freed, Z.~Komargodski, and N.~Seiberg, {\it {The Sum Over Topological
  Sectors and $\theta$ in the 2+1-Dimensional $\mathbb{C}\mathbb{P}^1$
  $\sigma$-Model}},  \href{http://arxiv.org/abs/1707.05448}{{\tt
  arXiv:1707.05448}}.

\bibitem{Baez:1995xq}
J.~C. Baez and J.~Dolan, {\it {Higher dimensional algebra and topological
  quantum field theory}},  {\em J. Math. Phys.} {\bf 36} (1995) 6073--6105,
  [\href{http://arxiv.org/abs/q-alg/9503002}{{\tt q-alg/9503002}}].

\bibitem{Lurie:2009keu}
J.~Lurie, {\it {On the Classification of Topological Field Theories}},
  \href{http://arxiv.org/abs/0905.0465}{{\tt arXiv:0905.0465}}.

\bibitem{Schommer-Pries:2017sdd}
C.~Schommer-Pries, {\it {Invertible Topological Field Theories}},
  \href{http://arxiv.org/abs/1712.08029}{{\tt arXiv:1712.08029}}.

\bibitem{Fidkowski:2013jua}
L.~Fidkowski, X.~Chen, and A.~Vishwanath, {\it {Non-Abelian Topological Order
  on the Surface of a 3D Topological Superconductor from an Exactly Solved
  Model}},  {\em Phys. Rev.} {\bf X3} (2013), no.~4 041016,
  [\href{http://arxiv.org/abs/1305.5851}{{\tt arXiv:1305.5851}}].

\bibitem{Wang:2014lca}
C.~Wang and T.~Senthil, {\it {Interacting fermionic topological
  insulators/superconductors in three dimensions}},  {\em Phys. Rev.} {\bf B89}
  (2014), no.~19 195124, [\href{http://arxiv.org/abs/1401.1142}{{\tt
  arXiv:1401.1142}}]. [Erratum: Phys. Rev.B91,no.23,239902(2015)].

\bibitem{Metlitski:2014xqa}
M.~A. Metlitski, L.~Fidkowski, X.~Chen, and A.~Vishwanath, {\it {Interaction
  effects on 3D topological superconductors: surface topological order from
  vortex condensation, the 16 fold way and fermionic Kramers doublets}},
  \href{http://arxiv.org/abs/1406.3032}{{\tt arXiv:1406.3032}}.

\bibitem{Morimoto:2015lua}
T.~Morimoto, A.~Furusaki, and C.~Mudry, {\it {Breakdown of the topological
  classification $\mathbb{Z}$ for gapped phases of noninteracting fermions by
  quartic interactions}},  {\em Phys. Rev.} {\bf B92} (2015), no.~12 125104,
  [\href{http://arxiv.org/abs/1505.06341}{{\tt arXiv:1505.06341}}].

\bibitem{Tachikawa:2016xvs}
Y.~Tachikawa and K.~Yonekura, {\it {Gauge interactions and topological phases
  of matter}},  {\em PTEP} {\bf 2016} (2016), no.~9 093B07,
  [\href{http://arxiv.org/abs/1604.06184}{{\tt arXiv:1604.06184}}].

\bibitem{Witten:2015aoa}
E.~Witten, {\it {Three Lectures On Topological Phases Of Matter}},  {\em Riv.
  Nuovo Cim.} {\bf 39} (2016), no.~7 313--370,
  [\href{http://arxiv.org/abs/1510.07698}{{\tt arXiv:1510.07698}}].

\bibitem{Gaiotto:2014kfa}
D.~Gaiotto, A.~Kapustin, N.~Seiberg, and B.~Willett, {\it {Generalized Global
  Symmetries}},  {\em JHEP} {\bf 02} (2015) 172,
  [\href{http://arxiv.org/abs/1412.5148}{{\tt arXiv:1412.5148}}].

\bibitem{Thorngren:2017vzn}
R.~Thorngren, {\it {Topological Terms and Phases of Sigma Models}},
  \href{http://arxiv.org/abs/1710.02545}{{\tt arXiv:1710.02545}}.

\bibitem{Kapustin:2013uxa}
A.~Kapustin and R.~Thorngren, {\it {Higher symmetry and gapped phases of gauge
  theories}},  \href{http://arxiv.org/abs/1309.4721}{{\tt arXiv:1309.4721}}.

\bibitem{Tachikawa:2017gyf}
Y.~Tachikawa, {\it {On gauging finite subgroups}},
  \href{http://arxiv.org/abs/1712.09542}{{\tt arXiv:1712.09542}}.

\bibitem{Cordova:2018cvg}
C.~C\'ordova, T.~T. Dumitrescu, and K.~Intriligator, {\it {Exploring 2-Group
  Global Symmetries}},  \href{http://arxiv.org/abs/1802.04790}{{\tt
  arXiv:1802.04790}}.

\bibitem{Benini:2018reh}
F.~Benini, C.~C\'ordova, and P.-S. Hsin, {\it {On 2-Group Global Symmetries and
  their Anomalies}},  \href{http://arxiv.org/abs/1803.09336}{{\tt
  arXiv:1803.09336}}.

\bibitem{Seiberg:2018ntt}
N.~Seiberg, Y.~Tachikawa, and K.~Yonekura, {\it {Anomalies of Duality Groups
  and Extended Conformal Manifolds}},
  \href{http://arxiv.org/abs/1803.07366}{{\tt arXiv:1803.07366}}.

\bibitem{Freed:2012hx}
D.~S. Freed, {\it {The cobordism hypothesis}},
  \href{http://arxiv.org/abs/1210.5100}{{\tt arXiv:1210.5100}}.

\bibitem{Atiyah:1989vu}
M.~Atiyah, {\it {Topological quantum field theories}},  {\em Inst. Hautes
  Etudes Sci. Publ. Math.} {\bf 68} (1989) 175--186.

\bibitem{Baez}
J.~C. Baez, {\it {Some Definitions Everyone Should Know}}, .
  \href{http://math.ucr.edu/home/baez/qg-fall2004/definitions.pdf}{http://math.ucr.edu/home/baez/qg-fall2004/definitions.pdf}.

\bibitem{Milnor}
J.~Milnor, {\it {Lectures on the h-cobordism theorem }},  {\em Princeton
  University Press} (1965).

\end{thebibliography}\endgroup

%%%%%%%%%%%%%%%%%%%%%%%%%%%%%%%%%%%%%%%%%%%%%%%%%%%%%%%%%%%
%%%%%%%%%%%%%%%%%%%%%%%%%%%%%%%%%%%%%%%%%%%%%%%%%%%%%%%%%%%
\end{document}